\documentclass[11pt, reqno]{amsart}

\usepackage{graphics, stackrel}
\usepackage{amsmath, amssymb, amsthm}
\numberwithin{equation}{section}
\usepackage{graphicx}
\usepackage{mathtools}
\usepackage{verbatim}
\usepackage{amsfonts}
\usepackage{natbib}
\usepackage{enumitem}

\usepackage{caption}
\usepackage{subcaption}
\usepackage{placeins}
\usepackage{flafter}
\usepackage{tikz}
\usetikzlibrary{arrows.meta,positioning,fit,calc,backgrounds}

\usepackage{fancyvrb}
\usepackage[dvipsnames]{xcolor}
\usepackage{mdwlist}

\usepackage{enumitem}
\setlist[enumerate]{itemsep=2pt,topsep=3pt}
\setlist[itemize]{itemsep=2pt,topsep=3pt}
\setlist[enumerate,1]{label={\upshape (\roman*)}}

\usepackage{mathrsfs}  
\usepackage{bbm}
\usepackage{bm}        

\usepackage[left=1in, right=1in, top=0.8in, bottom=0.8in, includehead, includefoot]{geometry}
\renewcommand{\leq}{\leqslant}
\renewcommand{\geq}{\geqslant}

\usepackage[ruled, linesnumbered]{algorithm2e}

\usepackage{setspace}
\usepackage{footmisc}

\usepackage[citecolor=Brown, colorlinks=true, linkcolor=blue]{hyperref}

\DeclareMathOperator{\ran}{ran}
\DeclareMathOperator{\diag}{diag}

\newcommand{\sss}{\scriptscriptstyle}

\newcommand{\1}{\mathbf 1}
\newcommand{\id}{\operatorname{id}}
\newcommand{\supp}{\operatorname{supp}}

\newcommand{\tr}{\operatorname{tr}}

\newcommand{\Gs}{\scriptscriptstyle{G}}
\newcommand{\Xs}{\scriptscriptstyle{X}}

\newcommand{\Pis}{\scriptscriptstyle{\Pi}}

\newcommand*\diff{\mathop{}\!\mathrm{d}}
\newcommand*\e{\mathrm{e}}

\newcommand{\cC}{\mathscr C}
\newcommand{\dD}{\mathscr D}
\newcommand{\kK}{\mathscr K}
\newcommand{\kKs}{\!\scriptscriptstyle\mathscr K}
\newcommand{\sS}{\mathscr S}

\newcommand{\fF}{\mathscr F}

\newcommand{\RR}{\mathbb R}
\newcommand{\CC}{\mathbb C}
\newcommand{\QQ}{\mathbb Q}
\newcommand{\NN}{\mathbb N}

\newcommand{\PP}{\mathbb P}
\newcommand{\EE}{\mathbb E}

\newcommand{\Xsf}{\mathsf X}

\newcommand{\Zsf}{\mathsf Z}
\newcommand{\Ecal}{\mathcal E}
\newcommand{\Gcal}{\mathcal G}

\newcommand{\Ccal}{\mathcal C}

\newcommand{\s}{\scriptscriptstyle}
\newcommand{\Lx}{L_{\scriptscriptstyle X}}
\newcommand{\Lp}{L_{\scriptscriptstyle \Pi}}
\newcommand{\LpN}{L_{\scriptscriptstyle \mathcal N}}
\newcommand{\LpD}{L_{\scriptscriptstyle \mathcal D}}
\newcommand{\LpDN}{L_{\scriptscriptstyle \mathcal D\mathcal N}}
\newcommand{\lambdaN}{\lambda_{\scriptscriptstyle\mathcal N}}
\newcommand{\lambdaD}{\lambda_{\scriptscriptstyle\mathcal D}}
\newcommand{\phiN}{\varphi_{\scriptscriptstyle\mathcal N}}
\newcommand{\phiD}{\varphi_{\scriptscriptstyle\mathcal D}}
\newcommand{\PhiD}{\Phi_{\scriptscriptstyle\mathcal D}}
\newcommand{\psiN}{\psi_{\scriptscriptstyle\mathcal N}}
\newcommand{\psiD}{\psi_{\scriptscriptstyle\mathcal D}}
\newcommand{\PsiD}{\Psi_{\scriptscriptstyle\mathcal D}}
\newcommand{\gN}{g_{\scriptscriptstyle\mathcal N}}
\newcommand{\gD}{g_{\scriptscriptstyle\mathcal D}}
\renewcommand{\t}{{\scriptscriptstyle{\top}}}

\renewcommand{\phi}{\varphi}
\renewcommand{\epsilon}{\varepsilon}

\theoremstyle{plain}

\newtheorem{theorem}{Theorem}[section]
\newtheorem{assumption}[theorem]{Assumption}
\newtheorem{corollary}[theorem]{Corollary}
\newtheorem{lemma}[theorem]{Lemma}
\newtheorem{proposition}[theorem]{Proposition}

\theoremstyle{definition}
\newtheorem{definition}{Definition}

\tikzset{
  state/.style={circle,draw,minimum size=7mm,inner sep=1pt,font=\scriptsize},
  classbox/.style={draw,dashed,rounded corners=3pt,inner sep=5pt},
  classnode/.style={draw,rounded corners=3pt,minimum width=20mm,
    minimum height=9mm,align=center,font=\small},
  corridorclass/.style={classnode,very thick,fill=black!10},
  flow/.style={-{Latex[length=2mm]},thick},
  rareflow/.style={-{Latex[length=2mm]},thick,dashed},
  biflow/.style={{Latex[length=1.8mm]}-{Latex[length=1.8mm]},thick}
}

\usepackage[normalem]{ulem}

\begin{document}

\hypersetup{
  pdftitle={Rare States and Long-Run Pricing},
  pdfauthor={Ye Lu and John Stachurski},
  pdfkeywords={Long-Run asset pricing, pricing dominance, pricing corridors,
  reducible and nearly reducible Markov dynamics, continuous-time models, positive
  pricing semigroups, rare transitions, disaster recovery}
}

\begin{center}
  \Large
  Rare States and Long-Run Pricing
    \vspace{1em}

  \large
  Ye Lu\textsuperscript{a}
  and John Stachurski\textsuperscript{b} \par \bigskip

  \small
  \textsuperscript{a} University of Sydney\par
  \textsuperscript{b} National Graduate Institute for Policy Studies
  \medskip

  \normalsize
  September 2026
\end{center}
\bigskip

\normalsize {\sc abstract}: Macroeconomic crises are rare, yet they can have large
effects on asset prices. Because the probability of entering a crisis from normal
times is small relative to the probability of recovery, the dynamics are nearly
reducible: setting the crisis-entry probability to zero makes crisis states
unreachable from normal states. We develop a graph-theoretic and spectral analysis of
long-run pricing under reducible and nearly reducible dynamics. We show that, at the
reducible limit, physical and pricing dynamics induce the same classes of economic
states and accessibility relations, but rank their long-run importance differently:
by recurrence and by class-specific pricing rates. Moreover, a claim's price depends
only on its \emph{pricing corridor}---the classes of states lying on directed paths
from the initial state to states where the payoff is positive---and the highest
class-specific pricing rate within this corridor determines the claim's long-run
pricing rate. When the reducible-limit corridor excludes the globally
pricing-dominant class, we characterize how restoring rare transitions introduces a
contribution from that class with a weight that vanishes as the transitions become
rarer. This contribution becomes dominant only beyond a \emph{crossover
maturity}, which increases with the rarity of the connecting paths and decreases with
the dominant class's pricing-rate advantage. Finally, a consumption-based
disaster-recovery application shows how preferences, consumption risk,
cash-flow exposure, and recovery timing jointly shape the long-run pricing dominance
of disaster states and crossover maturities.

\smallskip

{\sc keywords}: 
Long-run asset pricing; Reducible Markov chains; Rare transitions; Pricing
corridors; Positive pricing semigroups; Crossover horizons; Rare disasters.

\newpage

\section{Introduction}
\label{sec:intro}

Severe macroeconomic disasters, such as those associated with wars,
depressions, sovereign-debt crises, and banking or financial crises, are rare, but
the resulting contractions can last for years before recovery begins.%
\footnote{
  Using twentieth-century international data,
  \citet{barro2006rare} estimates an annual disaster probability of 1.5--2\% and
  documents peak-to-trough declines in per capita GDP of 15--64\%. Using a 10\%
  threshold and data extending back to 1870, \citet{barro2008macroeconomic}
  estimate an average peak-to-trough duration of about 3.5 years. In an estimated
  model of consumption disasters, \citet{nakamura2013crises} find that the
  average disaster reaches its trough after about 6 years and that roughly half
  of the decline is subsequently reversed.
}
These episodes also coincide with marked changes in cash-flow prospects,
uncertainty, and risk compensation, creating pricing conditions distinct from
those prevailing in normal times.%
\footnote{
  \citet{nakamura2013crises} find sharply elevated uncertainty
  about future consumption growth during consumption disasters, with important
  consequences for asset prices. \citet{muir2017financial} finds that risk premia
  rise substantially during financial crises, while
  \citet{krishnamurthy2025credit} document a sharp repricing of credit at crisis
  onset.
}
The pricing importance of adverse states therefore cannot be inferred from their
physical frequency alone. It depends jointly on how rarely they are entered, how long
they persist once reached, and how they affect cash flows and discounting, as well as
on the payoff horizon.

More generally, economic dynamics may feature regions of the state space associated
with distinct economic conditions for which entry intensities are much lower than
exit intensities.
For an economy initially
outside such a region, ignoring rare entry may have little effect on the pricing of
near-term payoffs, yet lead to substantial pricing errors for more distant payoffs.
Two questions naturally arise. Under what conditions can a region that the economy
rarely enters and eventually leaves nevertheless govern how a claim's price grows or
declines as its payoff horizon increases? When this occurs, what determines the
horizon at which the region's influence becomes quantitatively important, and how
does that horizon change as entry becomes rarer?%
\footnote{
  The empirical literature on equity term structures also underscores the importance
  of the maturity dimension. Evidence from dividend strips and estimated equity term
  structures documents maturity-dependent risk premia and discount rates
  \citep{vanbinsbergen2012timing,vanbinsbergen2013equity,giglio2024equity}.
}

Two strands of literature address complementary aspects of these questions.
Rare-disaster models show how low-probability contractions can have large effects on
asset prices
\citep{barro2006rare,gourio2012disaster,gabaix2012variable,wachter2013time}, while
models with disaster recovery highlight how post-disaster dynamics can reshape term
structures \citep{hasler2016disaster,wu2025disaster}. Furthermore,
\citet{martin2012valuation} shows that exceptionally adverse aggregate outcomes can
drive the valuation of a broad class of long-dated assets, even without an explicit
disaster mechanism.

Separately, the operator-semigroup approach provides a general framework for studying
asset prices across maturities and characterizing their long-run behavior. For each
maturity, a pricing operator maps a future payoff function into its current price as
a function of the initial economic state. In a time-homogeneous continuous-time
Markov setting, these operators form a \emph{semigroup}
\citep{garman1985towards,hansen2009long}. 
The pricing semigroup combines physical state transitions with stochastic discounting
and can also incorporate cash-flow growth, thereby providing a unified framework for
analyzing the pricing of risk across investment horizons
\citep{hansen2012pricing,hansen2012dynamic}.
In particular, \citet{hansen2009long} factorize the stochastic discount
factor (SDF) into an exponential term, a positive martingale, and an eigenfunction
ratio. Under suitable recurrence and integrability conditions, the exponential term
determines the long-run pricing rate; see also \citet{qin2016positive}.%
\footnote{
  Related contributions extend long-term factorization to general semimartingale
  environments \citep{qin2017long}, develop recovery results linking state prices
  to physical probabilities and discounting under additional restrictions
  \citep{ross2015recovery}, and provide nonparametric methods for estimating the
  eigenfunction decomposition \citep{christensen2017nonparametric}. For spectral
  conditions for the existence and uniqueness of infinite-horizon asset prices,
  see \citet{borovicka2021stability}.
}

In standard finite-state formulations of the operator-semigroup approach, the state
process is commonly assumed to be \emph{irreducible}, meaning that every state is
accessible from every other. Given a strictly positive SDF process, Perron--Frobenius
theory then yields strictly positive left and right principal eigenvectors common to
the pricing operators. Consequently, every nonzero nonnegative terminal payoff has a
positive loading on the dominant component from every initial state, and all such
payoffs share the same long-run pricing rate.%
\footnote{
  See, for example, \citet[Example~6.1 and Section~7.1]{hansen2009long},
  \citet[Section~5]{hansen2012dynamic}, and \citet[Section~I.B]{borovicka2016misspecified}.
  Appendix~\ref{app:hs} reviews this benchmark in the finite-state irreducible
  setting.
}

The questions raised above lie at the intersection of these
literatures. Rare-disaster and recovery models illuminate the roles of arrival,
persistence, recovery, and risk compensation within particular economic
specifications. The operator-semigroup approach provides a general framework for
pricing across horizons, but its standard irreducible formulations do not by
themselves characterize how a region's contribution to a particular price
changes as access to that region becomes rarer.
An economy can remain irreducible while approaching a
\emph{reducible limit}: letting the intensities of entering into a rare region vanish while
retaining transitions out of the region gives a natural reducible approximation. At
that limit, the dominant eigenvectors need no longer be strictly positive, and the
dominant pricing rate need not govern a given claim's price from a given initial
state. Work on nearly decomposable systems
\citep{simon1961aggregation,ando1963near,courtois1975error,meyer1989stochastic}
provides related tools for studying weakly connected dynamics, but its primary focus
is on aggregation and, in stochastic settings, transition dynamics and stationary
distributions rather than the pricing of state-contingent claims.

This paper provides a unified characterization of how rare access to economically
distinct states shapes asset prices across initial states, claims, and payoff
horizons. We model the economic state as a finite-state, time-homogeneous
continuous-time Markov chain, allowing the dynamics to be reducible, and represent
prices by a positive pricing semigroup. Using a graph-theoretic approach, our
analysis begins with a structural observation: the physical and pricing dynamics
partition the state space into the same \emph{communicating classes} (or simply
\emph{classes})---maximal sets of mutually accessible states---and induce the same
accessibility relations among them. Yet the two dynamics need not attach the same
importance to these classes. Physical dynamics distinguish recurrent from transient
classes, whereas pricing dynamics assign each class its own exponential long-run
pricing rate. Physical transience therefore need not imply pricing irrelevance: a
class with vanishing long-run physical probability (i.e., a class of rare states)
can nevertheless have the highest long-run pricing rate.

Our analysis under reducible dynamics yields exact representations of claim
prices and their long-maturity asymptotics.
In particular, we show that the
class-specific pricing rates relevant for a given claim depend jointly on the initial
state and the payoff, and we summarize this dependence through an object we call the
\emph{pricing corridor}: the collection of classes of states lying on directed paths
from the class containing the initial state to classes on which the payoff is
positive. Restricting the pricing dynamics to this corridor leaves the
claim price unchanged at every maturity. Importantly, when the pricing corridor is
nonempty, the highest class-specific pricing rate within the corridor, which we call
the \emph{corridor rate}, determines the long-run exponential rate of the claim
price. If several classes attain the corridor rate, their accessibility
relations determine the degree of the accompanying polynomial factor.

We next consider rare transitions that restore connections absent at the reducible
limit. Suppose the limiting economy has a uniquely dominant pricing class that lies
outside the claim's pricing corridor and has a higher pricing rate than the corridor
rate. Rare transitions can create paths from the initial state through this class to
the payoff support. We show that the contribution of the dominant class then enters
the claim price with a small coefficient but eventually dominates at sufficiently
long maturities.

More specifically, let $\epsilon$ index the strength of the rare connections, with
$\epsilon=0$ corresponding to the reducible limit and sufficiently small $\epsilon>0$
restoring irreducibility. Using a weighted-path analysis related to
minimum-resistance methods \citep{young1993evolution,kandori1993learning}, we show
that the pricing component associated with the dominant class has a coefficient of
order $\epsilon^d$, where $d>0$ measures the rarity of the connecting paths. Under an
additional two-term spectral separation condition, we define the \emph{crossover
maturity}, $t_\epsilon^c$, as the maturity at which this component equals the
component associated with the reducible-limit corridor rate. Beyond this maturity,
the contribution associated with the dominant class exceeds the corridor-rate
component and ultimately governs the claim's price. We derive an explicit expression
for $t^c_\epsilon$ and show that \[ t_\epsilon^c \sim
\frac{d}{\Delta}\log\frac{1}{\epsilon}, \qquad \epsilon\downarrow0, \] where
$\Delta>0$ is the difference between the dominant class's pricing rate and the
corridor rate. Therefore, holding other quantities fixed, rarer access delays the
crossover, whereas a larger pricing-rate advantage brings it forward.

We illustrate these results in a consumption-based disaster-recovery economy with
CRRA preferences. Without disaster entry, the disaster block does not affect prices
from normal initial states, but it can govern long-maturity prices starting in
disaster even when the claim pays only after recovery. The dominance condition
connects risk aversion, consumption dynamics, cash-flow exposure, and recovery
intensity. Changing cash-flow growth exposure can change which class is pricing
dominant, holding physical dynamics and preferences fixed. Restoring rare entry
gives exact two-state prices and crossover maturities from the normal state.

Under an illustrative parameterization with empirically motivated entry frequencies,
disaster durations, and consumption growth characteristis, we show that the
disaster-associated dominant component accounts for roughly 95\% of the price of a
fixed real payment due in 25 years, starting from the normal state. Greater disaster
persistence can nevertheless delay crossover: it widens the pricing-rate gap but
reduces the dominant-component weight.

The rest of the paper is organized as follows. Section~\ref{sec:environment}
introduces the pricing environment and framework.
Sections~\ref{sec:gs}--\ref{sec:network} establish the global pricing rate and
class-based semigroup asymptotics, and Section~\ref{sec:corridors} derives their
state- and payoff-specific counterparts. Section~\ref{sec:illustration} develops the
disaster-recovery illustration and its consumption-based specification.
Section~\ref{sec:near} studies rare transitions, crossover maturities, and the
continuation of the application with rare disaster entry. The appendices provide the
semimartingale foundation, additional comparative statics, and proofs.

\section{Environment and Pricing Semigroup}
\label{sec:environment}

This section introduces the finite-state continuous-time environment and the pricing
objects used throughout the paper. The economic state follows a finite-state Markov
process, while the underlying information flow may include additional shocks to
stochastic discounting and cash-flow growth. We model the SDF and cash-flow growth
processes as strictly positive multiplicative functionals whose product generates the
pricing semigroup. Appendix~\ref{app:param} provides a semimartingale foundation for
this formulation and derives the pricing generator. We use operator and matrix
representations interchangeably; Appendix~\ref{app:notation} records their
equivalence and collects the associated notation. Throughout, as a terminological
convention, we use \emph{pricing} for payoffs delivered at specified maturities and
\emph{valuation} for infinite-horizon cash-flow streams.

\subsection{Markov environment}

Let $\Xsf=\{x_1,\ldots,x_n\}$ be a finite state space. Consider a filtered measurable
space $(\Omega,\fF,(\fF_t^\circ)_{t\geq0})$ supporting a c\`adl\`ag state process
$X=(X_t)_{t\geq0}$ that is adapted to $(\fF_t^\circ)$ and takes values in $\Xsf$. For
each $x\in\Xsf$, let $\PP_x$ be a probability measure on $(\Omega,\fF)$ such that
$\PP_x(X_0=x)=1$, and write $\EE_x$ for expectation under $\PP_x$.
Fix a full-support probability distribution $\varpi$ on $\Xsf$, and define
$\PP=\sum_x\varpi(x)\PP_x$.
We denote by $(\fF_t)_{t\geq0}$ the usual right-continuous $\PP$-augmentation of
the raw filtration $(\fF_t^\circ)_{t\geq0}$.
%
Let $\fF_\infty^\circ:=\sigma(\bigcup_{t\geq0}\fF_t^\circ)$, the smallest $\sigma$-algebra
containing $\fF_t^\circ$ for every $t\geq0$.
Let $(\theta_t)_{t\geq0}$ be a family of measurable shift maps
$\theta_t:\Omega\to\Omega$ satisfying $\theta_0=\id$,
$\theta_{t+u}=\theta_u\circ\theta_t$, and $X_u\circ\theta_t=X_{t+u}$ for all
$t,u\geq0$. We impose the following Markov condition on the economic environment.

\begin{assumption}[Time-homogeneous Markov environment] \label{ass:markov}
  For every $x\in\Xsf$, $t\geq0$, and bounded $\fF_\infty^\circ$-measurable random variable $Y$,
  \begin{equation}\label{eq:EY}
    \EE_x\left(Y\circ\theta_t\mid\fF_t\right)
    =
    \EE_{\scriptscriptstyle X_t}Y,
    \qquad
    \PP_x\text{-a.s.}
  \end{equation}
\end{assumption}

Assumption~\ref{ass:markov} is stronger than requiring $X$ alone to be a time-homogeneous
Markov process. In particular, taking $Y=f(X_u)$ in \eqref{eq:EY} and using
$X_u\circ\theta_t=X_{t+u}$ shows that, under each $\PP_x$, $X$ is a time-homogeneous
continuous-time Markov chain with respect to $(\fF_t)_{t\geq0}$.

Denote by $\RR^\Xsf$ the vector space of real-valued functions on $\Xsf$ and by
$\RR_+^\Xsf$ its cone of nonnegative functions. Equip $\RR^\Xsf$ with the
supremum norm
\(
  \|f\|_\infty\coloneqq\max_{x\in\Xsf}|f(x)|,
\)
and use $\|\cdot\|_\infty$ also for the induced operator norm. For
$f\in\RR^\Xsf$, the \emph{transition operators} $(P_t)_{t\geq0}$ of $X$ are defined by
\begin{equation} \label{eq:Pt}
  (P_tf)(x) \coloneqq \EE_x f(X_t), \qquad t\geq0.
\end{equation}
Identifying $\RR^\Xsf$ with $\RR^n$, each $P_t$ is represented by an $n\times n$
row-stochastic transition matrix. The family $(P_t)_{t\geq0}$ satisfies $P_0=I$ and
$P_{t+u}=P_tP_u$ for $t,u\geq0$.
Since $X$ is finite-state and c\`adl\`ag, $P_t\to I$ as $t\to0$; hence
$(P_t)_{t\geq0}$ is a continuous matrix semigroup. Its infinitesimal generator, which
we call the \emph{Markov generator} of $X$, is
\begin{equation}\label{eq:Lx}
  \Lx := \lim_{t\,\downarrow\,0}\frac{P_t-I}{t}.
\end{equation}
In particular, $\Lx$ is a Metzler matrix (i.e., a matrix with nonnegative
off-diagonal entries) with zero row sums, and $P_t=\exp(t\Lx)$ for every $t\geq0$.%
\footnote{
  For background on Markov transition semigroups and their infinitesimal generators
  on general state spaces, see, for example,
  \citet[Sections~3.1--3.2]{applebaum2009levy}. For continuous-time finite-state
  Markov chains, the properties of the generator matrix and the transition matrix
  semigroup are provided in \citet[Sections~2.1 and~2.8]{norris1997markov}.
}
We interpret $(\mathbb P_x)_{x\in\mathsf X}$ as the physical probability measures
and accordingly also refer to $\Lx$ as the \emph{physical generator}, to distinguish
it from the pricing generator introduced below.

\subsection{Stochastic discounting and cash-flow growth}

Let $S=(S_t)_{t\geq0}$ denote the stochastic discount factor (SDF) process, and let
$G=(G_t)_{t\geq0}$ denote a cumulative stochastic cash-flow growth process.
The process $S$ prices payoffs: under $\PP_x$, the time-zero price of an
$\fF_t$-measurable payoff $Z_t$ delivered at time $t$ is $\EE_x(S_tZ_t)$, provided
that $\EE_x\lvert S_t Z_t\rvert<\infty$.

To represent stochastic discounting and cash-flow growth in a time-homogeneous
manner, we use multiplicative functionals, which are standard in operator
formulations of long-run valuation; see \cite{hansen2009long,hansen2012pricing}.
A jointly measurable, $(\fF_t^\circ)$-adapted process $M=(M_t)_{t\geq 0}$ is called a
\emph{multiplicative functional} if \(M_0=1\) and
\[
  M_{t+u}=M_t(M_u\circ\theta_t), \qquad t,u\geq0, \qquad \PP\text{-a.s.}
\]
It is \emph{strictly positive} if $M_t>0$, $\PP$-almost surely, for every $t\geq0$.
Appendix~\ref{app:additive} shows how exponential additive functionals driven by
diffusion, state-transition, and non-transition jump shocks generate such processes.

\begin{assumption}[Discounting and growth] \label{ass:mf}
  The SDF process $S$ and the cash-flow growth process $G$ are strictly positive
  multiplicative functionals with c\`adl\`ag sample paths.
  Moreover, $\EE_x(S_tG_t)<\infty$ for $x\in\Xsf$ and $t\geq0$.
\end{assumption}

Appendix~\ref{app:additive} gives a semimartingale parameterization of $S$ and
$G$ in terms of continuous shocks, state-transition jumps, and non-transition
jumps. Such parameterizations can accommodate SDF and cash-flow processes
derived from more specific structural models.%
\footnote{
  See \citet[Sections~3.2--3.5]{hansen2009long} for related constructions and examples.
}

For $0\leq s\leq t$, multiplicativity and strict positivity imply 
\(
  S_t/S_s = S_{t-s}\circ\theta_s
\)
and
\(
  G_t/G_s = G_{t-s}\circ\theta_s,
\)
$\PP$-almost surely.
For \(S\), this identity expresses intertemporal consistency in dynamic
valuation \citep{garman1985towards,duffie1986intertemporal}; for \(G\),
it is the corresponding compounding rule for cumulative cash-flow
growth. Since the product of two multiplicative functionals is itself
multiplicative, the \emph{growth-adjusted SDF} process defined by
\begin{equation*}
  S^{\Gs}_t \coloneqq S_tG_t, \qquad t\geq 0,
\end{equation*}
is also a strictly positive multiplicative functional. 
The moment condition in Assumption~\ref{ass:mf} can therefore be written as
$\EE_xS_t^{\Gs}<\infty$ for $x\in\Xsf$ and $t\geq0$.

\subsection{Pricing semigroup and pricing generator}

A function $g\in\RR^\Xsf$ specifies a state-contingent payoff profile: $g(x)$
determines the terminal payoff associated with state $x$. Thus $G_tg(X_t)$ is the
payoff delivered at time $t$, with $G_t$ capturing cumulative cash-flow growth and
$g(X_t)$ capturing its dependence on the economic state at maturity.

For each $t\geq0$, we call the linear operator
$\Pi_t:\RR^\Xsf\to\RR^\Xsf$, defined by
\begin{equation} \label{eq:Pit}
  (\Pi_tg)(x)
  \coloneqq
  \EE_x \, S_t^{\Gs}g(X_t)
  =
  \EE_x \, S_t G_t g(X_t),
  \qquad x\in\Xsf,\quad g\in\RR^\Xsf,
\end{equation}
the \emph{pricing operator} for maturity-$t$ payoffs. Thus $(\Pi_tg)(x)$ is the
time-zero price of $G_tg(X_t)$, conditional on $X_0=x$.
Assumption~\ref{ass:mf} makes $\Pi_t$ well defined, since
\(
  |(\Pi_tg)(x)|
  \leq
  \|g\|_\infty\,\EE_x S_t^{\Gs}
  <\infty.
\)
Moreover, $S_0^{\Gs}=1$ and $X_0=x$, $\PP_x$-almost surely, imply $\Pi_0=I$. Since
$S_t^{\Gs}>0$, we also have $\Pi_tg\geq0$ whenever $g\geq0$.
Hence $\Pi_t$ is a positive linear operator. Identifying $\RR^\Xsf$ with $\RR^n$, its
representing matrix is entrywise nonnegative.
Moreover, Assumptions~\ref{ass:markov} and \ref{ass:mf} imply the \emph{semigroup} property:
\begin{equation} \label{eq:sgp}
  \Pi_{t+u}=\Pi_t\,\Pi_u, \qquad t,u\geq0.
\end{equation}
Thus $\Pi=(\Pi_t)_{t\geq0}$ is a positive matrix semigroup, which we call the
\emph{pricing semigroup}.%
\footnote{
  This semigroup is the counterpart of the cash-flow-growth semigroup $Q$ in
  \citet[Section~4.4 and Table~I]{hansen2009long}. Hansen and Scheinkman describe
  this object as valuation with stochastic growth. Under our terminology, we call it
  a pricing semigroup because $\Pi_t$ assigns a time-zero price to a payoff delivered
  at maturity $t$. Appendix~\ref{app:mf-semigroup} verifies the semigroup property
  under our assumptions.
}
To obtain the infinitesimal generator of the pricing semigroup, we impose the following
continuity condition.

\begin{assumption} \label{ass:continuity}
  For every $x\in\Xsf$,
  \(
    \lim_{t\to0}\EE_x\,S_tG_t=1.
  \)
\end{assumption}
Together with the strict positivity of $S$ and $G$, Assumption~\ref{ass:continuity}
and the c\`adl\`ag paths imply that $\Pi_t\to I$ in operator norm as $t\to0$;
Appendix~\ref{app:mf-semigroup} gives the short proof. Hence $\Pi$ is a continuous
matrix semigroup. Its infinitesimal generator, which we call the \emph{pricing
generator}, is 
\begin{equation} \label{eq:Lp}
  \Lp := \lim_{t\,\downarrow\,0}\frac{\Pi_t-I}{t},
\end{equation}
where the limit is taken in operator norm. Finite-dimensional semigroup theory then
gives
\[
  \Pi_t=\e^{t\Lp}, \qquad t\geq 0.
\]
Because $\Pi$ is a positive matrix semigroup, its generator $\Lp$ is a
Metzler matrix.%
\footnote{
  For the definition of the infinitesimal generator and the exponential
  representation of a continuous matrix semigroup, see \citet[Proposition~9.4 and
  Definition~9.14]{batkai2017positive}. The equivalence between positivity of the
  matrix semigroup and the Metzler property of its generator is given in their
  Theorem~7.1.
}
Unlike the physical generator $\Lx$, however, $\Lp$ need not have zero row sums,
because it incorporates stochastic discounting and cash-flow growth in addition to
physical state transitions.
Appendix~\ref{app:effective-discounting} derives the entries of $\Lp$ from the
semimartingale specification and gives its Feynman--Kac decomposition into
pricing-adjusted transition dynamics and state-specific effective discounting.

\subsection{Strip prices, yields, and infinite-horizon values}

For $t\geq 0$ and $g\in\RR^\Xsf$, define the \emph{maturity-$t$ strip price}
\begin{equation} \label{eq:strip}
  q_t(x,g) \coloneqq (\Pi_tg)(x) = \bigl(\e^{t\Lp}g\bigr)(x),
  \qquad x\in\Xsf.
\end{equation}
In light of \eqref{eq:Pit}, $q_t(x,g)$ is the time-zero price of the payoff $G_t
g(X_t)$, delivered at time $t$, conditional on $X_0=x$.
For $t>0$, whenever
$q_t(x,g)>0$, define the corresponding continuously compounded strip yield by
\[
  y_t(x,g) \coloneqq -\frac{1}{t}\log q_t(x,g),
  \qquad x\in\Xsf.
\]
Let $\1\in\RR^\Xsf$ denote the function identically equal to one. The special
case $g=\1$ gives
\begin{equation} \label{eq:zc}
  z_t(x) \coloneqq q_t(x,\1) = (\Pi_t\1)(x),
  \quad t\geq0,
  \qquad
  y_t(x) \coloneqq -\frac{1}{t}\log z_t(x),
  \quad t>0.
\end{equation}
Here $z_t(x)$ is the time-zero price of the payoff $G_t$ delivered at date $t$. We
call $z_t$ the \emph{growth-adjusted zero-coupon price} and $y_t$ its continuously
compounded yield. Since $S_t^{\Gs}>0$, $z_t(x)>0$ for every $x\in\Xsf$ and $t\geq0$.

To pass from single-maturity payoffs to infinite-horizon assets%
\footnote{
  Examples include common equity, modeled as a claim to an indefinite
  dividend stream, and perpetuities such as consols, which pay coupons without a
  fixed maturity.
},
fix $g\in\RR_+^\Xsf$ and consider an asset paying the instantaneous cash-flow rate
$G_tg(X_t)$ for $t\geq0$. Its cash flow over an interval $[a,b]$ is
$\int_a^bG_tg(X_t)\,\diff t$. We define the \emph{valuation operator} $V$ for such
assets by
\begin{equation} \label{eq:ihf}
  (Vg)(x)
  \coloneqq
  \EE_x\int_0^\infty
    S_t G_t g(X_t)\,\diff t
  =
  \int_0^\infty q_t(x,g)\,\diff t
  \in[0,\infty],
  \qquad x\in\Xsf.
\end{equation}
The equality follows from Tonelli's theorem. At this stage, $V$ may take the value
$\infty$; under the global valuation stability condition introduced in the next
section, it extends to a finite positive linear operator.


\section{Global Long-Run Pricing Rate and Valuation Stability}\label{sec:gs}

In this section, we identify the global long-run pricing rate and relate its sign to
the finiteness of infinite-horizon values. The starting point is the \emph{spectral
bound} of the pricing generator.
For a square matrix \(A\), let $\sigma(A)$ denote its spectrum, and
\(
  s(A)
  \coloneqq
  \max\{\Re\lambda:\lambda\in\sigma(A)\}
\)
denote its spectral bound. Define 
\[
  \lambda^* \coloneqq s(\Lp).
\]
Since $\Lp$ is Metzler, Lemma~\ref{lem:msb} implies that
\begin{equation} \label{eq:gsv}
  \lambda^*
  \in
  \sigma(\Lp)\cap\RR,
  \qquad
  \Re\mu<\lambda^*
  \text{ for every }\mu\in\sigma(\Lp)\setminus\{\lambda^*\}.
\end{equation}
That is, $\lambda^*$ is a real eigenvalue of $\Lp$ whose real part strictly exceeds
that of every distinct eigenvalue.

If $\Lp$ is irreducible, $\lambda^*$ is the algebraically simple principal eigenvalue
associated with the strictly positive eigenfunction/eigenvector used in the
Hansen--Scheinkman factorization of the growth-adjusted SDF; Appendix~\ref{app:hs}
reviews this case and its implications for long-run pricing. Under reducibility,
$\lambda^*$ need not be algebraically simple or admit strictly positive eigenvectors,
and hence it need not govern every individual price.

To identify the global pricing interpretation of $\lambda^*$ without imposing
irreducibility, define the \emph{statewise zero-coupon price envelope} and its
associated \emph{envelope yield} by
\begin{equation*}
  z_t^{\mathrm{env}}
  \coloneqq
  \max_{x\in\Xsf}z_t(x),
  \quad t\geq0,
  \qquad\ \ 
  y_t^{\mathrm{env}}
  \coloneqq
  -\frac{1}{t}\log z_t^{\mathrm{env}}
  = \min_{x\in\Xsf}y_t(x),
  \quad t>0.
\end{equation*}
The following proposition relates their long-maturity behavior to $\lambda^*$.

\begin{proposition} \label{prop:ge}
  Under Assumptions~\ref{ass:markov}--\ref{ass:continuity},
  \begin{equation*} 
    \lim_{t\to\infty}
    \frac{1}{t}\log z_t^{\mathrm{env}}
    =
    \lambda^*,
    \qquad
    \lim_{t\to\infty}
    y_t^{\mathrm{env}}
    =
    -\lambda^*.
  \end{equation*}
\end{proposition}

Proposition~\ref{prop:ge} identifies $\lambda^*$ as the \emph{global long-run
pricing rate}: it is the asymptotic exponential rate of the statewise price
envelope, while $-\lambda^*$ is the corresponding long-run envelope yield.%
\footnote{
  The same rate also governs ex ante zero-coupon prices under any full-support
  distribution of the initial state; see Corollary~\ref{cor:ex-ante-rate} in
  Appendix~\ref{app:gvs}.
}
In fact, the sign of this rate also determines whether every modeled nonnegative
infinite-horizon cash-flow stream has finite value from every initial state. We call
this property \emph{global valuation stability}.

\begin{definition}\label{def:gs}
  We say that \emph{global valuation stability} holds if
  \begin{equation}\label{eq:gvs}
    (Vg)(x) < \infty
    \qquad
    \text{for every } g\in\RR_+^\Xsf
    \text{ and } x\in\Xsf.
  \end{equation}
\end{definition}

Because $0\leq g\leq\|g\|_\infty\1$ for every $g\in\RR_+^\Xsf$ and
$\Pi_t$ is positive,
\(
  0\leq q_t(x,g) \leq \|g\|_\infty z_t(x).
\)
Consequently, global valuation stability holds if and only if
\(
  (V\1)(x)
  =
  \int_0^\infty z_t(x)\diff t
  <\infty
\)
for every $x\in\Xsf$. Thus this finiteness condition can be tested using the
constant function $g=\1$ alone. The following theorem links this condition to
$\lambda^*$, the price envelope, and the yield envelope.

\begin{theorem}[Global valuation stability] \label{thm:gs}
  Under Assumptions~\ref{ass:markov}--\ref{ass:continuity}, the following statements
  are equivalent:
  \begin{enumerate}
    \item[{\rm (G1)}] Global valuation stability holds.
    \item[{\rm (G2)}] The global long-run pricing rate is negative: $\lambda^*<0$.
    \item[{\rm (G3)}] The long-run envelope yield is positive: $\lim_{t\to\infty}y_t^{\mathrm{env}}>0$.
    \item[{\rm (G4)}] The price envelope is below one at some maturity:
      $z_t^{\mathrm{env}}<1$ for some $t>0$.
    \item[{\rm (G5)}] The price envelope vanishes at long maturities:
      $\lim_{t\to\infty}z_t^{\mathrm{env}}=0$.
  \end{enumerate}
  When these conditions hold, the infinite-horizon valuation operator $V$ extends
  uniquely to a positive linear operator on $\RR^\Xsf$, and the operator-valued
  integral
  \(
    \int_0^\infty\Pi_t\,\diff t
  \)
  converges in norm. Moreover, $\Lp$ is invertible and $V$ admits the representation
  \[
    V
    =
    \int_0^\infty \Pi_t\,\diff t
    =
    -\Lp^{-1}.
  \]
\end{theorem}

Theorem~\ref{thm:gs} provides equivalent spectral, yield-based, and price-based
characterizations of global valuation stability. Specifically, global valuation
stability holds if and only if the global long-run pricing rate is negative, or
equivalently, the long-run envelope yield is positive. It also holds if and only if
the price envelope falls below one at some finite maturity, in which case the
envelope vanishes at long maturities.

The sign condition \(\lambda^*=s(\Lp)<0\) is stronger than invertibility of \(\Lp\).
For example, in a one-state economy with \(\Lp=a>0\), the equation
\(\Lp v=-1\) has the solution \(v=-1/a\), but this negative number cannot be
the value of a nonnegative cash-flow stream, whose infinite-horizon value is
infinite. Thus algebraic solvability of \(\Lp v=-1\) does not ensure that
\(-\Lp^{-1}\) is a positive valuation operator; stability, rather than
invertibility alone, justifies the inverse-generator representation.

The global rate $\lambda^*$ and the valuation stability criterion provide an
economy-wide characterization, but do not identify which parts of the state space
determine that rate or which prices inherit it. Addressing these questions requires
examining which states can reach one another and how this accessibility structure
interacts with pricing.

\section{Class Structure and Pricing Semigroup Asymptotics}
\label{sec:network}

We organize the state space into communicating classes and show how their
class-specific pricing rates and accessibility relations determine the pricing
semigroup's long-run behavior. In particular, the physical and pricing dynamics share
this class structure, but physical recurrence need not coincide with pricing
dominance. The matrix and graph-theoretic preliminaries are collected in
Appendix~\ref{app:preliminaries}.

\subsection{Communicating classes and pricing rates}
\label{sec:class-structure}

We call $\Gcal(\Lx)$ and $\Gcal(\Lp)$ the \emph{physical graph} and \emph{pricing
graph}, respectively. In each graph, directed edges correspond to positive
off-diagonal entries of the generator; see Appendix~\ref{app:metzler-graphs} for the
graph conventions.

Strict positivity of the growth-adjusted SDF ensures that the two graphs have
the same accessibility relations. Specifically, by \eqref{eq:Pt} and \eqref{eq:Pit},
for any states $x_i,x_j$,
\(
  (P_t)_{ij}=\PP_{x_i}(X_t=x_j)
\)
and
\(
  (\Pi_t)_{ij}
  =\EE_{x_i}\bigl[S_t^{\Gs}1_{\{X_t=x_j\}}\bigr].
\)
Since $S_t^{\Gs}>0$ almost surely, $(P_t)_{ij}>0$ if and only if $(\Pi_t)_{ij}>0$
for all $t>0$.
Lemma~\ref{lem:mr} therefore implies that the physical and pricing graphs
induce the same accessibility relation on $\Xsf$. Because communication is
mutual accessibility, they also partition $\Xsf$ into the same communicating
classes. Denote these classes by $\Ccal_1,\ldots,\Ccal_m$, and write
\[
  \cC\coloneqq\{\Ccal_1,\ldots,\Ccal_m\}.
\]
We call the economy \emph{irreducible} if all states belong to a single class
($m=1$), and \emph{reducible} otherwise. Clearly, this is equivalent to
irreducibility or reducibility of either $\Lx$ or $\Lp$.

Collapsing each class into a single vertex gives the \emph{reduced physical graph}
$\Gcal_{\mathrm r}(\Lx)$ and the \emph{reduced pricing graph}
$\Gcal_{\mathrm r}(\Lp)$. These graphs share the vertex set $\cC$ and the same
class-level accessibility relation. We write $x_i\leadsto x_j$ and
$\Ccal_k\leadsto\Ccal_\ell$ for the common state- and class-level accessibility
relations, respectively. Paths of length zero are allowed, so every state and
class is accessible from itself. Unless stated otherwise, $\Ccal_k\to\Ccal_\ell$
denotes a directed edge of the reduced pricing graph.

A class is \emph{closed} if it has no access to any other class. The common
accessibility relation therefore gives the physical and pricing graphs the same
closed classes. Under the finite-state physical Markov dynamics, closed classes
are \emph{recurrent}, while all other classes are \emph{transient}. From every
initial state, the probability of occupying a transient class tends to zero as
time increases.

The reduced graphs are acyclic, so we may order the classes such that
$\Ccal_k\leadsto\Ccal_\ell$ with $k\neq\ell$ implies $k>\ell$. We place all
closed classes first, in arbitrary order, and order the remaining classes subject
to this requirement.
The corresponding ordering of states puts the pricing generator $\Lp$ in block
lower-triangular Frobenius form
\citep[Section~2]{schneider1986influence}:
\begin{equation} \label{eq:frobenius}
  \Lp
  =
  \begin{pmatrix}
    L_1 & 0 & 0 & \cdots & 0\\
    L_{21} & L_2 & 0 & \cdots & 0\\
    L_{31} & L_{32} & L_3 & \cdots & 0\\
    \vdots & \vdots & \vdots & \ddots & \vdots\\
    L_{m1} & L_{m2} & L_{m3} & \cdots & L_m
  \end{pmatrix}.
\end{equation}
Each diagonal block $L_k$ is the principal submatrix of $\Lp$ indexed by the states in
$\Ccal_k$. Since each $\Ccal_k$ is a class of the pricing graph, $L_k$ is an
irreducible Metzler matrix. For $k>\ell$, the off-diagonal block $L_{k\ell}$ is
nonzero if and only if the reduced pricing graph contains the directed edge
$\Ccal_k\to\Ccal_\ell$.

\begin{figure}[t!]
\centering
\begin{subfigure}[t]{0.48\textwidth}
\centering
\begin{minipage}[c][30mm][c]{\linewidth}
\centering
\begin{tikzpicture}[scale=0.86,transform shape]
  \node[state] (d1) at (0,0.8) {$D_1$};
  \node[state] (d2) at (0,-0.8) {$D_2$};
  \node[state] (r1) at (2.5,0.8) {$R_1$};
  \node[state] (r2) at (2.5,-0.8) {$R_2$};
  \node[state] (n1) at (5,0.8) {$N_1$};
  \node[state] (n2) at (5,-0.8) {$N_2$};
  \draw[biflow] (d1) -- (d2);
  \draw[biflow] (r1) -- (r2);
  \draw[biflow] (n1) -- (n2);
  \draw[flow] (d2) to[bend right=8] (r2);
  \draw[flow] (d2) -- (r1);
  \draw[flow] (r1) to[bend left=8] (n1);
  \draw[flow] (r1) -- (n2);
  \begin{scope}[on background layer]
    \node[classbox,fit=(d1)(d2),label=above:{\scriptsize Disaster $\mathcal D$}] {};
    \node[classbox,fit=(r1)(r2),label=above:{\scriptsize Recovery $\mathcal R$}] {};
    \node[classbox,fit=(n1)(n2),label=above:{\scriptsize Normal $\mathcal N$}] {};
  \end{scope}
\end{tikzpicture}
\end{minipage}
\caption{Pricing graph $\Gcal(\Lp)$.}
\end{subfigure}\hfill
\begin{subfigure}[t]{0.48\textwidth}
\centering
\begin{minipage}[c][30mm][c]{\linewidth}
\centering
\begin{tikzpicture}[node distance=7mm,scale=0.82,transform shape]
  \node[classnode] (D) {Disaster\\$\mathcal D$};
  \node[classnode,right=of D] (R) {Recovery\\$\mathcal R$};
  \node[classnode,very thick,double,right=of R] (N) {Normal\\$\mathcal N$};
  \draw[flow] (D) -- (R);
  \draw[flow] (R) -- (N);
\end{tikzpicture}
\end{minipage}
\caption{Reduced pricing graph $\Gcal_{\mathrm{r}}(\Lp)$.}
\end{subfigure}
\caption{The pricing graph and its reduced graph in a disaster-recovery economy.
Panel~(b) collapses the classes of the pricing graph in panel~(a). The double border identifies
$\mathcal N$ as the unique closed class and hence the recurrent class supporting
the long-run physical distribution.}
\label{fig:graph}
\end{figure}
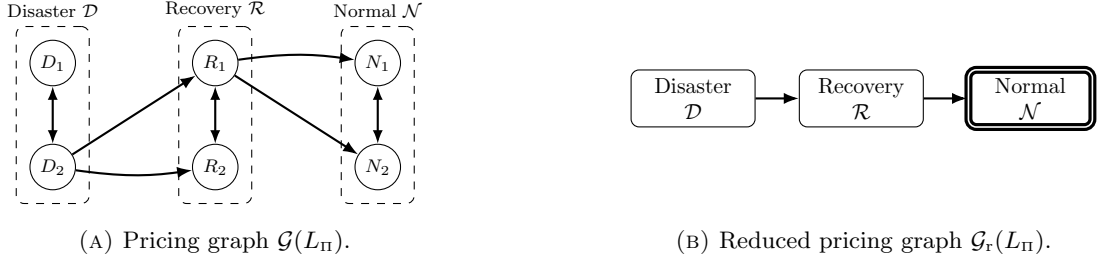

For each class, define its \emph{class long-run pricing rate} (or simply \emph{class
pricing rate}) by
\[
  \lambda_k \coloneqq s(L_k), \qquad k=1,\ldots,m.
\]
By Perron--Frobenius theory for irreducible Metzler matrices, $\lambda_k$ is
an algebraically simple real eigenvalue of $L_k$ that is strictly dominant in real
part over all other eigenvalues.
Since $\Lp$ is block triangular, $\sigma(\Lp) = \bigcup_{k=1}^m \sigma(L_k)$, and
therefore
\begin{equation} \label{eq:bs}
  \lambda^*
  =
  s(\Lp)
  =
  \max_{1\leq k\leq m}\lambda_k.
\end{equation}
We say a class \(\Ccal_k\in\cC\) is \emph{globally pricing-dominant} when
\(\lambda_k=\lambda^*\). Let
\begin{equation*} 
  \cC^*
  \coloneqq
  \{\Ccal_k\in\cC:\lambda_k=\lambda^*\}
\end{equation*}
denote the set of globally pricing-dominant classes. Together with \eqref{eq:bs},
Theorem~\ref{thm:gs} implies that global valuation stability holds if and only if
\(\lambda_k<0\) for every \(k=1,\ldots,m\).

The Frobenius form \eqref{eq:frobenius} separates class-level accessibility from
class pricing rates. Given the class decomposition, the zero pattern of the
off-diagonal blocks determines the reduced pricing graph and, through the common
accessibility relation, the recurrent--transient classification under the physical
dynamics. The diagonal blocks, by contrast, determine the class pricing rates and
hence the global rate \(\lambda^*\). Thus, global pricing dominance is a
diagonal-block property.

Figure~\ref{fig:graph} illustrates a pricing graph and its reduced graph in a
three-class disaster-recovery economy comprising a disaster class \(\mathcal D\), a
recovery class \(\mathcal R\), and a normal class \(\mathcal N\).
Communication within each class, together with the interclass edges shown in
panel~(a), gives every disaster state access to recovery and every recovery state
access to normal. Collapsing the classes therefore yields the reduced graph
\(\mathcal D\to\mathcal R\to\mathcal N\), which implies the class-level accessibility
relations
$\mathcal D\leadsto\mathcal R$,
$\mathcal R\leadsto\mathcal N$, and 
$\mathcal D\leadsto\mathcal N$.
Because the physical and pricing graphs induce the same class-level accessibility
relation, \(\mathcal N\) is the unique closed class in both reduced graphs and
therefore the unique recurrent class under the physical dynamics. The nonclosed
classes \(\mathcal D\) and \(\mathcal R\) are transient.
Consequently, the state process eventually enters the normal class with probability
one, and the probability of occupying that class tends to one.

Pricing dominance need not coincide with physical recurrence. Once discounting,
risk adjustment, and cash-flow growth are incorporated into the class blocks of
\(\Lp\), the recurrent normal class need not have the largest pricing rate. If the
disaster block has the largest spectral bound, the physically transient class
\(\mathcal D\) determines \(\lambda^*\) even though its physical probability vanishes
asymptotically. If this rate is nonnegative, global valuation stability fails despite
eventual physical recovery. Thus, a physically transient class can determine the
global long-run pricing rate.

\subsection{Unique pricing dominance}
\label{sec:unique-dominance}

Suppose that a single class attains the global pricing rate, $\cC^*=\{\Ccal^*\}$.
This condition allows the economy to be reducible. The following theorem shows that
the exponentially normalized pricing semigroup converges to a rank-one spectral
projection, whose right and left eigenvectors have supports determined by
accessibility to and from $\Ccal^*$, respectively.

\begin{theorem} \label{thm:Pi-asym}
  Under Assumptions~\ref{ass:markov}--\ref{ass:continuity}, suppose that
  $\cC^*=\{\Ccal^*\}$. Then $\lambda^*$ is an algebraically simple eigenvalue of
  $\Lp$. There exist unique nonzero nonnegative right and left eigenvectors $\phi$
  and $\psi$, normalized by $\1^\t\phi=1$ and $\psi^\t\phi=1$, such that
  $\Lp\phi=\lambda^*\phi$ and $\Lp^\t\psi=\lambda^*\psi$. 
  Moreover, $\phi\psi^\t$ is the rank-one spectral projection of $\Lp$ associated with
  $\lambda^*$, and
  \begin{equation} \label{eq:r1lim}
    \lim_{t\to\infty}\e^{-\lambda^*t}\,\Pi_t
    =
    \phi\psi^\t.
  \end{equation}
  The supports of these eigenvectors are determined by class-level accessibility. For
  each $\Ccal\in\cC$, let $\phi_{\scriptscriptstyle\Ccal}$ and
  $\psi_{\scriptscriptstyle\Ccal}$ denote the subvectors of $\phi$ and $\psi$ indexed
  by the states in $\Ccal$. Then,
  \begin{equation} \label{eq:supports}
    \phi_{\scriptscriptstyle\Ccal}
    \begin{cases}
      \gg0, & \text{if } \Ccal\leadsto\Ccal^*,\\
      =0, & \text{otherwise},
    \end{cases}
    \qquad
    \psi_{\scriptscriptstyle\Ccal}
    \begin{cases}
      \gg0, & \text{if } \Ccal^*\leadsto\Ccal,\\
      =0, & \text{otherwise}.
    \end{cases}
  \end{equation}
\end{theorem}

For every $x\in\Xsf$ and nonzero $g\in\RR_+^\Xsf$, \eqref{eq:r1lim} implies that
\[
  \e^{-\lambda^*t}q_t(x,g)\to\phi(x)\psi^\t g,
  \qquad \text{as } t\to\infty.
\]
Let $\Ccal$ be the class containing $x$. By \eqref{eq:supports}, the 
limit is positive precisely when $\Ccal\leadsto\Ccal^*\leadsto\Ccal'$ for some class
$\Ccal'$ on which $g$ is nonzero. Thus the global pricing rate governs the price's leading
asymptotic behavior exactly when the initial state can reach the dominant class and
that class can reach a state where the payoff is positive.

The dominant eigenvectors also give a Markov interpretation of long-run pricing.
Under reducibility, the right eigenvector may vanish, so the
Hansen--Scheinkman change of measure is available only from states in its support.
Let $\phi_{\scriptscriptstyle+}$ be the restriction of $\phi$ to $\supp(\phi)$
and $\Lp^{\phi_+}$ the corresponding principal submatrix of $\Lp$.
The \emph{right Doob chain} has generator
\[
  \Lp^\rightarrow
  \coloneqq
  \diag(\phi_{\scriptscriptstyle+})^{-1}
  \bigl(\Lp^{\phi_+}-\lambda^*I\bigr)
  \diag(\phi_{\scriptscriptstyle+}).
\]
This is the Doob $h$-transform with $h=\phi_{\scriptscriptstyle+}$
\citep[see, e.g.,][]{doob1957conditional}. Under irreducibility, it reduces to the
whole-state Hansen--Scheinkman transform reviewed in Appendix~\ref{app:hs}.

The right Doob chain has $\Ccal^*$ as its unique recurrent class, and its
distribution converges from every state in $\supp(\phi)$ to
$\pi^{\mathrm{LR}}\coloneqq\phi\odot\psi$, extended by zero outside that support.
We call $\pi^{\mathrm{LR}}$ the \emph{long-run pricing law}. It is supported on
$\Ccal^*$, so a physically transient class can carry all long-run pricing mass.%
\footnote{The same weights determine the sensitivity of the long-run envelope yield
  to state-specific effective discount rates; see
  Appendix~\ref{app:local-pass-through}.
}
Appendix~\ref{app:doob} establishes these claims and the change-of-measure
construction. It also constructs the left Doob chain from the adjoint pricing
semigroup. The left chain has the same long-run pricing law, but its state space
can differ.

Returning to Figure~\ref{fig:graph}, suppose the disaster class $\mathcal D$ is
uniquely pricing-dominant. Then $\supp(\phi)=\mathcal D$ and $\supp(\psi)=\Xsf$. Thus
the disaster pricing rate governs every nonzero nonnegative payoff priced from a
disaster state, including payoffs delivered only after recovery, but does not govern
prices starting in recovery or normal states. The right Doob chain remains within
$\mathcal D$, and its long-run pricing law is concentrated there, although the
physical state process eventually enters $\mathcal N$.

\subsection{Tied pricing dominance}

When several classes attain the global pricing rate $\lambda^*$, their
accessibility relations determine whether polynomial factors accompany the leading
exponential term. The relevant quantity is the largest number of globally dominant
classes that can be visited along a single directed path. We call this number the
\emph{global critical depth} and define
\begin{equation} \label{eq:global-critical-depth}
  \nu^*
  \coloneqq
  \max_{\substack{\gamma\text{ is a chain of classes}\\
                   \gamma\subseteq\cC}}
  \#\{\Ccal_k\in\gamma:\lambda_k=\lambda^*\}.
\end{equation}
Here \emph{chains} are defined by accessibility in the reduced pricing graph
$\Gcal_{\mathrm r}(\Lp)$; see Appendix~\ref{app:metzler-graphs}.%
\footnote{
  Upon a scalar shift that makes $\Lp$ nonnegative, the globally dominant classes in
  $\cC^*$ are precisely the \emph{basic classes} of \citet{rothblum1975algebraic}.
  In that terminology, the \emph{length} of a chain is the number of basic classes
  it contains; $\nu^*$ is the maximum such length over all chains.
}

The following theorem describes the long-maturity asymptotics of the pricing
semigroup, allowing for multiple globally dominant classes. Let
$P^*\coloneqq P_{\lambda^*}(\Lp)$ be the spectral projection onto the generalized
eigenspace associated with $\lambda^*$, and define
\(
  B^*\coloneqq [(\Lp-\lambda^*I)P^*]^{\nu^*-1}P^*/(\nu^*-1)!.
\)

\begin{theorem}\label{thm:Pi-asym-tied}
  Under Assumptions~\ref{ass:markov}--\ref{ass:continuity}, the largest Jordan block
  of $\Lp$ associated with $\lambda^*$ has size $\nu^*$. Moreover,
  \begin{equation} \label{eq:price-glim}
    \lim_{t\to\infty} t^{1-\nu^*}\e^{-\lambda^*t}\,\Pi_t
    =
    B^*,
  \end{equation}
  where $B^*$ is a nonzero, nonnegative matrix satisfying
  $\Lp B^*=B^*\Lp=\lambda^*B^*$. In particular, $\lambda^*$ is
  semisimple if and only if no chain of classes contains more than one class in
  $\cC^*$; in that case $\nu^*=1$ and $B^*=P^*$.
\end{theorem}

Theorem~\ref{thm:Pi-asym-tied} is the continuous-time Metzler formulation of classical
reduced-graph index results
\citep{rothblum1975algebraic,friedland1980growth,schneider1986influence}.
It covers \eqref{eq:r1lim} as a special case: when there is a unique pricing
dominant class, $\nu^*=1$ and $B^*=P^*=\phi\psi^\t$. If $\nu^*>1$, then the long-run
behavior of the pricing operators includes an additional polynomial factor.%
\footnote{
  Tied class rates alone do not produce polynomial factors: if no globally dominant
  class can reach another, $\nu^*=1$ and exponential normalization suffices.
}

For a given initial state $x$ and payoff $g$, the price term associated with the
global exponential rate and polynomial order has coefficient $(B^*g)(x)$, which may
be zero. A positive price may then have a lower exponential rate, or retain the
global exponential rate with a lower polynomial order. In the next section,
we refine the analysis and identify the leading rate and order by restricting the
pricing generator to classes lying on paths from the initial state to states where
the payoff is positive.

\section{Pricing Corridors and State- and Payoff-Specific Asymptotics}
\label{sec:corridors}

\begin{figure}[t!]
  \centering
  \begin{tikzpicture}[node distance=12mm and 20mm,scale=0.9,transform shape]
    \node[corridorclass] (D) {Disaster\\$\mathcal D$};
    \node[corridorclass,right=of D] (R) {Recovery\\$\mathcal R$};
    \node[corridorclass,right=of R] (N) {Normal\\$\mathcal N$};
    \node[classnode,below=of R] (L) {Liquidation\\$\mathcal L$};
    \node[classnode,above=of N] (F) {Fiscal stress\\$\mathcal F$};
    \draw[flow] (D) -- (R);
    \draw[flow] (R) -- (N);
    \draw[flow] (D) -- (L);
    \draw[flow] (F) -- (N);
  \end{tikzpicture}
  \caption{The pricing corridor in an extension of the reduced pricing graph in
    Figure~\ref{fig:graph}(b). For current class $\mathcal D$ and payoff support
    $\sS(g)=\{\mathcal D,\mathcal N\}$, the shaded classes form the corridor
    $\kK(\mathcal D,g)=\{\mathcal D,\mathcal R,\mathcal N\}$: each is accessible
    from disaster and can reach a class where the payoff is positive.
    Liquidation cannot reach the payoff support, while fiscal stress is
    inaccessible from disaster.}
  \label{fig:vc}
\end{figure}
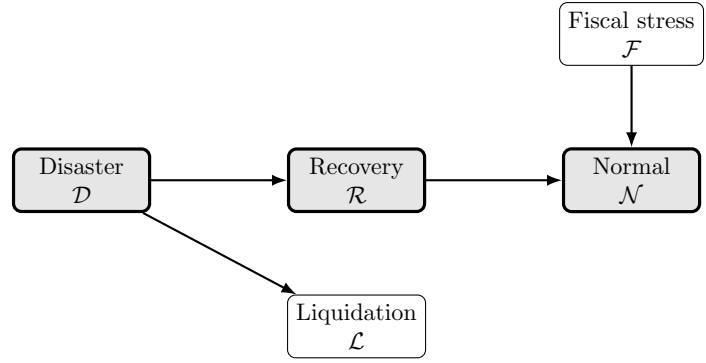

For a given initial state and nonnegative payoff, restricting the pricing generator
and payoff to the pricing corridor defined below preserves the price at every
maturity. The restricted generator determines the price's asymptotic exponential rate
and polynomial order, as well as whether the corresponding perpetual cash-flow stream
has finite value.

\subsection{Pricing corridors and restricted pricing generators}

For $g\in\RR_+^\Xsf$, define its state-level support by
\(
  \supp(g)\coloneqq\{x\in\Xsf:g(x)>0\}.
\)
Writing $g_{\Ccal}$ for its restriction to $\Ccal\in\cC$, define its
\emph{class-level support} by
\(
  \sS(g)\coloneqq\{\Ccal\in\cC:g_{\Ccal}\neq0\}.
\)
For an initial class $\Ccal\in\cC$, we call the following collection of classes
the \emph{pricing corridor}:
\begin{equation} \label{eq:corridor}
  \kK(\Ccal,g)
  \coloneqq
  \{\Ccal'\in\cC:\Ccal\leadsto\Ccal'\leadsto\Ccal''
  \text{ for some }\Ccal''\in\sS(g)\}.
\end{equation}
Thus, $\kK(\Ccal,g)$ contains exactly the classes lying on at least one directed
class path from the current class $\Ccal$ to a class on which the payoff $g$ is
nonzero.

Figure~\ref{fig:vc} extends the disaster--recovery example by adding liquidation
and fiscal-stress classes. With initial class $\mathcal D$ and payoff support
$\sS(g)=\{\mathcal D,\mathcal N\}$, the pricing corridor is
$\kK(\mathcal D,g)=\{\mathcal D,\mathcal R,\mathcal N\}$.
Liquidation is accessible from disaster but cannot reach the payoff support;
fiscal stress can reach the payoff support but is inaccessible from disaster.
For the same payoff, $\kK(\mathcal N,g)=\{\mathcal N\}$ and
$\kK(\mathcal L,g)=\varnothing$. Thus, holding the pricing graph and payoff fixed,
the corridor may contain several classes, one class, or no class, depending on the
initial state.

Restricting the pricing generator and payoff to the corridor leaves the price
unchanged at every maturity. To state this precisely, for a nonempty corridor
$\kK=\kK(\Ccal,g)$, let
\[
  \Xsf_{\kKs}
  \coloneqq
  \bigcup_{\Ccal'\in\kK}\Ccal',
  \qquad
  L_{\kKs}
  \coloneqq
  \Lp(\Xsf_{\kKs}, \Xsf_{\kKs}),
  \qquad
  g_{\kKs}
  \coloneqq
  g_{\scriptscriptstyle\Xsf_{\kKs}}.
\]
The \emph{restricted} pricing generator $L_{\kKs}$ is the principal submatrix of
$\Lp$ indexed by the corridor states.

\begin{proposition} \label{prop:corridor}
  Under Assumptions~\ref{ass:markov}--\ref{ass:continuity}, let $g\in\RR_+^\Xsf$ and
  fix $\Ccal\in\cC$.
  \begin{enumerate}
    \item If $\kK(\Ccal,g)=\varnothing$, then $q_t(x,g)=0$ for every $x\in\Ccal$ and
      $t\geq0$.
    \item If $\kK(\Ccal,g)\neq\varnothing$, then $\Ccal\in\kK(\Ccal,g)$,
      $\sS(g)\cap\kK(\Ccal,g)\neq\varnothing$, and $g_{\kKs}\neq0$. The reduced graph
      of $L_{\kKs}$ is the subgraph of the reduced pricing graph induced by
      $\kK(\Ccal,g)$, and
      \begin{equation} \label{eq:corridor-restriction}
        q_t(x,g)
        =
        \big( \e^{t \Lp} g \big)(x)
        =
        \big( \e^{t L_{\kKs}} g_{\kKs} \big)(x),
        \qquad x\in\Ccal, \ t\geq0.
      \end{equation}
      Moreover, \(q_t(x,g)>0\) for every \(x\in\mathcal C\) and \(t>0\).
  \end{enumerate}
\end{proposition}

Proposition~\ref{prop:corridor} distinguishes two cases. If the corridor is
empty, no state with a positive payoff is accessible from the current class,
so the price is zero at every maturity. If the corridor is nonempty, the
price is strictly positive at every positive maturity and can be computed
entirely from the restricted generator and payoff.
Every class path from $\Ccal$ to $\sS(g)$ lies entirely within the corridor, and the
restricted generator retains every such path. The restriction is therefore
exact at every maturity, not merely asymptotic.

\subsection{Corridor-specific price and yield asymptotics}

The exact corridor restriction allows us to apply Theorem~\ref{thm:Pi-asym-tied} to the
restricted pricing generator and thereby characterize the price of a given
maturity-$t$ payoff from a given initial state. The price's asymptotic exponential
rate and polynomial order are determined by the highest class pricing rate within the
corridor and the largest number of classes attaining that rate along a single
directed path within it. 

Specifically, for a nonempty pricing corridor $\kK(\Ccal,g)$, define
the \emph{corridor long-run pricing rate} (or simply \emph{corridor pricing rate}) by
\begin{equation} \label{eq:local-rate}
  \lambda(\Ccal,g)
  \coloneqq
  \max_{\Ccal_k\in\kK(\Ccal,g)}\lambda_k
  =
  s(L_{\kKs}).
\end{equation}
We call a class $\Ccal_k\in\kK(\Ccal,g)$ \emph{corridor-critical} when
$\lambda_k=\lambda(\Ccal,g)$. The \emph{critical depth} is the largest number of
corridor-critical classes that can be visited along a single directed path:
\begin{equation*} 
  \nu(\Ccal,g)
  \coloneqq
  \max_{\substack{\gamma\text{ is a chain of classes}\\
                   \gamma\subseteq\kK(\Ccal,g)}}
  \#\{\Ccal_k\in\gamma:\lambda_k=\lambda(\Ccal,g)\}.
\end{equation*}
In Figure~\ref{fig:vc}, the corridor from $\mathcal D$ satisfies
$\mathcal D\leadsto\mathcal R\leadsto\mathcal N$ and hence forms a chain of classes.
Its critical depth is one if only one of these classes is critical, two if two are
critical, and three if all three are critical.

The corridor-critical classes are precisely the globally dominant classes of the
restricted pricing generator $L_{\kKs}$. Theorem~\ref{thm:Pi-asym-tied}, applied to
$L_{\kKs}$, therefore identifies $\nu(\Ccal,g)$ as the size of its largest Jordan
block at $\lambda(\Ccal,g)$. To state the coefficient of the leading price term,
abbreviate $\lambda=\lambda(\Ccal,g)$ and $\nu=\nu(\Ccal,g)$, let $P_{\kKs}\coloneqq
P_\lambda(L_{\kKs})$ be the spectral projection associated with $\lambda$, and
define
\begin{equation} \label{eq:BK}
  B_{\kKs}
  \coloneqq
  [(L_{\kKs}-\lambda I)P_{\kKs}]^{\nu-1}P_{\kKs}/(\nu-1)!.
\end{equation}
An equivalent eigenvector representation of $B_{\kKs}$ is given in
Appendix~\ref{app:corridor-proofs}.

\begin{theorem} \label{thm:pyc}
  Under Assumptions~\ref{ass:markov}--\ref{ass:continuity}, let $g\in\RR_+^\Xsf$ be
  nonzero and fix $\Ccal\in\cC$. Suppose that $\kK(\Ccal,g)\neq\varnothing$.
  Then, for every $x\in\Ccal$,
  \begin{equation} \label{eq:price-clim}
    \lim_{t\to\infty} t^{1-\nu(\Ccal,\,g)}\e^{-\lambda(\Ccal,\,g)t}\,q_t(x,g)
    =
    (B_{\kKs} g_{\kKs})(x) \eqqcolon c_x(g) > 0.
  \end{equation}
  The corresponding yield satisfies
  \begin{equation} \label{eq:yield-clim}
    y_t(x,g)
    =
    -\lambda(\Ccal,g)
    -[\nu(\Ccal,g)-1]\, \frac{\log t}{t}
    -\frac{1}{t}\log c_x(g)
    +o(t^{-1}),
    \qquad \mbox{as } t\to\infty.
  \end{equation}
\end{theorem}

Unlike the global coefficient $(B^*g)(x)$, which may vanish, the corridor coefficient
$c_x(g) = (B_{\kKs} g_{\kKs})(x)$ in \eqref{eq:price-clim} is strictly positive. Thus
$\lambda(\Ccal,g)$ and $\nu(\Ccal,g)-1$ are the actual asymptotic exponential rate
and polynomial order of the price. For a fixed pricing generator, these depend only
on the initial class and payoff support: changing payoff magnitudes without changing
the support, or moving the initial state within its class, can change $c_x(g)$ but
not the rate or order.

In an irreducible economy, $\cC = \{\Xsf\}$, and every nonzero $g\in\RR^{\Xsf}_+$
satisfies $\lambda(\Xsf,g)=\lambda^*$, $\nu(\Xsf,g)=1$. Hence, for every initial state
$x\in\Xsf$,
\[
  q_t(x,g) \sim c_x(g)\, \e^{\lambda^* t},
  \qquad 
  y_t(x,g) = -\lambda^*-t^{-1}\log c_x(g) + o(t^{-1}), 
  \qquad 
  \text{as }t\to\infty.
\]
Reducibility allows not only the leading coefficient $c_x(g)$ but also the asymptotic
exponential rate and polynomial order to depend on the initial class and payoff
support.

If the corridor has a unique critical class, then $\nu(\Ccal,g)=1$, so the leading
price term is $c_x(g)\e^{\lambda(\Ccal,g)t}$, and \eqref{eq:BK}
gives $B_{\kKs} = P_{\kKs}$. Moreover, Theorem~\ref{thm:Pi-asym}, applied to
$L_{\kKs}$, implies that $\lambda(\Ccal,g)$ is an algebraically simple eigenvalue of
$L_{\kKs}$. Let $\phi_{\kKs}$ and $\psi_{\kKs}$ be the corresponding nonzero
nonnegative right and left eigenvectors, normalized by $\1^\t\phi_{\kKs}=1$ and
$\psi_{\kKs}^\t\phi_{\kKs}=1$. The spectral projection and leading price coefficient
therefore satisfy $P_{\kKs} = \phi_{\kKs}\psi_{\kKs}^\t$ and $c_x(g) =
\phi_{\kKs}(x)\,\psi_{\kKs}^\t g_{\kKs}>0$.
With multiple critical classes, $\nu(\Ccal,g)=1$ still holds if no directed path
within the corridor contains two of them. Otherwise, the leading price term includes
the polynomial factor $t^{\nu(\Ccal,g)-1}$.

\subsection{Perpetual values and state-dependent yields}

In a reducible economy, a perpetual cash-flow stream can have finite value from a
given initial state even when global valuation stability fails. The positive leading
coefficient in Theorem~\ref{thm:pyc} implies that, for a nonempty corridor, the price
is integrable over maturities exactly when the corridor pricing rate is negative.
Together with the zero-price conclusion of Proposition~\ref{prop:corridor} for an
empty corridor, this gives the following corollary.

\begin{corollary} \label{cor:corridor-stability}
  Under Assumptions~\ref{ass:markov}--\ref{ass:continuity}, let
  $g\in\RR_+^\Xsf$ and $x\in\Ccal$. If
  $\kK(\Ccal,g)=\varnothing$, then $(Vg)(x)=0$. Otherwise,
  $(Vg)(x)<\infty$ if and only if $\lambda(\Ccal,g)<0$.
\end{corollary}

For a nonempty corridor, finite value requires all class pricing rates within
the corridor to be negative, not all class rates in the economy. Classes with
nonnegative pricing rates therefore do not prevent finite value when they lie
outside the relevant corridor.

For $g=\1$, every class supports the payoff, so
$\kK(\Ccal,\1)=\{\Ccal_k\in\cC:\Ccal\leadsto\Ccal_k\}$ consists of all classes
accessible from $\Ccal$. Consequently, every modeled nonnegative perpetual cash-flow
stream has finite value from every state in an initial class $\Ccal$ if and only if
\[
  \max_{\Ccal_k:\,\Ccal\leadsto\Ccal_k}\lambda_k<0.
\]
Requiring this condition for every initial class recovers global valuation
stability. Applying the yield expansion \eqref{eq:yield-clim} in
Theorem~\ref{thm:pyc} to $g=\1$ gives the corresponding long-run zero-coupon yield:
\begin{corollary}[State-dependent long-run zero-coupon yields]
\label{cor:state-yields}
  Under Assumptions~\ref{ass:markov}--\ref{ass:continuity}, for every
  $\Ccal\in\cC$ and $x\in\Ccal$,
  \[
    \lim_{t\to\infty}y_t(x)
    =
    -\max_{\Ccal_k:\,\Ccal\leadsto\Ccal_k}\lambda_k.
  \]
\end{corollary}
Since each class is accessible from itself,
$\lambda^*=\max_{\Ccal\in\cC}\lambda(\Ccal,\1)$.
Thus the global long-run pricing rate is the maximum of the state-dependent
zero-coupon pricing rates, and the long-run envelope yield is the minimum of
their limiting yields.

\section{Disaster Recovery and Long-Run Pricing}
\label{sec:illustration}

We use disaster recovery to illustrate how a physically transient class can govern
long-horizon prices and how this possibility depends on preferences and cash-flow
exposure. Suppose the pricing graph has two communicating classes, a normal class
$\mathcal N$ and a disaster class $\mathcal D$, with the single interclass edge
$\mathcal D\to\mathcal N$. We call this a \emph{one-way disaster-recovery economy}.
The common physical and pricing accessibility relations imply that $\mathcal N$ is
physically recurrent and $\mathcal D$ is transient: recovery eventually occurs, and
the physical probability of disaster vanishes. Nevertheless, the disaster class may
have the higher long-run pricing rate.

Section~\ref{sec:owr} develops this structure allowing multiple states within each
class, specializing the general results of Sections~\ref{sec:network}
and~\ref{sec:corridors}. Section~\ref{sec:crra} then studies a two-state
consumption-based specification in which the pricing-rate gap is determined by risk
aversion, consumption dynamics, cash-flow exposure, and recovery intensity.

\subsection{One-way recovery and long-run prices}
\label{sec:owr}

Ordering normal states before disaster states gives the pricing generator
\begin{equation} \label{eq:Lp2}
  \Lp
  =
  \begin{pmatrix}
    \LpN&0\\ \LpDN&\LpD
  \end{pmatrix},
\end{equation}
where $\LpN$ and $\LpD$ are the irreducible Metzler blocks associated with the normal
and disaster classes, respectively, and $\LpDN\geq0$ is the nonzero recovery block.
Partition the payoff accordingly:
\begin{equation} \label{eq:g2}
  g
  =
  \begin{pmatrix}
    \gN \\ \gD
  \end{pmatrix},
\end{equation}
where $\gN$ and $\gD$ are the restrictions of $g$ to $\mathcal N$ and
$\mathcal D$. Define the class pricing rates 
\[
  \lambdaN\coloneqq s(\LpN), \qquad \lambdaD\coloneqq s(\LpD).
\]
Then $\lambda^*=s(\Lp)=\max\{\lambdaN,\lambdaD\}$. The next proposition gives the
exact pricing semigroup and the limiting strip yields implied by
Theorem~\ref{thm:pyc}.

\begin{proposition} \label{prop:owr}
  Under Assumptions~\ref{ass:markov}--\ref{ass:continuity}, suppose the pricing
  generator $\Lp$ has the block form \eqref{eq:Lp2}, where $\LpN$ and $\LpD$ are
  irreducible Metzler matrices and $\LpDN\geq0$ is nonzero.
  Then,
  \begin{equation} \label{eq:etLp}
    \Pi_t
    =
    \begin{pmatrix}
      \e^{t\LpN} & 0\\[3pt] H_t & \e^{t\LpD}
    \end{pmatrix},
    \qquad
    H_t
    =
    \int_0^t \e^{s\LpD} \LpDN\, \e^{(t-s)\LpN}\diff s,
    \qquad 
    t\geq0.
  \end{equation}
  Here $H_t$ is the recovery block. For every nonzero $g\in\RR_+^\Xsf$:
  \begin{enumerate}
    \item 
      If $x\in\mathcal N$ and $\gN\neq0$, then $\lim_{t\to\infty} y_t(x,g) =
      -\lambdaN$; if $\gN=0$, then $q_t(x,g)=0$ for $t\geq0$.
    \item 
      If $x\in\mathcal D$, then
      \[
        \lim_{t\to\infty}y_t(x,g)
        =
        \begin{cases}
          -\lambdaD, & \text{if }\gN=0,\\[2pt]
          -\max\{\lambdaN,\lambdaD\}, & \text{if } \gN\neq0.
        \end{cases}
      \]
  \end{enumerate}
\end{proposition}

The block representation of $\Pi_t$ in \eqref{eq:etLp}, together with \eqref{eq:g2},
implies that
\begin{equation} \label{eq:owrb}
  \Pi_t g = 
  \begin{pmatrix}
    \e^{t\LpN} \gN 
    \\[3pt]
    \e^{t\LpD} \gD + H_t \gN
  \end{pmatrix},
  \qquad t\geq0.
\end{equation}
From a normal initial state, the upper block in \eqref{eq:owrb} shows
that the price depends only on $\LpN$ and $\gN$ at every maturity. When $\gN\neq0$,
the pricing corridor is $\{\mathcal N\}$, so the disaster block $\LpD$ cannot affect
the price.%
\footnote{
  Allowing rare entry into disaster removes the block-triangular structure of
  $\Pi_t$, so prices from normal initial states also reflect paths through disaster.
  Their contribution vanishes as entry becomes rare at each fixed maturity, but
  Section~\ref{sec:near} shows that it can eventually dominate at long maturities
  when the disaster class pricing rate is higher.
}
For a disaster initial state $X_0 = x\in\mathcal D$, denote the recovery stopping
time by
\begin{equation}\label{eq:tau}
  \tau\coloneqq\inf\{s\geq0:X_s\in\mathcal N\}.
\end{equation}
Since the economy cannot return to disaster after recovery, the lower block in
\eqref{eq:owrb} decomposes the price $(\Pi_t g)(x)$ according to whether recovery has
occurred by maturity. In particular,
\[
  (\e^{t\LpD}\gD)(x)
  = 
  \EE_x \, S_tG_tg(X_t)1_{\{\tau>t\}},
  \qquad
  (H_t\gN)(x)
  =
  \EE_x \, S_tG_tg(X_t)1_{\{\tau\leq t\}}.
\]
These two terms average the same discounted payoff over sample paths \emph{without} and
\emph{with recovery} by maturity, respectively. In the integral representation of
$H_t$ in \eqref{eq:etLp}, $s$ is a possible recovery time, leaving $t-s$ units of
time in the normal class before maturity.

A case of particular interest is when the disaster class has the higher pricing rate:
\begin{equation} \label{eq:dd}
  \lambdaD > \lambdaN.
\end{equation}
Let $\phiN,\psiN\gg0$ be right and left eigenvectors of $\LpN$ associated with
$\lambdaN$, and let $\phiD,\psiD\gg0$ be the corresponding eigenvectors of $\LpD$
associated with $\lambdaD$, normalized by $\1^\t\phiN=\1^\t\phiD=1$ and
$\psiN^\t\phiN=\psiD^\t\phiD=1$. The following corollary gives the long-run
asymptotics of strip prices for both initial classes.

\begin{corollary} \label{cor:owr}
  Under the hypotheses of Proposition~\ref{prop:owr} and \eqref{eq:dd}, let
  $g\in\RR_+^\Xsf$ be nonzero.
  \begin{enumerate}
    \item
      For every $x\in\mathcal N$,
      \[
        \lim_{t\to\infty}\e^{-\lambdaN t}q_t(x,g) 
        =
        \phiN(x)\,\psiN^\t\gN,
      \]
      where the limit is positive for $\gN\neq0$.
    \item
      For every $x\in\mathcal D$,
      \begin{equation} \label{eq:dlim}
        \lim_{t\to\infty}\e^{-\lambdaD t}q_t(x,g)
        =
        \phiD(x)\,\psiD^\t
        \left[
          \gD
          +
          \int_0^\infty \e^{-\lambdaD v}\LpDN\e^{v\LpN}\gN\diff v
        \right]>0.
      \end{equation}
  \end{enumerate}
\end{corollary}

For a normal initial state and $\gN\neq0$, part~(i) of the corollary gives
$q_t(x,g)\sim\e^{\lambdaN t}\phiN(x)\psiN^\t\gN$, with $\lambdaN<\lambdaD$. For a
disaster initial state, the two terms in brackets in \eqref{eq:dlim}, premultiplied by
$\e^{\lambdaD t}\phiD(x)\psiD^\t$, give the leading contributions to the price from
sample paths without and with recovery, respectively. The recovery contribution is
positive whenever $\gN\neq0$, so even a payoff supported entirely on normal states
inherits the disaster pricing rate.

When $\gN\neq0$, the integrand in \eqref{eq:dlim} decays exponentially at rate
$\lambdaD-\lambdaN>0$ as $v\to\infty$, where $v$ represents the time between recovery
and maturity. The asymptotics imply that, at sufficiently long
maturities, most of the recovery contribution to the price comes from sample paths
that remain in disaster for almost the entire horizon and recover near maturity.
The pricing importance of these sample paths, however, does not imply that such late
recovery is physically likely.

\subsection{Consumption risk and pricing dominance}
\label{sec:crra}

We now specify economic primitives and characterize the conditions under which
\eqref{eq:dd} holds. Specialize the one-way recovery economy to two states,
$\Xsf=\{N,D\}$, with $\mathcal N=\{N\}$ and $\mathcal D=\{D\}$.
Recovery occurs at physical intensity $\kappa>0$, and the normal state is absorbing,
so the physical generator is
\[
  \Lx =
  \begin{pmatrix}
    0 & 0 \\ \kappa & -\kappa
  \end{pmatrix}.
\]
Conditional on $X_0=D$, the recovery stopping time $\tau$ defined in \eqref{eq:tau}
is exponentially distributed with mean $1/\kappa$, with density $f_\tau(s) \coloneqq
\kappa \e^{-\kappa s}$ for $s>0$. Therefore, a smaller $\kappa$ means a longer
expected disaster duration and hence higher disaster persistence.

\subsubsection{A consumption-based pricing specification}

We consider a consumption-based asset-pricing model with time-separable CRRA
preferences, as in \citet[Example~3.8]{hansen2009long}. Here, the state process $X$
is the two-state continuous-time Markov chain with physical generator $\Lx$ specified
above.

Let $C_t>0$ denote aggregate consumption at time $t$, with $C_0>0$. Between
state transitions, log consumption has drift $\mu_x\in\RR$ and volatility
$\sigma_x\geq0$ in state $x\in\{N,D\}$. At recovery, consumption is multiplied by
$\chi_{\sss DN}>0$, so log consumption jumps by $\log\chi_{\sss DN}$. Specifically, we
assume
\begin{equation}\label{eq:logc}
  \diff\log C_t
  =
  \mu_{X_{t-}}\diff t
  +\sigma_{X_{t-}}\diff W_t
  +\log\chi_{\sss DN}\,\diff N_t^{\Xs},
\end{equation}
where $W$ is a standard Brownian motion independent of $X$, and $N_t^{\Xs}$
counts the jumps of $X$ over $(0,t]$. Since recovery is the only possible
transition, $N_t^{\Xs}=\1_{\{\tau\leq t\}}$ when $X_0=D$.

Let $\alpha>0$ be the relative risk-aversion parameter of the CRRA utility function
$u$, and let $\beta>0$ be the subjective discount rate. Since $u'(c)=c^{-\alpha}$,
the SDF in this model takes the form
\begin{equation}\label{eq:sdf}
  S_t
  =\e^{-\beta t}\frac{u'(C_t)}{u'(C_0)}
  =\e^{-\beta t}(C_t/C_0)^{-\alpha}.
\end{equation}
We assume that cumulative cash-flow growth is a power of normalized consumption:
\begin{equation}\label{eq:gr}
  G_t=(C_t/C_0)^{\vartheta},\qquad\vartheta\geq0,
\end{equation}
where $\vartheta$ is the elasticity of cash-flow growth with respect to consumption.%
\footnote{
  Appendix~\ref{app:separate-exposures} relaxes this power specification by allowing
  cash-flow growth to have its own drift, Brownian exposure, and transition
  multipliers.
}
For $g=\1$, $\vartheta=0$ gives a fixed real payment and $\vartheta=1$ gives a
payment proportional to aggregate consumption. 

\subsubsection{The pricing generator and pricing-rate gap}

By \eqref{eq:logc}--\eqref{eq:gr}, we have
\[
  \diff\log(S_tG_t)
  =
  (-\beta-(\alpha-\vartheta)\mu_{X_{t-}})\diff t
  -(\alpha-\vartheta)\sigma_{X_{t-}}\diff W_t
  -(\alpha-\vartheta)\log\chi_{\sss DN}\,\diff N_t^{\Xs}.
\]
At recovery, $\log S_tG_t$ jump by $(\vartheta-\alpha)\log\chi_{\sss DN}$, so the
growth-adjusted SDF is multiplied by $\chi_{\sss DN}^{\vartheta-\alpha}$.
Proposition~\ref{prop:param-generator} then gives the pricing generator
\begin{equation}\label{eq:Lp-crra}
  \Lp 
  =
  \begin{pmatrix}
    -\varrho_{\sss N} & 0
    \\[3pt]
    \kappa\chi_{\sss DN}^{\vartheta-\alpha} & -\varrho_{\sss D} - \kappa
  \end{pmatrix},
\end{equation}
where  
\[
  \varrho_x\coloneqq\beta+(\alpha-\vartheta)\mu_x-\frac12\,(\alpha-\vartheta)^2\sigma_x^2,
  \qquad x\in\{N,D\}.
\]

Conditional on remaining in the initial state $x$, we have
$\EE_x(S_tG_t\mid N_t^{\Xs}=0)=\exp(-\varrho_x t)$.%
\footnote{
  On $\{N_t^{\Xs}=0\}$, starting from $x$,
  \(
    S_tG_t
    =
    \exp\{(-\beta-(\alpha-\vartheta)\mu_x)t-(\alpha-\vartheta)\sigma_xW_t\}.
  \)
  Independence of $W$ and $X$ yields $W_t\sim N(0,t)$ under this conditioning. The
  Gaussian exponential-moment formula then gives 
  \(
    \EE_x(S_tG_t\mid N_t^{\Xs}=0) 
    =
    \exp\{(-\beta-(\alpha-\vartheta)\mu_x+\tfrac12(\alpha-\vartheta)^2\sigma_x^2)t\}
    =
    \exp(-\varrho_x t).
  \)
  At fixed log-consumption drift, greater volatility raises this expectation
  through the convexity of the exponential, lowering the corresponding discount
  rate when $\alpha\neq\vartheta$.
}
Thus $\varrho_x$ is the within-state growth-adjusted discount rate, excluding
transition effects. Its adjustment to subjective discounting $\beta$ depends on
$\alpha-\vartheta$, the excess of marginal-utility sensitivity over cash-flow
sensitivity to consumption. When $\alpha>\vartheta$, higher expected log-consumption
growth raises $\varrho_x$, whereas greater volatility lowers it, holding other
parameters fixed.

If $\alpha>\vartheta$, $\mu_{\sss N}>\mu_{\sss D}$, and $\sigma_{\sss
D}^2>\sigma_{\sss N}^2$, then $\varrho_{\sss D}<\varrho_{\sss N}$.
The difference $\varrho_{\sss N}-\varrho_{\sss D}$ measures how much lower
the within-state discount rate is in disaster than in the normal state, before accounting
for recovery. In particular,
\begin{equation}\label{eq:rnd}
  \varrho_{\sss N}-\varrho_{\sss D}
  =(\alpha-\vartheta)(\mu_{\sss N}-\mu_{\sss D})
   +\frac12\,(\alpha-\vartheta)^2(\sigma_{\sss D}^2-\sigma_{\sss N}^2).
\end{equation}
Under the above conditions, greater risk aversion (i.e., a larger $\alpha$) widens the
discount-rate difference $\varrho_{\sss N}-\varrho_{\sss D}$, whereas greater
cash-flow sensitivity to consumption (i.e., a larger $\vartheta$) narrows it.%
\footnote{
  Under these conditions,
  $\partial(\varrho_{\sss N}-\varrho_{\sss D})/\partial\alpha
  =-\partial(\varrho_{\sss N}-\varrho_{\sss D})/\partial\vartheta
  =(\mu_{\sss N}-\mu_{\sss D})
   +(\alpha-\vartheta)(\sigma_{\sss D}^2-\sigma_{\sss N}^2)>0$.
 }

Since each class contains a single state, its pricing rate is the corresponding
diagonal entry of $\Lp$:
\begin{equation}\label{eq:pr-crra}
  \lambdaN = -\varrho_{\sss N},\qquad
  \lambdaD = -\varrho_{\sss D}-\kappa.
\end{equation}
Hence the pricing-rate gap is
\begin{equation}\label{eq:crra-gap}
  \Delta
  \coloneqq
  \lambdaD-\lambdaN 
  =
  \varrho_{\sss N}-\varrho_{\sss D}-\kappa.
\end{equation}
Disaster is pricing dominant exactly when its within-state discount rate is lower
than the normal-state rate by more than the recovery intensity:
$\varrho_N-\varrho_D>\kappa$. Holding other parameters fixed, greater disaster
persistence---a smaller $\kappa$---raises $\Delta$. 

\subsubsection{Exact prices and recovery timing}
\label{sec:crra-recovery}

Applying Proposition~\ref{prop:owr} to the pricing generator in \eqref{eq:Lp-crra}
gives the pricing semigroup
\[
  \Pi_t 
  =
  \begin{pmatrix}
    \e^{t\lambdaN} & 0\\[3pt] H_t & \e^{t\lambdaD}
  \end{pmatrix},
  \qquad
  t\geq0,
\]
where $\lambdaN$ and $\lambdaD$ are given by \eqref{eq:pr-crra}, and
\begin{equation}\label{eq:Hh}
  H_t = \int_0^t h_t(s)\diff s, 
  \qquad
  h_t(s)
  \coloneqq
  \kappa\chi_{\sss DN}^{\vartheta-\alpha} \,\e^{\lambdaD s}\e^{\lambdaN(t-s)} > 0,
  \quad 0<s<t.
\end{equation}
As in Section~\ref{sec:owr}, here $s$ represents a possible value of the recovery
stopping time $\tau$, and $t-s$ the time spent in the normal state before maturity.

For a payoff $g$, write $g_N = g(N)$ and $g_D = g(D)$. The corresponding maturity-$t$
strip prices are 
\[
  q_t(N,g)
  =
  \e^{\lambdaN t}g_N,
  \qquad
  q_t(D,g)
  =
  \e^{\lambdaD t}g_D + H_t g_N.
\]
To make the late-recovery interpretation at the end of Section~\ref{sec:owr}
explicit, assume $g_N>0$ and consider the \emph{recovery price-contribution} $H_tg_N$ to
$q_t(D,g)$. By \eqref{eq:Hh}, $g_Nh_t(s)$ is the \emph{density} of this recovery
contribution: integrating it over any subinterval of $(0,t)$ gives the price
contribution from paths that recover during that interval.

Using \eqref{eq:pr-crra}, we can factor $h_t(s)$ defined in \eqref{eq:Hh} as
\begin{equation}\label{eq:h-factor}
  h_t(s)
  =
  f_\tau(s)\, \EE_{\sss D}(S_tG_t\mid\tau=s),
  \qquad 
  0 < s < t,
\end{equation}
where $f_\tau(s) = \kappa \e^{-\kappa s}$ is the physical density of the recovery
time $\tau$ under $\PP_{\sss D}$, and 
\[
  \EE_{\sss D}(S_tG_t\mid\tau=s) 
  =
  \chi_{\sss DN}^{\vartheta-\alpha} \e^{-\varrho_{\sss D}s-\varrho_{\sss N}(t-s)}
\]
is the expected growth-adjusted SDF conditional on recovery at time $s\in(0,t)$.
Thus the second factor in \eqref{eq:h-factor} weights the physical recovery-time
density to obtain the recovery price-contribution density per unit of normal-state
payoff. Since
\[
  \frac{\diff}{\diff s} \log f_\tau(s) = -\kappa,
  \qquad
  \frac{\diff}{\diff s} \log \EE_{\sss D}(S_tG_t\mid\tau=s) 
  = 
  \varrho_{\sss N} - \varrho_{\sss D},
\]
we obtain
\[
  \frac{\diff}{\diff s}\log h_t(s)
  =\varrho_{\sss N}-\varrho_{\sss D}-\kappa
  =\Delta.
\]
Therefore, $h_t(s)$ is strictly increasing in $s$ exactly when $\Delta>0$, or,
equivalently, when the conditional expected growth-adjusted SDF increases with
recovery time at a faster exponential rate than the physical density decreases.

To describe how the recovery contribution is distributed across recovery times,
define the normalized recovery price-contribution density 
\begin{equation}\label{eq:h-til}
  \tilde h_t(s)
  \coloneqq
  \frac{h_t(s)}{H_t}
  =
  \frac{h_t(s)}{\int_0^t h_t(r)\diff r}
  ,
  \qquad 0<s<t.
\end{equation}
For comparison, conditioning the recovery time $\tau$ on recovery by maturity $t$
yields the density
\[
  f_\tau(s\mid\tau\leq t)
  =\frac{f_\tau(s)}{\PP_{\sss D}\{\tau\leq t\}}
  =\frac{\kappa\e^{-\kappa s}}{1-\e^{-\kappa t}},
  \qquad 0<s<t.
\]
Both densities integrate to one over $(0,t)$, but they have distinct interpretations.

\begin{figure}[!t]
  \centering
  \includegraphics[width=0.9\textwidth]{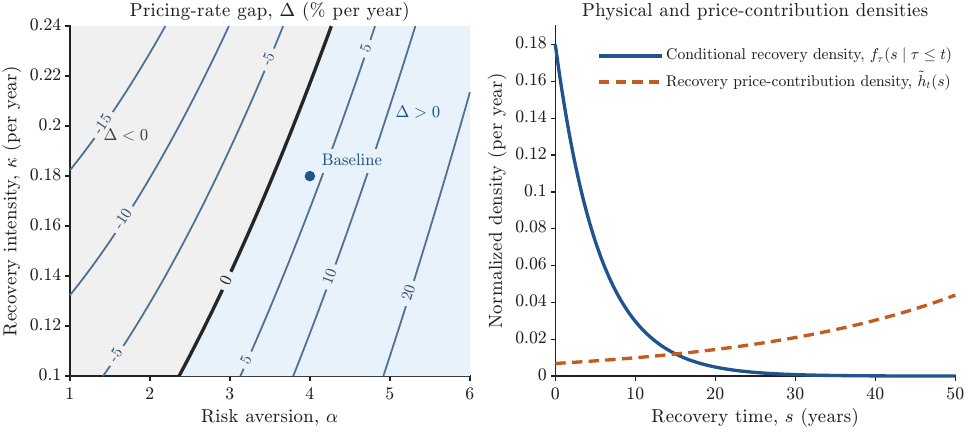}
  \caption{
    Pricing-rate gap and physical versus price-contribution densities.
    Parameters follow \eqref{eq:crra-benchmark} with $\vartheta=0$.
    Left: contours of $\Delta=\lambdaD-\lambdaN$ (\% per year), varying
    $\alpha\in[1,6]$ and $\kappa\in[0.10,0.24]$ per year. The bold $\Delta=0$
    boundary separates disaster (blue) from normal-state (gray) pricing dominance;
    the dot marks $(\alpha,\kappa)=(4,0.18)$.
    Right: the conditional physical density $f_\tau(s\mid\tau\leq t)$ (blue)
    and normalized recovery price-contribution density $\tilde h_t(s)=h_t(s)/H_t$
    (orange dashed), for $X_0=D$ and $t=50$ years at the baseline.
  }
  \label{fig:dp}
\end{figure}

When $\Delta>0$, \eqref{eq:Hh} and \eqref{eq:h-til} give
\(
  \tilde h_t(s)=\Delta\e^{\Delta s}/(\e^{\Delta t}-1)
\)
for $s\in(0,t)$. For fixed $\Delta>0$ and $v>0$, and $t>v$, the share of
the recovery contribution attributable to recovery during $[t-v,t]$ is
\[
  \int_{t-v}^t \tilde h_t(s)\diff s
  =\frac{1-\e^{-\Delta v}}{1-\e^{-\Delta t}}
  \to
  1-\e^{-\Delta v} > 0,
  \qquad \mbox{as } t\to\infty.
\]
Thus, for any $\varepsilon\in(0,1)$, choosing
$v\geq\Delta^{-1}\log(1/\varepsilon)$ ensures that recovery during $[t-v,t]$
accounts for at least $1-\varepsilon$ of the recovery price-contribution.
For any such fixed $v$, however, 
\[
  \int_{t-v}^t f_\tau(s\mid\tau\leq t)\diff s
  =
  \frac{\e^{\kappa v}-1}{\e^{\kappa t}-1}
  \to
  0
  \qquad \mbox{as } t\to\infty.
\]
This quantifies the distinction between pricing importance and physical likelihood.

\subsubsection{A quantitative illustration}

We measure time in years and set 
\begin{equation}\label{eq:crra-benchmark}
  \beta=0.02,\quad
  \mu_{\sss N} = 0.025,\quad\mu_{\sss D} = 0,
  \quad
  \sigma_{\sss N} = 0.02,\quad \sigma_{\sss D} = \sqrt{0.02^2+0.121^2},
  \quad
  \kappa=0.18.
\end{equation}
These inputs are illustrative, not a calibration. The disaster-growth,
excess-volatility, and recovery magnitudes are informed by
\citet[Table~1]{nakamura2013crises}.%
\footnote{The growth and excess-volatility inputs draw on their permanent
disaster-shock estimates. These are not simply contemporaneous consumption-growth
shocks, and their benchmark preferences are Epstein--Zin rather than CRRA. Thus
\eqref{eq:crra-benchmark} is not a replication of their consumption process or
equilibrium prices.}

The recovery intensity $\kappa$ has a direct duration interpretation: the mean
disaster spell is $1/\kappa$, and its probability of lasting more than one year is
$\e^{-\kappa}$. Matching the annual disaster-continuation probability of $0.835$
reported by \citet[Table~1]{nakamura2013crises} gives
$\kappa=-\log(0.835)\approx0.1803$. Our baseline therefore implies a mean disaster
spell of $5.56$ years, close in magnitude to the roughly six-year mean in their
annual model. The range $\kappa\in[0.10,0.24]$ in Figure~\ref{fig:dp} spans mean
durations of approximately 4.2--10 years. Recovery restores normal consumption-growth
dynamics but need not return consumption to its pre-disaster level.

Under \eqref{eq:crra-benchmark}, \eqref{eq:rnd} and \eqref{eq:crra-gap} imply that
\(
  \Delta=0.025\,(\alpha-\vartheta) + 0.00732\,(\alpha-\vartheta)^2 - 0.18.
\)
For $\alpha\geq\vartheta$, disaster dominance requires $\alpha-\vartheta>3.5369$. At
$\alpha=4$, a fixed payoff ($\vartheta=0$) gives $\lambdaN=-0.1168$,
$\lambdaD=-0.0797$, and $\Delta=0.0371$.%
\footnote{
  The implied normal-state real yield, $-\lambdaN=11.68\%$, belongs to the
  counterfactual economy with disaster entry ruled out. It is not a calibrated
  prediction for yields observed during normal times, when investors still price the
  possibility of future disasters.
}
By contrast, consumption-proportional cash flows
($\vartheta=1$) give $\Delta=-0.0391$; so the same physical economy and preferences
do not imply disaster dominance for every cash-flow family.

Figure~\ref{fig:dp} fixes $\vartheta=0$. The left panel varies risk aversion $\alpha$
and recovery intensity $\kappa$, holding the remaining parameters at the values in
\eqref{eq:crra-benchmark}. The pricing-rate gap $\Delta$ increases with $\alpha$ and
decreases with $\kappa$, so greater risk aversion and more persistent disasters favor
disaster dominance. The blue and gray regions correspond to disaster and normal-state
dominance, respectively. The upward-sloping zero-gap boundary shows that faster
recovery requires greater risk aversion for disaster dominance.

The right panel of Figure~\ref{fig:dp} considers the baseline setting
$(\alpha,\kappa)=(4,0.18)$ and maturity $t=50$ years, starting from disaster. The
conditional physical density places most probability on recovery within the first ten
years, whereas the normalized recovery price-contribution density increases toward
maturity. For example, recovery during the final ten years accounts for approximately
$36.8\%$ of the recovery price-contribution, but its physical probability conditional
on recovery by maturity is only $0.062\%$. Thus physically unlikely late recovery can
carry substantial pricing importance.

\section{Rare Transitions and Crossover Maturities}
\label{sec:near}

The discussion following \eqref{eq:owrb} in Section~\ref{sec:owr} shows
that the disaster class can have the highest pricing rate yet remain absent
from prices starting in the normal class. We now study how rare transitions
bring an otherwise excluded pricing-dominant class into the pricing corridor.
We characterize how the rarity of the connecting paths determines its
vanishing price loading and how this loading, together with the pricing-rate
gap, determines the crossover to long-maturity dominance.

\subsection{Path rarity and dominant-component loadings}
\label{sec:rare-entry}

Consider Metzler pricing generators $\Lp(\epsilon)$, $\epsilon\geq0$, such that
$\Lp(\epsilon)$ is irreducible for every sufficiently small $\epsilon>0$ and
converges to a reducible generator $\Lp(0)$ as $\epsilon\downarrow0$. For
$x\in\Xsf$ and $g\in\RR^\Xsf$, write
\[
  \Pi_t(\epsilon)\coloneqq\e^{t\Lp(\epsilon)},
  \qquad
  q_t(x,g;\epsilon)\coloneqq(\Pi_t(\epsilon)g)(x),
  \qquad t\geq0.
\]
Throughout Section~\ref{sec:near}, we write $\Lp\coloneqq\Lp(0)$ for the limiting
pricing generator, and maintain Assumptions~\ref{ass:markov}--\ref{ass:continuity}
for each sufficiently small $\epsilon\geq0$, on the same finite state space $\Xsf$.

Let $\cC=\{\Ccal_1,\ldots,\Ccal_m\}$, $m\geq2$, be the class partition of the
\emph{limiting} graph $\Gcal(\Lp(0))$. We retain these \emph{limiting classes} as
a fixed partition of $\Xsf$ for $\epsilon>0$, even though irreducibility makes
$\Gcal(\Lp(\epsilon))$ a single class. Partition
\[
  \Lp(\epsilon)=\bigl[L_{k\ell}(\epsilon)\bigr],
  \qquad k,\ell = 1,\ldots,m,
\]
conformably with \eqref{eq:frobenius}. Thus $L_{k\ell}(\epsilon)$ has rows in
$\Ccal_k$ and columns in $\Ccal_\ell$; when $k\neq\ell$, a nonzero block represents
a direct edge $\Ccal_k\to\Ccal_\ell$ in $\Gcal(\Lp(\epsilon))$.

Set $L_k\coloneqq L_{kk}(0)$ and $\lambda_k\coloneqq s(L_k)$. Assume that the
limiting network has a unique globally pricing-dominant class $\Ccal^*$:
\begin{equation}\label{eq:gvd}
  \lambda^* \coloneqq s(\Lp(0)) =s(L_{k^*}) > \max_{\Ccal_k\neq\Ccal^*}\lambda_k,
\end{equation}
where $L_{k^*}$ is the block indexed by $\Ccal^*$.

The next assumption is standard in singular perturbations of nearly decomposable
systems; see, for example,
\citet{schweitzer1968perturbation,meyer1989stochastic,YinZhang2012} and
\citet{avrachenkov2013analytic}.

\begin{assumption}[Rare-edge perturbation] \label{ass:edge}
  For every sufficiently small $\epsilon>0$, $\Lp(\epsilon)$ is Metzler and
  irreducible. There is a positive function $\eta(\epsilon)\to0$ such that
  \(
    L_{kk}(\epsilon) = L_k+O(\eta(\epsilon)),
  \)
  for every $k$. For each $k\neq\ell$, there is an edge-rarity exponent
  $w_{k\ell}\in[0,+\infty]$ such that, as $\epsilon\downarrow0$,
  \[
    L_{k\ell}(\epsilon)
    =
    \begin{cases}
      \epsilon^{w_{k\ell}}
      \bigl(K_{k\ell}+O(\eta(\epsilon))\bigr),
      & w_{k\ell}<+\infty,\\
      0, & w_{k\ell}=+\infty,
    \end{cases}
  \]
  where, in the first case, $K_{k\ell}\geq0$ and $K_{k\ell}\neq0$; in the
  second, the equality holds for every $\epsilon\geq0$.
\end{assumption}

Assumption~\ref{ass:edge} is inherited from a regular perturbation of the physical
Markov dynamics. For example, if $r(\epsilon)\to0$, $\bar Q_{k\ell}\geq0$ is
nonzero, $\bar\Gamma_{k\ell}\gg0$, and, for $k\neq\ell$ with
$w_{k\ell}<+\infty$,
\[
  [\Lx(\epsilon)]_{k\ell}
  =\epsilon^{w_{k\ell}}
  \bigl(\bar Q_{k\ell}+O(r(\epsilon))\bigr),
  \qquad
  \Gamma_{k\ell}(\epsilon)
  =\bar\Gamma_{k\ell}+O(r(\epsilon)),
\]
then Proposition~\ref{prop:param-generator} gives
\[
  L_{k\ell}(\epsilon)
  =[\Lx(\epsilon)]_{k\ell}\odot\Gamma_{k\ell}(\epsilon)
  =\epsilon^{w_{k\ell}}
  \bigl(\bar Q_{k\ell}\odot\bar\Gamma_{k\ell}+O(r(\epsilon))\bigr).
\]
Thus regular risk adjustment changes the leading coefficient but preserves the edge
exponent and zero pattern. If within-class errors are uniformly $O(r(\epsilon))$,
they and the diagonal effects of new exit intensities can be absorbed by taking, for
example,
\(
  \eta(\epsilon) = r(\epsilon)+\epsilon^{w_{\min}},
\)
where $w_{\min}$ is the smallest positive edge exponent. Positive limiting jump
multipliers imply that physical irreducibility yields pricing irreducibility;
singular or vanishing risk adjustments may instead change the effective exponents.

Define the weighted directed graph $\Gcal_w$ on $\cC$ by including the edge
$\Ccal_k\to\Ccal_\ell$ exactly when $w_{k\ell}<+\infty$ and assigning it weight
$w_{k\ell}$. It records which interclass edges can open and their asymptotic
orders. An edge of weight zero is \emph{non-rare} and is already present at the
limit; one of positive finite weight is \emph{rare} and opens at order
$\epsilon^{w_{k\ell}}$. An infinite weight means that the direct edge remains
absent, although an indirect path may exist. The zero-weight subgraph of $\Gcal_w$
is precisely the reduced pricing graph $\Gcal_{\mathrm r}(\Lp(0))$, while
irreducibility for $\epsilon>0$ makes $\Gcal_w$ strongly connected.

For $\Ccal,\Ccal'\in\cC$, let $\dD_w(\Ccal,\Ccal')$ be the set of directed paths
from $\Ccal$ to $\Ccal'$ in $\Gcal_w$.
Here a path $\gamma=(\Ccal_{i_0},\ldots,\Ccal_{i_p})$ is a sequence of distinct
vertices with $\Ccal_{i_0}=\Ccal$, $\Ccal_{i_p}=\Ccal'$, and a directed edge from
each vertex to the next.
Its weight is
\(
  W(\gamma) \coloneqq \sum_{r=1}^p w_{i_{r-1}i_r}.
\)
When $\Ccal=\Ccal'$, include the length-zero path with weight zero. Define
\begin{equation} \label{eq:rdm}
  d_{\scriptscriptstyle\Ccal\Ccal^*} \coloneqq \min_{\gamma\in\dD_w(\Ccal,\,\Ccal^*)}W(\gamma),
  \qquad
  d_{\scriptscriptstyle\Ccal^*\Ccal} \coloneqq \min_{\gamma\in\dD_w(\Ccal^*,\,\Ccal)}W(\gamma).
\end{equation}
The subscripts record direction, and $d_{\sss \Ccal^*\Ccal^*}=0$.

\begin{figure}[t!]
  \centering
  \begin{tikzpicture}
    \node[classnode] (C) at (0,0) {Distress class\\$\Ccal$};
    \node[classnode,very thick,fill=black!10] (S) at (4,1.65)
      {Dominant crisis class\\$\Ccal^*\in\sS(g)$};
    \node[classnode] (G) at (8,0)
      {Recovery class\\$\Ccal'\in\sS(g)$};
    \draw[flow] (C) -- node[below,font=\scriptsize] {$w=0$} (G);
    \draw[rareflow] (C) -- node[above left,font=\scriptsize] {$w=2$} (S);
    \draw[flow] (S) -- node[above right,font=\scriptsize] {$w=0$} (G);
    \draw[rareflow] (G.south) to[bend left=35]
      node[below,font=\scriptsize] {$w=1$} (C.south);
  \end{tikzpicture}
  \caption{A stylized distress--crisis--recovery weighted class graph. The payoff is
  supported on both the globally dominant crisis class $\Ccal^*$ and the recovery
  class $\Ccal'$. Solid and dashed arrows denote zero- and positive-weight edges,
  respectively. The rare relapse edge makes $\Gcal_w$ strongly connected.}
  \label{fig:weighted-rare-graph}
\end{figure}
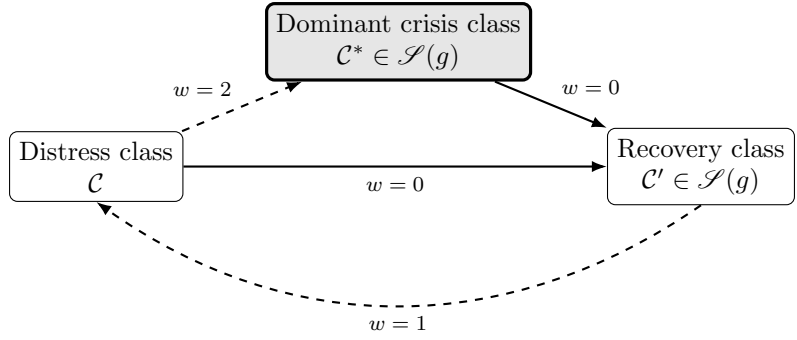

For a nonzero $g\geq0$ with class-level support $\sS(g)$, define
\begin{equation} \label{eq:rare-price-distance}
  d(\Ccal,g)
  \coloneqq
  d_{\sss \Ccal\Ccal^*} + \min_{\Ccal'\in\sS(g)}d_{\sss \Ccal^*\Ccal'}.
\end{equation}
It is the minimum combined rarity exponent of two directed paths: one from the
initial class $\Ccal$ to $\Ccal^*$ and the other from $\Ccal^*$ to the payoff support.
Since zero-weight paths are exactly the paths in
the reduced graph of $\Lp(0)$,
\begin{equation} \label{eq:rare-corridor-equivalence}
  d(\Ccal,g)=0
  \quad\Longleftrightarrow\quad
  \Ccal^*\in\kK(\Ccal,g),
\end{equation}
where $\kK(\Ccal,g)$ is the limiting pricing corridor. Thus $d(\Ccal,g)>0$
exactly when at least one path in every such pair contains a rare edge.

Figure~\ref{fig:weighted-rare-graph} may be read as a stylized
distress--crisis--recovery network. Let $\sS(g)=\{\Ccal^*,\Ccal'\}$, so the payoff
is supported in both the globally dominant crisis class and the recovery class.
Starting from the distress class $\Ccal$, direct recovery is order one, whereas
escalation to $\Ccal^*$ occurs at order $\epsilon^2$. Recovery from crisis is again
order one, while relapse from $\Ccal'$ to $\Ccal$ occurs at order $\epsilon$.
This economic reading is illustrative; the results below require only a rarely
accessed globally dominant class, not a crisis interpretation.

The least-weight initial-state leg has weight $2$. Because
$\Ccal^*\in\sS(g)$, the length-zero path at $\Ccal^*$ makes the least-weight payoff
leg equal to zero; the solid recovery edge also gives
$d_{\scriptscriptstyle\Ccal^*\Ccal'}=0$. Hence
\[
  d_{\scriptscriptstyle\Ccal\Ccal^*}=2,
  \qquad
  \min_{\Ccal''\in\sS(g)}
  d_{\scriptscriptstyle\Ccal^*\Ccal''}=0,
  \qquad
  d(\Ccal,g)=2.
\]
The zero-weight graph gives the limiting corridor
$\kK(\Ccal,g)=\{\Ccal,\Ccal'\}$, which excludes $\Ccal^*$. Thus $\epsilon^2$ is the
order at which the dominant component enters the price, even though direct recovery
to a payoff class remains order one. The positive rarity exponent comes entirely
from the initial-state leg; neither singleton payoff support nor a rare payoff-side
edge is required.

For sufficiently small $\epsilon>0$, let $\lambda^*(\epsilon)=s(\Lp(\epsilon))$
and let $P_*(\epsilon)$ be its rank-one Riesz spectral projection. Write
$P_{*,\s\Ccal\Ccal'}(\epsilon)$ for the block with rows in $\Ccal$ and columns in
$\Ccal'$, and define the long-run pricing law by
\(
  \pi_\epsilon^{\mathrm{LR}}(x)\coloneqq P_*(\epsilon)(x,x).
\)

\begin{theorem}[Perron projection and pricing-mass asymptotics]
  \label{thm:rare-distance}
  Under Assumption~\ref{ass:edge} and the unique-dominant-class condition
  \eqref{eq:gvd}, as $\epsilon\downarrow0$,
  $\lambda^*(\epsilon)\to\lambda^*$ and, for every $\Ccal,\Ccal'\in\cC$, there is a
  matrix $\bar P_{\s\Ccal\Ccal'}\gg0$ such that
  \begin{equation} \label{eq:rp}
    \epsilon^{-\left(
      d_{\s\Ccal\Ccal^*}+d_{\s\Ccal^*\Ccal'}
    \right)} 
    P_{*,\s\Ccal\Ccal'}(\epsilon)
    \to \bar P_{\s\Ccal\Ccal'}.
  \end{equation}
  Moreover,
  \begin{equation} \label{eq:rv}
    \epsilon^{-\big(d_{\s\Ccal\Ccal^*}+d_{\s\Ccal^*\Ccal}\big)}
    \sum_{x\in\Ccal} \pi_\epsilon^{\mathrm{LR}}(x)
    \to
    \operatorname{tr}\bar P_{\s\Ccal\Ccal}>0.
  \end{equation}
  In particular, $\sum_{x\in\Ccal^*}\pi_\epsilon^{\mathrm{LR}}(x)\to1$ and
  $\sum_{x\in\Ccal}\pi_\epsilon^{\mathrm{LR}}(x)\to0$ for every $\Ccal\neq\Ccal^*$.
\end{theorem}

For fixed sufficiently small $\epsilon>0$, Perron--Frobenius theory gives
\(
  \e^{-\lambda^*(\epsilon)t}\Pi_t(\epsilon)\to P_*(\epsilon)
\)
and hence, for $x\in\Ccal$ and nonzero $g\geq0$,
\[
  q_t(x,g;\epsilon)
  \sim
  \e^{\lambda^*(\epsilon)t}\big(P_*(\epsilon)g\big)(x),
  \qquad t\to\infty.
\]
For such a payoff, let
\(
  \sS_{\min}(g)
  \coloneqq
  \operatorname*{arg\,min}_{\Ccal\in\sS(g)}d_{\s\Ccal^*\Ccal}.
\)
Combining this representation with \eqref{eq:rp} yields, for every
$\Ccal\in\cC$,
\begin{equation} \label{eq:rpl}
  \epsilon^{-d(\Ccal,g)} \bigl(P_*(\epsilon)g\bigr)_{\s\Ccal}
  \to
  \sum_{\Ccal'\in\sS_{\min}(g)}
  \bar P_{\s\Ccal\Ccal'}g_{\s\Ccal'} \gg0.
\end{equation}
Only the minimizing payoff classes survive in this limit. Hence $d(\Ccal,g)$ is
the loading exponent of the globally dominant component, while the two directed
exponents in \eqref{eq:rp} govern initial-state and terminal-payoff exposure.
Equations~\eqref{eq:rare-corridor-equivalence} and~\eqref{eq:rpl} also separate
the order-one and vanishing-loading cases. If $d(\Ccal,g)=0$, then
$\Ccal^*\in\kK(\Ccal,g)$,
$\lambda(\Ccal,g)=\lambda^*$, and the global loading has a strictly positive
order-one limit; both iterated long-run yield limits equal $-\lambda^*$. If
$d(\Ccal,g)>0$, then $\Ccal^*\notin\kK(\Ccal,g)$ and the global loading vanishes
at order $\epsilon^{d(\Ccal,g)}$. Only the latter case can generate a diverging
local-to-global crossover driven by a vanishing loading on $\Ccal^*$.

\subsection{Crossover to global pricing dominance}
\label{sec:crossover}

Fix $x\in\Ccal$ and nonzero $g\geq0$. Throughout this subsection, suppose the
limiting pricing corridor $\kK=\kK(\Ccal,g)$ is nonempty and
$d(\Ccal,g)>0$. Let $L_{\kKs}$ be its restricted pricing generator, as defined
in Section~\ref{sec:corridors}, evaluated at $\Lp=\Lp(0)$, and abbreviate
$\lambda_{\kKs}\coloneqq\lambda(\Ccal,g)=s(L_{\kKs})$.

By Proposition~\ref{prop:corridor}, $L_{\kKs}$ gives an exact representation of
the limiting price. Because $d(\Ccal,g)>0$ is equivalent to $\Ccal^*\notin\kK$ by
\eqref{eq:rare-corridor-equivalence}, we have
$\lambda_{\kKs}<\lambda^*$. Theorem~\ref{thm:rare-distance} nevertheless gives the
global mode a loading of order $\epsilon^d$. The key question is when its higher
exponential rate overcomes this vanishing loading, which requires a spectral
expansion uniform over maturities that diverge as $\epsilon\downarrow0$.

For every sufficiently small $\epsilon>0$, let $P_*(\epsilon)$ be the rank-one
Riesz projection of $\Lp(\epsilon)$ associated with $\lambda^*(\epsilon)$, and set
$P_*(0)\coloneqq P_*$, the corresponding projection of $\Lp$ at $\lambda^*$.
Unique global dominance makes $\lambda^*$ algebraically simple, while $\ker P_*$
is invariant under $\Lp$ and is the direct sum of the generalized eigenspaces
associated with all eigenvalues other than $\lambda^*$. Define
\begin{equation*}
  s_\circ \coloneqq s\big(\Lp\,|_{\ker P_*}\big),
  \qquad
  \Sigma_\circ
  \coloneqq \big\{z\in\sigma(\Lp):\Re z = s_\circ\big\}
  \subset \sigma(\Lp).
\end{equation*}
Thus $\Sigma_\circ$ consists of the eigenvalues other than $\lambda^*$ with the
largest real part. Although $s_\circ$ is real by definition, it need not itself be
an eigenvalue. Since $\sigma(\Lp|_{\ker P_*})=\sigma(\Lp)\setminus\{\lambda^*\}$ and
$\lambda_{\kKs}$ is a distinct eigenvalue of $\Lp$, we have
\(
  \lambda_{\kKs}\leq s_\circ<\lambda^*.
\)

For each $\lambda\in\sigma(\Lp)$, let $P_\lambda$ be its Riesz projection and
define
\[
  P_\circ
  \coloneqq
  \sum_{\lambda\in\Sigma_\circ}P_\lambda.
\]
Choose a bounded open neighborhood $U_\circ$ of $\Sigma_\circ$ such that
$\overline{U_\circ}\cap\sigma(\Lp)=\Sigma_\circ$. For all sufficiently small
$\epsilon\geq0$, define
\begin{equation}\label{eq:Sc}
  \Sigma_\circ(\epsilon)
  \coloneqq
  \sigma(\Lp(\epsilon))\cap U_\circ,
\end{equation}
and let $P_\circ(\epsilon)$ be its Riesz projection. Thus
$\Sigma_\circ(0)=\Sigma_\circ$ and $P_\circ(0)=P_\circ$.

Also let \( w_{\min} \coloneqq \min\{w_{k\ell}:0<w_{k\ell}<+\infty\}>0. \)
Such a positive-weight edge exists because the zero-weight subgraph is reducible
whereas $\Gcal_w$ is strongly connected. Fix a submultiplicative matrix norm.
Assumption~\ref{ass:edge} gives
\begin{equation} \label{eq:generator-perturbation-rate}
  r_L(\epsilon)
  \coloneqq
  \|\Lp(\epsilon)-\Lp\|
  =O\!\left(\eta(\epsilon)+\epsilon^{w_{\min}}\right)
  =o(1).
\end{equation}
This is the generator-level perturbation rate used below. In the physical
construction following Assumption~\ref{ass:edge},
$r_L(\epsilon)=O(r(\epsilon)+\epsilon^{w_{\min}})$, although the two error rates
need not coincide.

For sufficiently small $\epsilon\geq0$, define the global Perron loading and the
contribution associated with $\Sigma_\circ(\epsilon)$ by
\[
  a_*(x,g;\epsilon) \coloneqq \big(P_*(\epsilon)g\big)(x),
  \qquad
  b_t(x,g;\epsilon) \coloneqq \big(\e^{t\Lp(\epsilon)}P_\circ(\epsilon)g\big)(x),
  \quad t\geq0.
\]
The following result separates the vanishing global loading from the contribution
associated with $\Sigma_\circ(\epsilon)$ and controls all remaining spectral
contributions uniformly over maturity.

\begin{theorem} \label{thm:cd}
  Under Assumption~\ref{ass:edge} and the unique-dominant-class condition
  \eqref{eq:gvd}, fix $\Ccal\in\cC$, $x\in\Ccal$, and a nonzero payoff $g\geq0$.
  Suppose that $\kK=\kK(\Ccal,g)\neq\varnothing$ and $d=d(\Ccal,g)>0$.
  There exist constants $\epsilon_0,C,\delta>0$ such that, for every
  $\epsilon\in(0,\epsilon_0]$ and $t\geq0$,
  \begin{equation} \label{eq:usg}
    q_t(x,g;\epsilon)
    =
    \e^{\lambda^*(\epsilon)t} a_*(x,g;\epsilon) 
    +
    b_t(x,g;\epsilon) + R_{\epsilon,t},
    \qquad
    |R_{\epsilon,t}| \leq C \e^{(s_\circ-\delta)t}.
  \end{equation}
  Moreover, as $\epsilon\downarrow0$,
  \begin{equation} \label{eq:ggo}
    \lambda^*(\epsilon) = \lambda^*+O(r_L(\epsilon)),
    \qquad
    a_*(x,g;\epsilon) = a\,\epsilon^d(1+o(1)), \quad\text{for some } a>0.
  \end{equation}
  For every fixed $T<\infty$,
  $b_t(x,g;\epsilon)=b_t(x,g;0)+O(r_L(\epsilon))$ uniformly over $0\leq t\leq T$.
  At $\epsilon=0$, this contribution is
  \begin{equation} \label{eq:lsg}
    b_t(x,g;0)
    =
    \sum_{\lambda\in\Sigma_\circ} \e^{\lambda t}
    \sum_{j=0}^{\nu_\lambda-1}
    \frac{t^j}{j!} \big((\Lp-\lambda I)^jP_\lambda g\big)(x),
  \end{equation}
  where $\nu_\lambda$ is the size of the largest Jordan block at $\lambda$. If
  $\lambda_{\kK}<s_\circ$, then
  $b_t(x,g;0)\equiv0$: the spectral contribution associated with
  $\Sigma_\circ$ is absent from this limiting corridor price.
\end{theorem}

For bounded maturities, Theorem~\ref{thm:cd} makes the global contribution vanish at
order $\epsilon^d$, while $b_t(x,g;\epsilon)$ converges to its reducible-limit
counterpart. For every fixed positive perturbation, however, the higher Perron rate
eventually governs the long run. The resulting order-of-limits discontinuity is
immediate.

\begin{corollary}[Noncommuting long-run limits] \label{cor:noncommuting}
  Under the conditions of Theorem~\ref{thm:cd},
  \begin{equation} \label{eq:noncommuting}
    \lim_{\epsilon\downarrow0}\lim_{t\to\infty}
    \left[-\frac1t\log q_t(x,g;\epsilon)\right]
    =-\lambda^*,
    \qquad
    \lim_{t\to\infty}\lim_{\epsilon\downarrow0}
    \left[-\frac1t\log q_t(x,g;\epsilon)\right]
    =-\lambda_{\kKs}
    =-\lambda(\Ccal,g).
  \end{equation}
\end{corollary}

Thus taking maturity first selects the global rate, whereas taking the reducible
limit first restores the corridor rate. This conclusion does not require the local
contribution to be a single exponential term: by \eqref{eq:lsg}, it may combine
several eigenvalues and Jordan terms. Moreover, eigenvalues other than $\lambda^*$
with real parts above $\lambda_{\kKs}$ may have zero coefficients in the limiting
corridor price but acquire nonzero loadings when rare edges open, producing
intermediate crossovers.
A two-term exponential formula therefore requires the following spectral condition.

\begin{assumption}[Single simple eigenvalue at the corridor rate]
  \label{ass:single-corridor-eigenvalue}
  The corridor rate is the unique eigenvalue of $\Lp$, other than $\lambda^*$, with
  maximal real part:
  \begin{equation} \label{eq:single-corridor-eigenvalue-main}
    \lambda_{\kKs}=s_\circ,
    \qquad
    \Sigma_\circ=\{\lambda_{\kKs}\}.
  \end{equation}
  Moreover, $\lambda_{\kKs}$ is algebraically simple as an eigenvalue of $\Lp$.
\end{assumption}

Under Assumption~\ref{ass:single-corridor-eigenvalue}, set
\(
  \lambda_\circ\coloneqq s_\circ=\lambda_{\kKs}=\lambda(\Ccal,g).
\)
The assumption implies that $P_\circ=P_{\lambda_\circ}$ has rank one, that the
Jordan block at $\lambda_\circ$ has size one, and that every eigenvalue other than
$\lambda^*$ and $\lambda_\circ$ has real part strictly smaller than
$\lambda_\circ$. By the projection estimate
\eqref{eq:cp} in Lemma~\ref{lem:sd} and rank stability, for every sufficiently
small $\epsilon\geq0$ the set
$\Sigma_\circ(\epsilon)$ consists of a single real algebraically simple eigenvalue
$\lambda_\circ(\epsilon)$, with $\lambda_\circ(0)=\lambda_\circ$. Define
$b_\circ(x,g;\epsilon)\coloneqq(P_\circ(\epsilon)g)(x)$. Then
\begin{equation} \label{eq:corridor-exponential-term}
  b_t(x,g;\epsilon)
  =
  \e^{\lambda_\circ(\epsilon)t}b_\circ(x,g;\epsilon),
  \qquad t\geq0.
\end{equation}
The same projection estimate and the rank-one trace identity give, as
$\epsilon\downarrow0$,
\begin{equation} \label{eq:corridor-eigenvalue-rates}
  \lambda_\circ(\epsilon)
  =\lambda_\circ+O(r_L(\epsilon)),
  \qquad
  b_\circ(x,g;\epsilon)
  =b_\circ(x,g;0)+O(r_L(\epsilon)).
\end{equation}

Write $\Delta(\epsilon)\coloneqq\lambda^*(\epsilon)-\lambda_\circ(\epsilon)$ and
$\Delta\coloneqq\Delta(0)=\lambda^*-\lambda_\kK>0$.

\begin{theorem}\label{thm:crossover}
  Under the conditions of Theorem~\ref{thm:cd} and
  Assumption~\ref{ass:single-corridor-eigenvalue}, $b_\circ(x,g;0)>0$, and the
  following hold.
  \begin{enumerate}
  \item There exist $\epsilon_0,C,\delta>0$ such
  that, for every $\epsilon\in(0,\epsilon_0]$ and $t\geq0$,
  \begin{equation} \label{eq:uts}
    q_t(x,g;\epsilon)
    =
    a_*(x,g;\epsilon)\e^{\lambda^*(\epsilon)t}
    +b_\circ(x,g;\epsilon)\e^{\lambda_\circ(\epsilon)t}
    +R_{\epsilon,t},
    \qquad
    |R_{\epsilon,t}|
    \leq
    C\e^{(\lambda_\circ-\delta)t}.
  \end{equation}
  \item For all sufficiently small $\epsilon>0$,
  the two exponential terms in \eqref{eq:uts} are equal at the unique positive
  maturity
  \begin{equation} \label{eq:crossover-main}
    t_\epsilon^c
    =
    \frac{\log[b_\circ(x,g;\epsilon)/a_*(x,g;\epsilon)]}
    {\Delta(\epsilon)}.
  \end{equation}
  With $a>0$ as in \eqref{eq:ggo}, as $\epsilon\downarrow0$,
  \begin{equation}\label{eq:crossover-loading-correction}
    t_\epsilon^c
    =
    \frac{d\log(1/\epsilon)+\log[b_\circ(x,g;0)/a]}{\Delta(\epsilon)}
    +o(1).
  \end{equation}
  \item As $\epsilon\downarrow0$,
  \begin{equation} \label{eq:crossover}
    \frac{
      a_*(x,g;\epsilon)
      \e^{\lambda^*(\epsilon)(t_\epsilon^c+s)}
    }{
      q_{t_\epsilon^c+s}(x,g;\epsilon)
    }
    \to
    \frac{\e^{\Delta s}}{1+\e^{\Delta s}}
    \quad\text{uniformly for $s$ in compact subsets of $\RR$.}
  \end{equation}
  \end{enumerate}
\end{theorem}

If the generator additionally satisfies
\begin{equation}\label{eq:crossover-rate-condition}
  r_L(\epsilon)\log(1/\epsilon)\to0,
  \qquad \epsilon\downarrow0,
\end{equation}
then \eqref{eq:crossover-loading-correction} simplifies to
\begin{equation}\label{eq:crossover-constant-correction}
  t_\epsilon^c
  =
  \frac{d}{\Delta}\log(1/\epsilon)
  +\frac{1}{\Delta}\log\frac{b_\circ(x,g;0)}{a}
  +o(1).
\end{equation}
The condition holds whenever $r_L(\epsilon)=O(\epsilon^p)$ for some $p>0$ and
ensures that replacing $\Delta(\epsilon)$ by $\Delta$ changes the approximation
by only $o(1)$. Holding $d$ and $\Delta$ fixed, a larger global loading
coefficient $a$ relative to $b_\circ(x,g;0)$ advances the crossover through the
constant correction.

Under the theorem's assumptions, \eqref{eq:crossover-loading-correction}
implies the leading scale $t_\epsilon^c\sim(d/\Delta)\log(1/\epsilon)$. Holding the
other quantity fixed, a larger rarity exponent $d$ delays the crossover on this
scale, whereas a larger pricing-rate gap $\Delta$ brings it forward. This is a
spectral peso problem: a rarely accessed class can be nearly invisible at ordinary
maturities yet govern the long end because its initially small price contribution has
a higher exponential rate. In \eqref{eq:crossover}, the global term's limiting share
of the full price passes through one half at the crossover. A larger $\Delta$ makes
this transition sharper as maturity moves through $t_\epsilon^c$.

\subsection{Rare disaster entry with fixed recovery}
\label{sec:rare-fixed-recovery}

In Section~\ref{sec:illustration}, we set normal-to-disaster entry equal to zero. We now
restore entry at order $\epsilon$, holding within-class dynamics and recovery
fixed, and apply Sections~\ref{sec:rare-entry}--\ref{sec:crossover} to determine
its price loading and crossover horizon.

\subsubsection{Pricing concentration and crossover in the two-class model}

Let $F_\epsilon$ collect the physical normal-to-disaster transition intensities, set
$D_\epsilon\coloneqq\diag(F_\epsilon\1_{\mathcal D})$, and let $\Gamma\gg0$ collect
the corresponding growth-adjusted SDF jump multipliers. With normal states ordered
first, the physical and pricing generators are
\begin{equation} \label{eq:rare-fixed-recovery-generators}
  \Lx(\epsilon)
  =
  \begin{pmatrix}
    Q-D_\epsilon&F_\epsilon\\
    H&U
  \end{pmatrix},
  \qquad
  \Lp(\epsilon)
  =
  \begin{pmatrix}
    \LpN-D_\epsilon&E_\epsilon\\
    \LpDN&\LpD
  \end{pmatrix},
  \qquad
    E_\epsilon\coloneqq F_\epsilon\odot\Gamma.
\end{equation}
Here $Q$ and $U$ govern the within-normal and within-disaster physical dynamics,
respectively, $H$ contains the recovery rates, and the pricing blocks are those
introduced in Section~\ref{sec:illustration}. The diagonal adjustment
$-D_\epsilon$ records the additional exit intensity from the normal class. Suppose
\begin{equation} \label{eq:rare-entry-main}
  F_\epsilon
  =
  \epsilon F+O(\epsilon^{1+\eta_F}),
  \qquad
  E_\epsilon
  =
  \epsilon J+O(\epsilon^{1+\eta_F}),
  \qquad
  J\coloneqq F\odot\Gamma,
\end{equation}
so regular risk adjustment changes the leading entry coefficient, through
$\Gamma$, but not its rarity exponent. The following assumption collects the
primitive regularity conditions; Appendix~\ref{app:rare-entry-primitive} derives
the pricing generator from those primitives.

\begin{assumption}[Rare entry with fixed recovery] \label{ass:illustration-near}
  The matrices $Q$ and $U$ are irreducible, $Q$ is a Markov generator, $U$ is a
  transient subgenerator, $H\geq0$, and
  $H\1_{\mathcal N}+U\1_{\mathcal D}=0$. Equation~\eqref{eq:rare-entry-main} holds
  for some $\eta_F>0$ and some nonzero $F\geq0$, and $F_\epsilon\geq0$ for every
  sufficiently small $\epsilon>0$. The matrices $\LpN$ and $\LpD$ are
  irreducible Metzler matrices, $\LpDN\geq0$ is nonzero, and $\Gamma$ is fixed and
  strictly positive.
\end{assumption}

These conditions imply that the physical and pricing generators in
\eqref{eq:rare-fixed-recovery-generators} are irreducible for every sufficiently
small $\epsilon>0$: the entry and recovery blocks connect the two irreducible
within-class graphs in opposite directions.

Maintain disaster pricing dominance, $\lambdaD>\lambdaN$. Let
$\pi_\epsilon^P$ denote the stationary distribution of the physical generator
$\Lx(\epsilon)$, and let $\pi_\epsilon^{\mathrm{LR}}$ be the long-run pricing law defined
in Section~\ref{sec:rare-entry}. The next result makes their contrasting limits
transparent.

\begin{corollary}[Physical rarity and pricing concentration]
\label{cor:rare-disaster-concentration}
Under Assumption~\ref{ass:illustration-near} and $\lambdaD>\lambdaN$, there exist
constants $c_P,c_{\mathrm{LR}}>0$ such that, as $\epsilon\downarrow0$,
\[
  \pi_\epsilon^P(\mathcal D)
  =
  c_P\epsilon+o(\epsilon),
  \qquad
  \pi_\epsilon^{\mathrm{LR}}(\mathcal D)
  =
  1-c_{\mathrm{LR}}\epsilon+o(\epsilon).
\]
\end{corollary}

Physical frequency and long-run pricing importance thus have opposite limits:
the disaster class becomes physically negligible while the long-run pricing law
concentrates on it. This does not determine when its contribution becomes
important in a price starting from the normal class.

For the crossover, fix $x\in\mathcal N$ and the growth-adjusted zero-coupon
payoff $g=\1$. A single rare entry edge reaches disaster, where $g$ is positive,
so $d(\mathcal N,\1)=1$; the limiting corridor is $\{\mathcal N\}$ with rate
$\lambdaN$. Theorem~\ref{thm:crossover} requires $\lambdaN$ to exceed the real
part of every other eigenvalue except $\lambdaD$. Irreducibility of $\LpN$
therefore leaves only the condition
\begin{equation} \label{eq:rds}
  \sup_{z\in\sigma(\LpD)\setminus\{\lambdaD\}}\Re z
  <
  \lambdaN,
\end{equation}
with the convention $\sup\varnothing=-\infty$. Under this condition,
Assumption~\ref{ass:single-corridor-eigenvalue} holds with
$\lambda_{\kK}=\lambda_\circ=\lambdaN$.

\begin{corollary}[Rare-entry crossover from the normal class] \label{cor:rare-disaster-crossover}
  Under Assumption~\ref{ass:illustration-near}, $\lambdaD>\lambdaN$, and
  \eqref{eq:rds}, fix $x\in\mathcal N$. Theorem~\ref{thm:crossover} applies to
  $g=\1$ with $d=1$ and $\Delta=\lambdaD-\lambdaN$. The coefficients
  \[
    a_x\coloneqq\lim_{\epsilon\downarrow0}
      \epsilon^{-1}a_*(x,\1;\epsilon),
    \qquad
    b_x\coloneqq b_\circ(x,\1;0)
  \]
  exist and are strictly positive. As $\epsilon\downarrow0$,
  \begin{equation} \label{eq:rare-disaster-crossover-main}
    t_\epsilon^c(x)
    =
    \frac{\log(1/\epsilon)+\log(b_x/a_x)}{\lambdaD-\lambdaN}
    +o(1).
  \end{equation}
\end{corollary}

With recovery held fixed, $r_L(\epsilon)=O(\epsilon)$, so
\eqref{eq:crossover-rate-condition} holds and yields the expansion
\eqref{eq:rare-disaster-crossover-main}.
Corollary~\ref{cor:noncommuting} also shows that the order of limits matters:
letting maturity tend to infinity before sending entry intensity to zero
gives the limiting yield $-\lambdaD$, whereas reversing these limits gives
$-\lambdaN$.

\subsubsection{A two-state consumption-based application}
\label{sec:crra-rare}

We retain the within-state consumption dynamics, preferences, and cash-flow
specification of Section~\ref{sec:crra}, together with the recovery intensity
$\kappa$ and consumption multiplier $\chi_{\sss DN}$, and vary only the physical
disaster-entry intensity $\epsilon\ell$. Here $\ell>0$ is a fixed reference
intensity, and $\epsilon>0$ controls the rarity of disaster entry. At entry to
disaster, consumption is multiplied by $\chi_{\sss ND}>0$, so the growth-adjusted SDF
is multiplied by $\chi_{\sss ND}^{\vartheta-\alpha}$.

With time measured in years, an economy starting in the normal state $N$
enters disaster within one year with probability $1-\e^{-\epsilon\ell}$,
and its mean waiting time to first entry is $1/(\epsilon\ell)$.
The recovery intensity $\kappa$, by contrast, determines the mean disaster
duration $1/\kappa$.

The physical and pricing generators are, respectively,
\begin{equation}\label{eq:crra-rare-generators}
  \Lx(\epsilon)
  =
  \begin{pmatrix}
    -\epsilon\ell&\epsilon\ell\\[2pt] \kappa&-\kappa
  \end{pmatrix},
  \qquad
  \Lp(\epsilon)
  =
  \begin{pmatrix}
    -\varrho_{\sss N}-\epsilon\ell&\epsilon\ell\chi_{\sss ND}^{\vartheta-\alpha}
    \\[2pt]
    \kappa\chi_{\sss DN}^{\vartheta-\alpha}&-\varrho_{\sss D}-\kappa
  \end{pmatrix}.
\end{equation}
Here $\varrho_{\sss N}$ and $\varrho_{\sss D}$ are the within-state growth-adjusted
discount rates from Section~\ref{sec:crra}. Both generators are irreducible for every
$\epsilon>0$. We retain $\lambdaN=-\varrho_{\sss N}$ and $\lambdaD=-\varrho_{\sss
D}-\kappa$ for the zero-entry class pricing rates and assume
\[
  \Delta=\lambdaD-\lambdaN > 0.
\]
The spectral separation condition \eqref{eq:rds} holds automatically as both
limiting class blocks are scalar.

The two eigenvalues of the pricing generator are given by
\begin{equation}\label{eq:crra-rare-eigenvalues}
  \lambda_\pm(\epsilon)
  =
  \frac12 [\lambdaN-\epsilon\ell+\lambdaD\pm\Delta(\epsilon)],
  \qquad
  \Delta(\epsilon)
  \coloneqq
  \sqrt{(\Delta+\epsilon\ell)^2
  + 4\, \kappa\, \epsilon\ell\, (\chi_{\sss ND}\chi_{\sss DN})^{\vartheta-\alpha}}
  >0.
\end{equation}
Here $\lambda_+(\epsilon)$ is the global long-run pricing rate and
$\Delta(\epsilon)$ is the eigenvalue gap; as $\epsilon\downarrow0$,
$\lambda_+(\epsilon)\to\lambdaD$, $\lambda_-(\epsilon)\to\lambdaN$, and
$\Delta(\epsilon)\to\Delta$.
Direct calculation gives the physical stationary mass and long-run pricing mass of
the disaster state as
\begin{equation}\label{eq:crra-rare-masses}
  \pi_\epsilon^P(D)
  =
  \frac{\epsilon\ell}{\kappa+\epsilon\ell},
  \qquad
  \pi_\epsilon^{\mathrm{LR}}(D)
   =\frac12\left(1+\frac{\Delta+\epsilon\ell}{\Delta(\epsilon)}\right).
\end{equation}
Their first-order expansions recover
Corollary~\ref{cor:rare-disaster-concentration}; Appendix~\ref{app:rare-entry-primitive}
derives these identities, their expansion constants, and the price formulas below.

\begin{figure}[!t]
  \centering
  \includegraphics[width=0.9\textwidth]{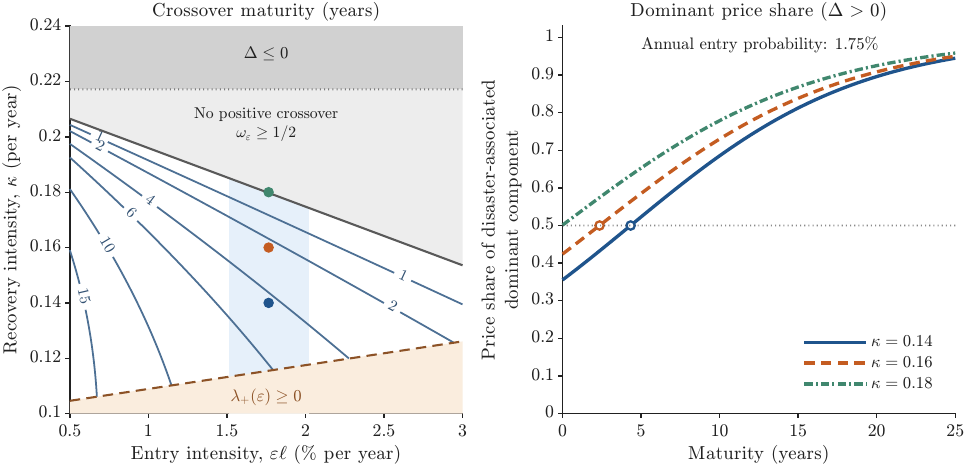}
  \caption{
    Exact crossover maturities and dominant-component price shares from the normal state.
    Both panels use $\alpha=4$ and $\vartheta=0$; remaining fixed parameters
    are given in the text.
    Left: crossover contours (years) where $\Delta>0$, $0<\omega_\epsilon<1/2$,
    and $\lambda_+(\epsilon)<0$; the blue band marks annual entry probabilities
    of 1.5--2\%. Light gray denotes $\omega_\epsilon\geq1/2$, dark gray
    $\Delta\leq0$, and tan $\lambda_+(\epsilon)\geq0$; the dashed line marks
    $\lambda_+(\epsilon)=0$.
    Right: $\omega_\epsilon\e^{\lambda_+(\epsilon)t}/z_t(N;\epsilon)$ at
    annual entry probability 1.75\% and $\kappa=0.14,0.16,0.18$ per year
    (all with $\Delta>0$). Matching filled points identify these cases on the
    left; open circles mark positive crossovers at share $1/2$ on the right.
  }
  \label{fig:rare-entry-crossover}
\end{figure}

As in Corollary~\ref{cor:rare-disaster-crossover}, take $g=\1$ and start in the
normal state $N$. The maturity-$t$ payment is $G_t$; it is a fixed real payment only when
$\vartheta=0$. Write $z_t(N;\epsilon)=q_t(N,\1;\epsilon)$. We have
\begin{equation}\label{eq:crra-rare-price}
  z_t(N;\epsilon)
  =
  \omega_\epsilon\, \e^{\lambda_+(\epsilon)t}
  +(1-\omega_\epsilon)\, \e^{\lambda_-(\epsilon)t},
  \qquad
  \omega_\epsilon
  \coloneqq
  \frac{\lambda_+(\epsilon)-\lambdaD +\epsilon\ell\chi_{\sss ND}^{\vartheta-\alpha}}
  {\Delta(\epsilon)}.
\end{equation}
The loading $\omega_\epsilon$ on the dominant exponential is strictly positive,
with%
\footnote{
  The term in the brackets is the coefficient of $\e^{\lambdaD t}$ in the long-run
  price $q_t(D,\1;0)$ derived in Section~\ref{sec:crra-recovery}. Its two terms arise
  from sample paths without and with recovery by maturity, respectively. 
}
\begin{equation}\label{eq:crra-rare-loading}
  \omega_\epsilon=\epsilon\omega+O(\epsilon^2),
  \qquad
  \omega\coloneqq\frac{\ell\chi_{\sss ND}^{\vartheta-\alpha}}{\Delta}
                         \left(1+\frac{\kappa\chi_{\sss DN}^{\vartheta-\alpha}}{\Delta}\right)>0.
\end{equation}
For all sufficiently small $\epsilon$, $0<\omega_\epsilon<1/2$. The two positive
price contributions in \eqref{eq:crra-rare-price} then become equal at the
strictly positive maturity
\begin{equation}\label{eq:crra-exact-crossover}
  t_\epsilon^c(N)
   =\frac{\log[(1-\omega_\epsilon)/\omega_\epsilon]}{\Delta(\epsilon)}
   =\frac{\log(1/\epsilon)}{\Delta}-\frac{\log\omega}{\Delta}+o(1),
  \qquad \epsilon\downarrow0.
\end{equation}
This specializes Corollary~\ref{cor:rare-disaster-crossover} with $a_N=\omega$
and $b_N=1$.

At any $\epsilon$ for which $0<\omega_\epsilon<1$, we have
\[
  \frac{\omega_\epsilon\e^{\lambda_+(\epsilon)t}}{z_t(N;\epsilon)}
  =
  \left[
    1+ \omega_\epsilon^{-1}(1-\omega_\epsilon)\e^{-\Delta(\epsilon) t}
  \right]^{-1}.
\]
This share reaches $1/2$ at $t_\epsilon^c(N)$ when the crossover is positive.
If $1/2\leq\omega_\epsilon<1$, the dominant contribution already accounts for at
least half of the price at maturity zero, so there is no strictly positive
crossover.%
\footnote{Outside the small-entry region, $1-\omega_\epsilon$ can be
nonpositive. The two-positive-contribution interpretation then fails, although
\eqref{eq:crra-rare-price} remains an exact identity.}

Figure~\ref{fig:rare-entry-crossover} uses the within-state consumption
dynamics and subjective discounting specified in \eqref{eq:crra-benchmark}.
Both panels set $\alpha=4$, $\vartheta=0$, $\ell=1$ per year,
$\chi_{\sss ND}=\e^{-0.111}$, $\chi_{\sss DN}=\e^{0.111}$, and $g=\1$.
With time measured in years, the zero-entry pricing-rate gap in
\eqref{eq:crra-gap} becomes
\[
  \Delta=\lambdaD-\lambdaN
  =4(0.025)+\frac12\,4^2(0.121^2)-\kappa
  =0.217128-\kappa.
\]

The contour plot on the left panel of Figure~\ref{fig:rare-entry-crossover} shows how
entry and recovery intensities affect crossover maturity from the normal state,
holding the remaining parameters fixed. The blue band marks annual entry
probabilities of 1.5--2\%, using the disaster probabilities reported by
\citet{barro2006rare} as a benchmark. With $\ell=1$ per year, these probabilities
correspond to $\epsilon\approx0.01511$--$0.02020$ and mean waiting times to entry of
approximately 50--66 years from the normal state.%
\footnote{
  These long waiting times reflect the rarity of severe disasters, rather than
  ordinary business-cycle downturns. These estimates provide an entry-frequency
  benchmark, not a full calibration of the consumption-based model.
}
The recovery-intensity range matches that in Figure~\ref{fig:dp}, corresponding to
mean disaster durations of approximately 4.2--10 years. Within the region displaying
crossover contours, higher entry intensity brings crossover forward. Higher recovery
intensity also brings crossover forward in the current setting. In particular, more
persistent disasters do not necessarily produce earlier crossover, because recovery
intensity affects both the pricing-rate gap and the dominant-component loading.

The right panel of Figure~\ref{fig:rare-entry-crossover} plots the
disaster-associated dominant component's share of the price from the normal state.
The annual disaster entry probability is fixed at 1.75\%, implying
$\epsilon=-\log(1-0.0175)\approx0.01765$ and a mean waiting time to entry of 56.6
years. Recovery intensities $\kappa=0.14,0.16,0.18$ correspond to mean disaster
durations of $7.14,6.25,5.56$ years, respectively. For the first two cases, the share
reaches one half at maturities of $4.34$ and $2.36$ years. In the third case,
$\omega_\epsilon\approx0.501$, so the dominant component already accounts for
slightly more than half the price at maturity zero. All three shares increase with
maturity, reaching approximately $94.5\%$, $95.0\%$, and $95.8\%$ at 25 years.

The three recovery intensities give
$\Delta\approx0.0771,0.0571,0.0371$ per year, respectively, so all
three cases satisfy $\Delta>0$. We call the dominant component
``disaster-associated'' because $\lambda_+(\epsilon)\to\lambdaD$ as
$\epsilon\downarrow0$; for every $\epsilon>0$, however, it is a global spectral
component involving both states. Its price share is neither a disaster probability
nor the price share of paths entering disaster. At the baseline $\kappa=0.18$, for
example, the physical stationary disaster probability is $8.9\%$, whereas the
disaster mass under the long-run pricing law is $71.9\%$, by
\eqref{eq:crra-rare-masses}. These stationary state masses measure different objects
from the maturity-dependent price share plotted in the right panel.

Both panels of Figure~\ref{fig:rare-entry-crossover} use exact finite-entry formulas.
The left-panel contours are computed from
$\log[(1-\omega_\epsilon)/\omega_\epsilon]/\Delta(\epsilon)$ in
\eqref{eq:crra-exact-crossover}, while the right-panel price shares are computed from
the decomposition in \eqref{eq:crra-rare-price}. At the benchmark entry probabilities
and parameter values considered here, the first-order loading approximation in
\eqref{eq:crra-rare-loading} and the asymptotic crossover expansion in
\eqref{eq:crra-exact-crossover} need not provide close approximations to the exact
quantities. Appendix~\ref{app:rare-entry-primitive} illustrates this approximation
error.

\section{Conclusion}

We show that physical recurrence and long-run pricing dominance are distinct.
Although the physical and pricing generators share the same communicating classes and
accessibility relations, stochastic discounting and cash-flow growth determine the
ordering of class pricing rates. A physically transient class can therefore govern
long-maturity prices.

The pricing corridor identifies when such dominance matters for a particular initial
state and nonnegative payoff. Restricting the pricing generator to this corridor
preserves the price at every maturity. For a nonempty corridor, its highest class
rate determines the long-run yield, while chains of classes tied at that rate
determine the polynomial order of the price. The corresponding perpetual cash-flow
stream has finite value exactly when the corridor rate is negative. Global valuation
instability can thus coexist with finite values for particular states and claims.

We characterize how rare transitions change these conclusions when the limiting
corridor excludes the uniquely pricing-dominant class. Minimum-rarity paths from the
initial class to the dominant class and onward to the payoff support determine the
asymptotic order of its price loading. Under two-term spectral separation, this
loading and the perturbed pricing-rate gap determine the crossover maturity, which
diverges logarithmically as the connecting transitions vanish. The reducible and
long-maturity limits need not commute: neglecting a rare transition can give an
accurate fixed-maturity price but a different long-run yield.

The consumption-based disaster-recovery application illustrates the economic
significance of this mechanism. With empirically motivated entry frequencies and
disaster durations, the disaster-associated dominant component can account for most
of a claim's price from a normal initial state at economically relevant maturities.
Quantitative assessment nevertheless requires jointly disciplining physical
transitions, cash-flow exposures, and stochastic discounting. Rarity alone does not
justify neglecting a transition: its pricing importance depends on the claim and
maturity.

\bigskip

\begin{center}
  \textsc{\large Appendix}
\end{center}

\appendix

\section{Multiplicative Functionals and the Primitive Pricing Generator}
\label{app:param}

Section~\ref{sec:environment} constructs the pricing semigroup abstractly from
the SDF and cash-flow growth processes. This appendix provides its semimartingale
foundation: it verifies the semigroup and continuity properties, parameterizes the
two processes as exponential additive functionals, and derives the pricing generator
from their primitive shock exposures. The construction builds on and extends the
specification in \citet{hansen2009long} by explicitly allowing the information flow
to contain continuous and jump shocks that affect stochastic discounting and
cash-flow growth without changing the finite-state economic state. Related
continuous-time formulations include \citet{hansen2012pricing},
\citet{borovicka2011risk}, and \citet{duffie2000transform}.

\subsection{From multiplicative functionals to the pricing semigroup}
\label{app:mf-semigroup}

Let $g\in\RR_+^\Xsf$. Multiplicativity of $S^{\Gs}$, iterated expectations, and
the Markov property \eqref{eq:EY}%
\footnote{
  By truncation and monotone convergence, \eqref{eq:EY} extends to every nonnegative
  $\fF_\infty^\circ$-measurable $Y$, with both sides interpreted as
  $[0,+\infty]$-valued random variables. We apply this extension to $Y = S_u^{\Gs}
  g(X_u)$, which is nonnegative and $\fF_\infty^\circ$-measurable.
}
give
\[
  \begin{aligned}
    (\Pi_{t+u}g)(x)
    =
    \EE_x\!\left[
      S_t^{\Gs}
      \bigl(S_u^{\Gs}g(X_u)\bigr)\circ\theta_t
    \right] 
    &=
    \EE_x\!\left\{
      S_t^{\Gs}\,
      \EE
      \left[
        \bigl(S_u^{\Gs}g(X_u)\bigr)\circ\theta_t \mid \mathscr F_t
      \right]
    \right\} \\
    &=
    \EE_x\!\left[
      S_t^{\Gs}\,
      \EE_{\scriptscriptstyle X_t}
      \bigl(S_u^{\Gs}g(X_u)\bigr)
    \right]
    =\bigl(\Pi_t\Pi_ug\bigr)(x).
  \end{aligned}
\]
Linearity extends this identity to every $g\in\RR^\Xsf$, proving the semigroup
property \eqref{eq:sgp}.

It remains to justify continuity at zero. The c\`adl\`ag paths and the normalization
$S_0^{\Gs}=1$ imply $S_t^{\Gs}\to1$, $\PP_x$-almost surely, as $t\to0$.
Assumption~\ref{ass:continuity} and positivity then imply
$\EE_x|S_t^{\Gs}-1|\to0$ by Scheff\'e's lemma. For
$\|g\|_\infty\leq1$,
\[
  \bigl|(\Pi_tg)(x)-g(x)\bigr|
  \leq
  \EE_x\bigl|S_t^{\Gs}-1\bigr|
  +2\PP_x(X_t\neq x).
\]
Right continuity of $X$ makes the second term vanish. Because $\Xsf$ is finite,
both convergences are uniform in $x$, so $\|\Pi_t-I\|\to0$. This proves the
continuity assertion used to define $\Lp$ in \eqref{eq:Lp}.

\subsection{Additive-functional semimartingale specification}
\label{app:additive}

Let $N^X$ be the jump measure of the state process,
\[
  N^X\bigl((0,t]\times B\bigr)
  \coloneqq
  \sum_{0<u\leq t}
  \1_{\{X_{u-}\neq X_u,\,X_u\in B\}},
  \qquad t\geq0,\quad B\subseteq\Xsf.
\]
Thus $N^X((0,t]\times B)$ counts transitions up to time $t$ whose destination
lies in $B$. If $\ell_{ij}=(\Lx)_{ij}$, its predictable compensator has, when
$X_{t-}=x_i$, intensity $\ell_{ij}$ at $x_j\neq x_i$.

In addition to these state-transition shocks, allow shocks that affect discounting
or growth without changing $X$ at the same instant. Let $W$ be an
$(\fF_t^\circ)$-adapted $d$-dimensional process that is an $(\fF_t)$-Brownian
motion under every $\PP_x$, and let $N$ be an integer-valued random measure on
$\RR_+\times\Zsf$, where $(\Zsf,\mathcal Z)$ is a measurable mark space. Under each
$\PP_x$, suppose that $N((0,t]\times B)$ is $\fF_t^\circ$-measurable and that $N$
has predictable compensator
$\diff t\,\nu(X_{t-},\diff z)$ for a finite state-dependent kernel $\nu$, and that
$N$ and $N^X$ have no common jumps. Assume also that the driving shocks are
compatible with the shifts: for $t,u\geq0$ and $B\in\mathcal Z$,
\[
  W_u\circ\theta_t=W_{t+u}-W_t,
  \qquad
  N\bigl((0,u]\times B\bigr)\circ\theta_t
  =N\bigl((t,t+u]\times B\bigr),
\]
$\PP_x$-almost surely for every $x$.

An $(\fF_t^\circ)$-adapted real-valued process $A$ is an \emph{additive
functional} relative to the shifts if $A_0=0$ and
\(
  A_{t+u}=A_t+A_u\circ\theta_t
\)
for $t,u\geq0$, $\PP$-a.s. Then $\exp(A)$ is a strictly positive multiplicative
functional. The following specification realizes the abstract processes in
Assumption~\ref{ass:mf} in this way.

For $K\in\{S,G\}$, assume
\begin{align}
  \log K_t
  &=
  \int_0^t\mu_K(X_{u-})\,\diff u
  +
  \int_0^t\varsigma_K(X_{u-})^\t\,\diff W_u
  \nonumber\\
  &\quad
  +
  \int_0^t\int_{\Xsf}
  \xi_K(X_{u-},y)N^X(\diff u,\diff y)
  +
  \int_0^t\int_{\Zsf}
  \zeta_K(X_{u-},z)N(\diff u,\diff z).
  \label{eq:param-appendix}
\end{align}
Here $\mu_K:\Xsf\to\RR$ is the local drift of $\log K$,
$\varsigma_K:\Xsf\to\RR^d$ is its continuous-shock exposure,
$\xi_K:\Xsf\times\Xsf\to\RR$ is the jump at a state transition, and
$\zeta_K:\Xsf\times\Zsf\to\RR$ is the response to a non-transition jump; set
$\xi_K(x,x)=0$. Assume these functions are measurable and finite and that
\[
  \max_{x\in\Xsf}
  \int_{\Zsf}
  \exp\bigl(\zeta_S(x,z)+\zeta_G(x,z)\bigr)\nu(x,\diff z)
  <\infty.
\]
Since $\Xsf$ and the jump intensities are finite, this exponential-moment condition
ensures the finite-horizon moment required in Assumption~\ref{ass:mf}. The shift
identities above make the right-hand side of \eqref{eq:param-appendix} additive, so
its exponential is a strictly positive c\`adl\`ag multiplicative functional; the
same conditions also give the continuity at zero in
Assumption~\ref{ass:continuity}.

The four coefficient families isolate three sources of uncertainty. Brownian
innovations generate continuous shocks without changing the state. A state transition
may simultaneously change $S$ and $G$ through $\xi_S$ and $\xi_G$. The random
measure $N$ generates discrete shocks that leave the current state unchanged through
$\zeta_S$ and $\zeta_G$. Because $S$ and $G$ share these driving processes, their
continuous and jump exposures may be correlated and state dependent.

This construction builds on the multiplicative-functional framework of
\citet[Sections~3.1--3.5]{hansen2009long}, while explicitly allowing discounting and
cash-flow growth to respond to shocks that do not change the finite-state economic
state.\footnote{\citet[Footnote~6]{hansen2009long} note that a larger filtration
could be accommodated but assume, for simplicity, that the filtration is generated
by the Markov state.}
The distinction between regime transitions and shocks conditional on the regime is
familiar from finite-state asset-pricing models \citep{BonomoGarcia1996}. Related
continuous-time formulations with state-dependent diffusion and jump exposures
include \citet{duffie2000transform}, \citet{borovicka2011risk}, and
\citet{hansen2012pricing}; structural applications featuring rare jump risk include
\citet{gabaix2012variable} and \citet{wachter2013time}.

For the growth-adjusted SDF $S^{\Gs}=SG$, define the aggregate coefficients
\[
  \mu=\mu_S+\mu_G,
  \quad
  \varsigma=\varsigma_S+\varsigma_G,
  \quad
  \xi=\xi_S+\xi_G,
  \quad
  \zeta=\zeta_S+\zeta_G.
\]
Then $\log S^{\Gs}$ has the representation \eqref{eq:param-appendix} with
$(\mu_K,\varsigma_K,\xi_K,\zeta_K)$ replaced by
$(\mu,\varsigma,\xi,\zeta)$.

\subsection{Pricing generator and effective discounting}
\label{app:effective-discounting}

The next proposition maps the primitive shock exposures into the matrix generator
defined abstractly in \eqref{eq:Lp}.

\begin{proposition}[Primitive pricing generator] \label{prop:param-generator}
  Under Assumptions~\ref{ass:markov}--\ref{ass:continuity} and the semimartingale
  specification above, let $\ell_{ij}=(\Lx)_{ij}$. The pricing generator has
  entries
  \begin{equation}
  \label{eq:param-generator}
    (\Lp)_{ij}
    =
    \begin{cases}
      \ell_{ii}-\varrho(x_i), & i=j,\\[3pt]
      \ell_{ij}\e^{\xi(x_i,x_j)}, & i\neq j,
    \end{cases}
  \end{equation}
  where $\|\cdot\|$ denotes the Euclidean norm on $\RR^d$ and
  \[
    \varrho(x)
    =
    -\mu(x) -\frac12\|\varsigma(x)\|^2 -
    \int_{\Zsf} \left(\e^{\zeta(x,z)}-1\right)\nu(x,\diff z).
  \]
  In particular, $\Lp$ and $\Lx$ have the same off-diagonal zero pattern.
\end{proposition}

\begin{proof}[Proof of Proposition~\ref{prop:param-generator}]
  Applying It\^o's formula to $S_t^{\Gs}=\exp(\log S_t^{\Gs})$ gives
  \begin{align*}
    \frac{\diff S_t^{\Gs}}{S_{t-}^{\Gs}}
    &=
    \left[
      \mu(X_{t-})
      +\frac12\|\varsigma(X_{t-})\|^2
    \right]\diff t
    +
    \varsigma(X_{t-})^\t\diff W_t
    \\
    &\quad
    +
    \int_{\Xsf}
    \left[e^{\xi(X_{t-},y)}-1\right]N^X(\diff t,\diff y)
    +
    \int_{\Zsf}
    \left[e^{\zeta(X_{t-},z)}-1\right]N(\diff t,\diff z).
  \end{align*}
  For $g\in\RR^\Xsf$, apply the product rule to $S_t^{\Gs}g(X_t)$. At a
  transition from $x_i$ to $x_j$, the product changes from
  $S_{t-}^{\Gs}g(x_i)$ to
  $S_{t-}^{\Gs}e^{\xi(x_i,x_j)}g(x_j)$. Compensation of $N^X$ and $N$
  therefore gives predictable drift
  \[
    S_{t-}^{\Gs}
    \left\{
    [\ell_{ii}-\varrho(x_i)]g(x_i)
    +
    \sum_{j\neq i}
    \ell_{ij}e^{\xi(x_i,x_j)}g(x_j)
    \right\}
  \]
  when $X_{t-}=x_i$. Let $\mathcal A$ denote the matrix in
  \eqref{eq:param-generator}. Under the stated integrability conditions, the
  compensated stochastic-integral terms are martingales, so the product formula
  gives
  \[
    S_t^{\Gs}g(X_t)
    =
    g(X_0)
    +\int_0^t S_{u-}^{\Gs}(\mathcal Ag)(X_{u-})\,\diff u
    +\mathcal M_t,
  \]
  where $\mathcal M$ has zero expectation. Consequently,
  $\Pi_tg-g=\int_0^t\Pi_u(\mathcal Ag)\,\diff u$. Dividing by $t$ and letting
  $t\downarrow0$ yields $\Lp g=\mathcal Ag$ for every $g$, which proves the claim.
\end{proof}

Formula~\eqref{eq:param-generator} separates transition and local pricing effects.
For $i\neq j$, the physical transition rate $\ell_{ij}$ is multiplied by
$\e^{\xi(x_i,x_j)}$, the proportional response of the growth-adjusted SDF to that
transition. Continuous shocks and non-transition jumps leave $X$ unchanged and
therefore enter only through the diagonal term $-\varrho(x_i)$.

To separate transition reweighting from local discounting, define
$\ell_{ij}^{\QQ}\coloneqq\ell_{ij}\e^{\xi(x_i,x_j)}$ for $i\neq j$, and let
$\Lx^{\QQ}$ be the Markov generator with these off-diagonal entries. Define the
state-specific \emph{growth-adjusted effective discount rate}
\begin{equation} \label{eq:rho-appendix}
  \rho(x_i)
  \coloneqq
  -(\Lp\1)(x_i)
  =
  \varrho(x_i)
  -\sum_{j\neq i}\ell_{ij}
  \left[\e^{\xi(x_i,x_j)}-1\right].
\end{equation}
The pricing generator then has the Feynman--Kac decomposition
\begin{equation} \label{eq:FK-appendix}
  \Lp=\Lx^{\QQ}-\diag(\rho),
  \qquad
  \Lp\1=-\rho.
\end{equation}
Let $(\QQ_x)_{x\in\Xsf}$ be the Markov family generated by $\Lx^{\QQ}$. The
finite-state Feynman--Kac formula yields
\[
  (\Pi_tg)(x)
  =
  \EE_x^{\QQ}\!\left[
    \exp\!\left(-\int_0^t\rho(X_u)\,\diff u\right)g(X_t)
  \right].
\]
This representation is a matrix identity and does not require an equivalent change
of measure on the original filtered space. Such a change of measure is available if
\(
  M_t^{\QQ}
  \coloneqq
  \exp(\int_0^t\rho(X_{u-})\,\diff u)S_t^{\Gs}
\)
is a true martingale under every $\PP_x$: then
\(
  \left.\frac{\diff\QQ_x}{\diff\PP_x}\right|_{\fF_t}=M_t^{\QQ}
\)
and $X$ has generator
$\Lx^{\QQ}$ under $\QQ_x$. When $G$ is nontrivial, $\QQ$ is a growth-adjusted
pricing law rather than, in general, the usual risk-neutral law, and $\rho$ need
not be the money-market short rate.

The row-sum identity in \eqref{eq:FK-appendix} gives $\rho$ a direct price
interpretation. From $z_t=\e^{t\Lp}\1$,
\[
  z_t(x)=1-\rho(x)t+o(t),
  \qquad
  \lim_{t\downarrow0}y_t(x)=\rho(x).
\]
Thus $\rho(x)$ is the short-horizon yield on the growth-adjusted unit payoff. By
contrast, the long-run envelope yield is $-\lambda^*=-s(\Lp)$. The row-sum bounds
for Metzler matrices imply
\[
  \min_{x\in\Xsf}\rho(x)
  \leq
  -\lambda^*
  \leq
  \max_{x\in\Xsf}\rho(x).
\]
Combining \eqref{eq:FK-appendix} with Theorem~\ref{thm:gs}, global valuation
stability is equivalent to
$s(\Lx^{\QQ}-\diag(\rho))<0$. Hence positive effective discounting in every state
is sufficient, but not necessary: a negative local rate can be offset by sufficiently
rapid departure under the pricing-adjusted dynamics.

\subsubsection{Recovery jumps and recovery intensity}

In the one-way CRRA model of Section~\ref{sec:crra}, the
effective-discounting representation distinguishes the recovery jump multiplier
from the physical recovery intensity. With the primitives of that model,
\eqref{eq:rho-appendix} gives
\[
  \kappa_D^{\Pis}=\kappa\chi_{\sss DN}^{\vartheta-\alpha},\qquad
  \rho_D=\varrho_D-\kappa(\chi_{\sss DN}^{\vartheta-\alpha}-1),\qquad
  \lambdaD=-(\rho_D+\kappa_D^{\Pis})=-\varrho_D-\kappa,
\]
where $\kappa_D^{\Pis}$ is the recovery intensity under the growth-adjusted
pricing law. Changing the recovery multiplier $\chi_{\sss DN}^{\vartheta-\alpha}$ alone changes this intensity and the effective
discount rate by offsetting amounts. It therefore changes the nonzero off-diagonal
entry of the pricing generator but not the disaster-class pricing rate. By contrast,
increasing the physical intensity $\kappa$ lowers $\lambdaD$ one-for-one and
also changes that off-diagonal entry, holding the other primitives fixed. Increasing pricing-adjusted
recovery while holding $\rho_D$ fixed is thus a different comparative static
from changing the recovery jump in this primitive model.
Appendix~\ref{app:local-pass-through} gives the local sensitivity to effective
discounting for general multi-state blocks.

\subsubsection{Separate consumption and cash-flow exposures.}
\label{app:separate-exposures}

The power specification in Section~\ref{sec:crra} is convenient but
not required. Retain the CRRA SDF in \eqref{eq:sdf} and let consumption
and cash-flow growth have separate log drifts $\mu_i^C,\mu_i^G$ and Brownian
exposure vectors $\sigma_i^C,\sigma_i^G\in\RR^d$:
\[
  \diff\log C_t=\mu_i^C\diff t+(\sigma_i^C)^\t\diff W_t,
  \qquad
  \diff\log G_t=\mu_i^G\diff t+(\sigma_i^G)^\t\diff W_t
\]
between transitions in state $i$. Normalize $G_0=1$ and let $\chi_{ij}>0$ and
$\chi_{ij}^G>0$ be the respective consumption and cash-flow multipliers at $i\to j$.
With no non-transition jumps, Proposition~\ref{prop:param-generator} gives
\begin{equation}\label{eq:crra-general-exposures}
  \varrho_i=\beta+\alpha\mu_i^C-\mu_i^G
              -\tfrac12\|\sigma_i^G-\alpha\sigma_i^C\|^2,
  \qquad m_{ij}=\chi_{ij}^{-\alpha}\chi_{ij}^G.
\end{equation}
In the one-way two-state economy, the class rates remain
$\lambdaN=-\varrho_N$ and $\lambdaD=-\varrho_D-\kappa$. Their difference is
\begin{align*}
  \lambdaD-\lambdaN
  &=(\mu_D^G-\mu_N^G)+\alpha(\mu_N^C-\mu_D^C)-\kappa\\
  &\quad+\tfrac12\bigl(\|\sigma_D^G\|^2-\|\sigma_N^G\|^2\bigr)
     +\tfrac12\alpha^2\bigl(\|\sigma_D^C\|^2-\|\sigma_N^C\|^2\bigr)
  -\alpha\bigl[(\sigma_D^C)^\t\sigma_D^G
                         -(\sigma_N^C)^\t\sigma_N^G\bigr].
\end{align*}
The last term records the change in the instantaneous covariance of consumption
and cash-flow growth. This extension separates cash-flow deterioration from
consumption risk compensation. The main-text model sets
$\mu_i^G=\vartheta\mu_i^C$, $\sigma_i^G=\vartheta\sigma_i^C$, and
$\chi_{ij}^G=\chi_{ij}^{\vartheta}$.

\subsection{Long-run sensitivity to effective discounting}
\label{app:local-pass-through}

The decomposition \eqref{eq:FK-appendix} also links the long-run pricing law to
changes in effective discounting. Increasing $\rho$ by $\delta\in\RR^n$ while
holding $\Lx^{\QQ}$ fixed replaces $\Lp$ by $\Lp-\diag(\delta)$.

\begin{proposition}[Effective-discount-rate sensitivity] \label{prop:pass-through}
  Suppose that $\cC^*=\{\Ccal^*\}$, and let $\phi$ and $\psi$ be the right and left
  eigenvectors of $\Lp$ in Theorem~\ref{thm:Pi-asym}. For $\delta\in\RR^n$ near
  zero, define
  \[
    \Lp(\delta)\coloneqq\Lp-\diag(\delta),
    \qquad
    y^*(\delta)\coloneqq-s\bigl(\Lp(\delta)\bigr).
  \]
  Then $y^*$ is differentiable at zero and, for every $h\in\RR^n$,
  \begin{equation} \label{eq:pass-through}
    \left.\frac{\diff}{\diff\theta}y^*(\theta h)\right|_{\theta=0}
    =
    \sum_{i=1}^n\pi_i^{\mathrm{LR}}h_i,
    \qquad
    \pi^{\mathrm{LR}}\coloneqq\phi\odot\psi.
  \end{equation}
  In particular,
  \(
    \partial y^*(0)/\partial\delta_i=\pi_i^{\mathrm{LR}}
  \)
  and
  \(
    \sum_{i=1}^n\pi_i^{\mathrm{LR}}=1.
  \)
  More generally, if a real matrix-valued perturbation satisfies
  \(
    \Lp(\theta)=\Lp+\theta H+o(\theta)
  \)
  in matrix norm, then
  \begin{equation}
  \label{eq:general-sensitivity}
    \left.\frac{\diff}{\diff\theta}s\bigl(\Lp(\theta)\bigr)\right|_{\theta=0}
    =
    \psi^\t H\phi.
  \end{equation}
\end{proposition}

\begin{proof}
By Theorem~\ref{thm:Pi-asym}, $\lambda^*=s(\Lp)$ is algebraically simple and
strictly separated in real part from the remaining spectrum. Its continuation under
a sufficiently small perturbation is therefore differentiable and remains the
spectral bound. Differentiating the eigenvalue equation for
$\Lp(\theta)=\Lp+\theta H+o(\theta)$ at zero and premultiplying by $\psi^\t$ gives
\eqref{eq:general-sensitivity}, because $\psi^\t\Lp=\lambda^*\psi^\t$ and
$\psi^\t\phi=1$. Taking $H=-\diag(h)$ and changing sign yields
\eqref{eq:pass-through}; moreover,
$\sum_i\pi_i^{\mathrm{LR}}=\psi^\t\phi=1$.
\end{proof}

Thus $\pi_i^{\mathrm{LR}}$ is the marginal pass-through from the effective discount rate in
state $x_i$ to the long-run envelope yield, as well as the stationary occupation
weight of $x_i$ under either Doob chain. Proposition~\ref{prop:valuation-laws}
implies that $\supp(\pi^{\mathrm{LR}})=\Ccal^*$, so sufficiently small effective-rate
changes confined to classes outside $\Ccal^*$ leave the global rate unchanged,
although they may alter finite-maturity prices and asymptotic loadings. By contrast,
\eqref{eq:general-sensitivity} shows that changes in transition or risk-adjustment
terms generally involve the two-sided weights $\psi_i\phi_j$, rather than an
expectation under $\pi^{\mathrm{LR}}$.

The same calculation applies claim by claim. Fix a current class $\Ccal$ and a
nonnegative payoff $g$ with nonempty corridor $\kK(\Ccal,g)$, and suppose that this
corridor has a unique critical class $\Ccal_c$. Let $L_c$ be its diagonal block and
let $\phi_c,\psi_c\gg0$ be corresponding right and left Perron vectors, normalized
by $\psi_c^\t\phi_c=1$. If $L_c$ is replaced by
$L_c-\theta\diag(h)$ for $h\in\RR^{\Ccal_c}$, while all other corridor blocks are
held fixed, then $\Ccal_c$ remains critical for all sufficiently small $\theta$ and
\[
  \left.\frac{\diff}{\diff\theta}
  [-\lambda(\Ccal,g;\theta)]\right|_{\theta=0}
  =
  \sum_{x\in\Ccal_c}\psi_c(x)\phi_c(x)h(x),
\]
where $\lambda(\Ccal,g;\theta)$ is the resulting corridor rate.
Sufficiently small effective-rate changes confined to classes outside $\Ccal_c$
leave the corridor rate unchanged. If several corridor-critical classes are tied,
the corridor rate is the maximum of their perturbed block rates; the corresponding
long-run yield need not be differentiable when those rates have different
first-order responses.

\subsection{Rare-entry specialization}
\label{app:rare-entry-primitive}

Because $\e^{\xi(x_i,x_j)}>0$, $\Lp$ and $\Lx$ have the same off-diagonal zero
pattern. Thus the primitive specification preserves every physical edge and, in
particular, the class accessibility relation used in Section~\ref{sec:network}.

The fixed-recovery model in Section~\ref{sec:rare-fixed-recovery} follows directly
from Proposition~\ref{prop:param-generator}. Partition the states into
$\mathcal N$ and $\mathcal D$, hold the other coefficients in
\eqref{eq:param-appendix} fixed, and vary only the physical normal-to-disaster
transition block $F_\epsilon$. Set
\(
  D_\epsilon=\diag(F_\epsilon\1_{\mathcal D})
\)
and
\(
  \Gamma_{ij}=\e^{\xi(x_i,x_j)}
\)
for $x_i\in\mathcal N$ and $x_j\in\mathcal D$. If $\LpN$, $\LpDN$, and $\LpD$
denote the pricing blocks at $\epsilon=0$, then
\[
  \Lx(\epsilon)
  =
  \begin{pmatrix}
    Q-D_\epsilon&F_\epsilon\\
    H&U
  \end{pmatrix},
  \qquad
  \Lp(\epsilon)
  =
  \begin{pmatrix}
    \LpN-D_\epsilon&F_\epsilon\odot\Gamma\\
    \LpDN&\LpD
  \end{pmatrix}.
\]
Thus rare physical entry changes only the normal-state exit diagonals and the
normal-to-disaster pricing block; disaster dynamics and recovery remain fixed. If
$F_\epsilon=\epsilon F+O(\epsilon^{1+\eta_F})$, the pricing entry block is
$\epsilon(F\odot\Gamma)+O(\epsilon^{1+\eta_F})$ and has the same rarity exponent.


For the consumption-based application in Section~\ref{sec:crra-rare},
the physical entry intensity is $\epsilon\ell$ and $\Delta=\lambdaD-\lambdaN>0$.
Physical flow balance gives
$\pi_\epsilon^P(N)\epsilon\ell=\pi_\epsilon^P(D)\kappa$, yielding the first
identity in \eqref{eq:crra-rare-masses}. The pricing generator has the two
distinct eigenvalues in \eqref{eq:crra-rare-eigenvalues}, and its spectral
projection at $\lambda_+(\epsilon)$ is
\[
  \frac{\Lp(\epsilon)-\lambda_-(\epsilon)I}
       {\Delta(\epsilon)}.
\]
The disaster-state diagonal of this rank-one projection is the product of the
corresponding normalized left and right Perron eigenvector entries, hence equals
$\pi_\epsilon^{\mathrm{LR}}(D)$. Substituting the eigenvalues gives the second
identity in \eqref{eq:crra-rare-masses}. Expanding at $\epsilon=0$ gives
\[
  \pi_\epsilon^P(D)=\frac{\ell}{\kappa}\epsilon+O(\epsilon^2),
  \qquad
  \pi_\epsilon^{\mathrm{LR}}(D)
  =1-\frac{\ell\kappa(\chi_{\sss ND}\chi_{\sss DN})^{\vartheta-\alpha}}{\Delta^2}
    \epsilon+O(\epsilon^2).
\]
Thus the constants in Corollary~\ref{cor:rare-disaster-concentration} are
$c_P=\ell/\kappa$ and
$c_{\mathrm{LR}}=\ell\kappa(\chi_{\sss ND}\chi_{\sss DN})^{\vartheta-\alpha}/\Delta^2$;
in particular,
$\pi_\epsilon^{\mathrm{LR}}(N)=c_{\mathrm{LR}}\epsilon+O(\epsilon^2)$.

The spectral decomposition expresses $z_t(N;\epsilon)$ as a linear combination
of $\e^{\lambda_+(\epsilon)t}$ and $\e^{\lambda_-(\epsilon)t}$. Its coefficients
are fixed by
\[
  z_0(N;\epsilon)=1,
  \qquad
  \left.\partial_tz_t(N;\epsilon)\right|_{t=0}
  =\lambdaN+\epsilon\ell(\chi_{\sss ND}^{\vartheta-\alpha}-1),
\]
which give \eqref{eq:crra-rare-price}. In particular,
\[
  \begin{aligned}
  \omega_\epsilon
  &=\frac{\lambdaN+\epsilon\ell(\chi_{\sss ND}^{\vartheta-\alpha}-1)-\lambda_-(\epsilon)}
        {\Delta(\epsilon)}\\
  &=\frac{\lambda_+(\epsilon)-\lambdaD+\epsilon\ell\chi_{\sss ND}^{\vartheta-\alpha}}
        {\Delta(\epsilon)}>0.
  \end{aligned}
\]
With the primitives fixed,
\[
  \begin{aligned}
  \lambda_+(\epsilon)
  &=\lambdaD+\frac{\ell\kappa(\chi_{\sss ND}\chi_{\sss DN})^{\vartheta-\alpha}}{\Delta}
    \epsilon+O(\epsilon^2),\\
  \Delta(\epsilon)
  &=\Delta+\ell\left(1+
    \frac{2\kappa(\chi_{\sss ND}\chi_{\sss DN})^{\vartheta-\alpha}}{\Delta}\right)\epsilon
    +O(\epsilon^2).
  \end{aligned}
\]
These expansions yield \eqref{eq:crra-rare-loading}. Equating the two
exponentials in \eqref{eq:crra-rare-price}, including their coefficients, gives
the exact crossover in \eqref{eq:crra-exact-crossover} when
$0<\omega_\epsilon<1/2$. Finally,
\[
  \log\frac{1-\omega_\epsilon}{\omega_\epsilon}
  =\log(1/\epsilon)-\log\omega+O(\epsilon),
  \qquad
  \Delta(\epsilon)=\Delta+O(\epsilon),
\]
give its stated expansion, since $\epsilon\log(1/\epsilon)\to0$.

These expansions hold with the primitives and $\Delta>0$ fixed and are not
uniform as $\Delta\downarrow0$. Changing primitives affects both the loading
and the eigenvalue gap, so a larger limiting gap need not imply an earlier
finite-entry crossover. Small annual entry probabilities need not make the
first-order loading or crossover approximations accurate. For example, at the
1.75\% annual entry probability and $\kappa=0.16$ in
Figure~\ref{fig:rare-entry-crossover}, $\epsilon\omega\approx1.3473$, whereas
$\omega_\epsilon\approx0.4241$; the coefficient-corrected crossover approximation
is negative, whereas the exact crossover is $2.36$ years. The figure therefore
uses exact formulas throughout.

In the numerical benchmark, the reciprocal consumption multipliers satisfy
$\chi_{\sss ND}\chi_{\sss DN}=1$. The recovery jump reverses the discrete entry
loss, not consumption-growth differences accumulated during disaster. Jump
amplitude cancels from the eigenvalues in \eqref{eq:crra-rare-eigenvalues} but
still affects loadings and crossover maturities; nonreciprocal multipliers can
also affect the eigenvalues.

\section{Notation and Preliminaries}
\label{app:preliminaries}

\subsection{Notation}
\label{app:notation}

For scalar functions $f$ and $h$, we write $f\sim h$, or say that $f$ is
\emph{asymptotic to} $h$, if $f/h\to1$ in the indicated limit, with $h$ nonzero
sufficiently near that limit. We use this notation both as maturity $t\to\infty$
and as the rare-entry parameter $\epsilon\downarrow0$.

The state space of the economic environment is the finite set
$\Xsf=\{x_1,\ldots,x_n\}$. Let $\RR^\Xsf$ denote the vector space of real-valued
functions on $\Xsf$. For $g\in\RR^\Xsf$, write $g\geq0$ if $g(x)\geq0$ for every
$x\in\Xsf$, write $g\gg0$ if $g(x)>0$ for every $x\in\Xsf$, and write $g>0$ if $g\geq0$
and $g\neq0$. The positive cone is $\RR_+^\Xsf \coloneqq \{g\in\RR^\Xsf: g\geq0\}$.
A linear operator $A:\RR^\Xsf\to\RR^\Xsf$ is \emph{positive}, written $A\geq0$, if $A\,
\RR_+^\Xsf\subseteq\RR_+^\Xsf$. Equivalently, $A\geq0$ if and only if $Ag\geq 0$
whenever $g\geq0$.

For $g\in\RR^\Xsf$, the supremum norm is
\(
  \|g\|_\infty\coloneqq\max_{x\in\Xsf}|g(x)|.
\)
For a linear operator $A$ on $\RR^\Xsf$, $\|A\|_\infty$ denotes the operator norm
induced by $\|\cdot\|_\infty$. Elsewhere, unsubscripted norms are fixed norms on the
relevant finite-dimensional spaces, with matrix norms taken to be submultiplicative;
on the shock space $\RR^d$, $\|\cdot\|$ denotes the Euclidean norm.

For each $x\in\Xsf$, let \(e_x \in \RR^\Xsf\) denote the indicator
function of state \(x\): \(e_{x}(x')\coloneqq 1_{\{x\}}(x')\) for \(x'\in\Xsf\).
Throughout the paper, we identify \(\RR^\Xsf\) with \(\RR^n\) using the ordered
basis \((e_{x_1},\ldots,e_{x_n})\), and use the same symbol for a function in
$\RR^\Xsf$ and its coordinate column vector. Thus, for $g\in\RR^\Xsf$, 
\[
  g_i\coloneqq g(x_i),
  \qquad
  g=\sum_{i=1}^n g_i\,e_{x_i}.
\]
For $v\in\RR^\Xsf$, let $\diag(v)$ denote the $n\times n$ diagonal matrix formed
from the coordinates of $v$; that is, the $(i,j)$-th entry of $\diag(v)$ is $v_i$ if
$i=j$ and zero otherwise.
For $u,v\in\RR^\Xsf$, their \emph{Hadamard product}, denoted by $u\odot v$, is the
componentwise product defined by
\(
  (u\odot v)_i\coloneqq u_i v_i
\)
for $i=1,\ldots,n$.
Moreover, for a subset $\mathcal X\subseteq\Xsf$, let $v_{\mathcal X}$ denote the
coordinate subvector indexed by $\mathcal X$. For subsets
$\mathcal X,\mathcal Y\subseteq\Xsf$, let $A(\mathcal X,\mathcal Y)$ denote the
submatrix of $A$ with rows indexed by $\mathcal X$ and columns indexed by
$\mathcal Y$; in particular, $A(\mathcal X,\mathcal X)$ is the principal submatrix
indexed by $\mathcal X$.

We similarly use the same symbol for a linear operator on \(\RR^\Xsf\) and its matrix
representation. For \(A:\RR^\Xsf\to\RR^\Xsf\), the \((i,j)\)-th matrix entry is
\[
  A_{ij}\coloneqq(Ae_{x_j})(x_i),
  \qquad i,j=1,\ldots,n.
\]
Hence, for \(g\in\RR^\Xsf\), \((Ag)(x_i)=(Ag)_i=\sum_{j=1}^n A_{ij}g_j\) for
\(i=1,\ldots,n\).
Under this identification, $A$ is positive if and only if its matrix is entrywise
nonnegative, that is, $A\geq0$ if and only if \(A_{ij}\geq0\) for all
$i,j=1,\ldots,n$. Moreover, we write $A\gg0$ whenever $A_{ij}>0$ for all
$i,j=1,\ldots,n$.

\subsection{Spectral preliminaries}
\label{app:spectral}

For a finite-dimensional linear operator $A$, we use $\sigma(A)$ to denote its
spectrum. An eigenvalue in the spectrum is \emph{semisimple} if its algebraic and
geometric multiplicities coincide, or equivalently, if every Jordan block associated
with it has size one. It is \emph{algebraically simple} if its algebraic multiplicity
is one. Since geometric multiplicity cannot exceed algebraic multiplicity, every
algebraically simple eigenvalue is semisimple, but the converse need not hold.

Furthermore, we use \( s(A)\coloneqq\max\{\Re\lambda:\lambda\in\sigma(A)\} \)
to denote the spectral bound of $A$.

\begin{lemma}[Spectral bound of a Metzler matrix]
\label{lem:msb}
If $A$ is Metzler, meaning that $A_{ij}\geq0$ for $i\neq j$, then $s(A)$ is a real
eigenvalue of $A$ and is the unique spectral value with maximal real part.
Equivalently,
\( s(A)\in\sigma(A)\cap\RR \) and \( \Re\mu<s(A) \) for every 
\( \mu\in\sigma(A)\setminus\{s(A)\}. \)
\end{lemma}

\begin{proof}[Proof of Lemma~\ref{lem:msb}]
  See Theorem 7.2 of \citet[][p. 82]{batkai2017positive}.
\end{proof}

For $z\notin\sigma(A)$, the \emph{resolvent} of $A$ is
$R(z;A)\coloneqq(zI-A)^{-1}$.
The \emph{Riesz spectral projection} of $A$ associated with an eigenvalue
$\lambda\in\sigma(A)$ is defined as
\[
  P_\lambda(A)
  \coloneqq
  \frac{1}{2\pi\mathrm i} \int_{\Gamma_\lambda} R(z;A)\diff z,
\]
where $\Gamma_\lambda$ is a positively oriented simple closed contour that encloses
$\lambda$ and no other point of $\sigma(A)$. $P_\lambda(A)$ is a projection
satisfying $P_\lambda(A)^2=P_\lambda(A)$. It is independent of the choice of
$\Gamma_\lambda$ and commutes with $A$. When the operator is clear from context,
we suppress its argument and write $P_\lambda$. In general, the range and rank of
\(P_\lambda(A)\) are the generalized eigenspace associated with \(\lambda\) and the
algebraic multiplicity of \(\lambda\), respectively. If $\lambda$ is algebraically 
simple and $\phi$ and $\psi$ are corresponding right and left eigenvectors normalized
by $\psi^\t\phi=1$, then \(P_\lambda(A)=\phi\psi^\t\).

\subsection{Graph structure and Frobenius form of Metzler matrices}
\label{app:metzler-graphs}

We collect several graph-theoretic concepts used in combinatorial matrix theory and
finite-state Markov-chain analysis. Much of this subsection is adapted from
\cite{rothblum1975algebraic} and \cite{schneider1986influence}; see also
\citet[][Section 3.1]{beare2026reducible} for a review.

Let \(A:\RR^\Xsf\to\RR^\Xsf\) be a linear operator whose matrix representation
\(A=(A_{ij})\) is Metzler, meaning that \(A_{ij}\geq0\) for \(i\neq j\). 
The \emph{directed graph} associated with \(A\), denoted by \(\Gcal(A)\),
records the off-diagonal zero pattern of $A$. It has vertex set \(\Xsf =
\{x_1,\ldots,x_n\}\) and a directed edge from \(x_i\) to \(x_j\), written $x_i\to
x_j$, if and only if \(i\neq j\) and \(A_{ij}>0\).

A \emph{directed walk} of length $p\in\NN_0$ is a sequence $(x_{i_0},\ldots,x_{i_p})$
of vertices of $\Gcal(A)$ such that $x_{i_{r-1}}\to x_{i_r}$ for every
$r=1,\ldots,p$. Directed walks with length zero are allowed. We say that \(x_i\) has
\emph{access to} \(x_j\), or equivalently that \(x_j\) is \emph{accessible from}
\(x_i\), written \(x_i\leadsto x_j\), if there exists a directed walk from \(x_i\) to
\(x_j\). Thus every state has access to itself through the walk of length zero. A
\emph{directed path} is a directed walk with no repeated vertices. By deleting cycles
from a walk, any directed walk from $x_i$ to $x_j$ can be reduced to a directed path
with the same endpoints. Hence $x_i\leadsto x_j$ if and only if there is a directed
path from $x_i$ to $x_j$.

We say that two states $x_i,x_j\in\Xsf$ \emph{communicate}, written $x_i\sim x_j$, if
$x_i\leadsto x_j$ and $x_j\leadsto x_i$. Thus, communication is mutual accessibility.
The relation $\sim$ is an equivalence relation and partitions $\Xsf$ into disjoint
equivalence classes $\Ccal_1,\ldots,\Ccal_m\subseteq\Xsf$, for some $m\in\NN$.
These equivalence classes are called \emph{classes}. Each class is a maximal subset
of states whose elements communicate pairwise; equivalently, the classes are the
\emph{strongly connected components} of $\Gcal(A)$.

A class $\Ccal\subseteq\Xsf$ is \emph{closed}, also called \emph{final} in
combinatorial matrix theory, if no state in $\Ccal$ has access to any state outside
$\Ccal$. Equivalently, $A_{ij}=0$ whenever $x_i\in\Ccal$ and
$x_j\notin\Ccal$. Thus, closed classes are precisely the \emph{terminal strongly
connected components} of $\Gcal(A)$.
When \(A\) is a Markov generator, this graph-theoretic
classification coincides with the usual probabilistic one: a state is
\emph{recurrent} if and only if it belongs to a closed class, and \emph{transient} if
and only if it belongs to a nonclosed class.

If $\Xsf$ consists of a single class, then that class is necessarily closed.
The graph $\Gcal(A)$ is then strongly connected, and the Metzler matrix $A$
is called \emph{irreducible}. In addition to the general spectral conclusion of
Lemma~\ref{lem:msb}, Perron--Frobenius theory for irreducible Metzler matrices
implies that $s(A)$ is an algebraically simple eigenvalue and there
exist strictly positive left and right eigenvectors associated with $s(A)$.

The \emph{reduced graph} of \(A\), also called its \emph{condensation graph} and
denoted by \(\Gcal_{\mathrm r}(A)\), is obtained by collapsing each class into a
single vertex while retaining all directed connections between distinct classes.
Its vertex set is therefore \(\cC(A)\coloneqq\{\Ccal_1,\ldots,\Ccal_m\}\), and it
contains a directed edge \( \Ccal_k\to\Ccal_\ell \) if and only if \(k\neq\ell\) and
\(x_i\to x_j\) for some \(x_i\in\Ccal_k\) and \(x_j\in\Ccal_\ell\).

In the reduced graph, a \emph{directed path} of length $p\in\NN_0$ is a sequence 
$\gamma = (\Ccal_{k_0},\ldots,\Ccal_{k_p})$ of distinct vertices
of $\Gcal_{\mathrm r}(A)$ such that $\Ccal_{k_{r-1}}\to\Ccal_{k_r}$ for every $r=1,\ldots,p$.
We write $|\gamma|=p$ for the length of $\gamma$, that is, its number of directed edges.
In particular, paths of length zero are allowed.
We say that $\Ccal_k$ has \emph{access to}
$\Ccal_\ell$, written $\Ccal_k\leadsto\Ccal_\ell$, if there exists a directed
path from $\Ccal_k$ to $\Ccal_\ell$. Allowing paths of length zero makes each
class have access to itself. In particular, class-level accessibility agrees with
state-level accessibility:
\[
  \Ccal_k\leadsto\Ccal_\ell
  \quad\Longleftrightarrow\quad
  x_i\leadsto x_j
  \text{ for some }x_i\in\Ccal_k,\ x_j\in\Ccal_\ell.
\]
Because all states within a class communicate, ``some'' may equivalently be
replaced by ``every.''

State-level accessibility is reflexive and transitive on $\Xsf$ but need not be
antisymmetric, because distinct states in the same class are mutually accessible.
It therefore defines a preorder on $\Xsf$. After quotienting by communication,
class-level accessibility is antisymmetric and hence defines a partial order on
$\cC(A)$. Equivalently, the reduced graph $\Gcal_{\mathrm r}(A)$ is acyclic.

Two distinct classes $\Ccal_k,\Ccal_\ell\in\cC(A)$ are said to be \emph{comparable}
if $\Ccal_k\leadsto\Ccal_\ell$ or $\Ccal_\ell\leadsto\Ccal_k$; otherwise they are
\emph{incomparable}. Following \citet{rothblum1975algebraic}, a \emph{chain of
classes} is a nonempty collection of pairwise comparable classes; in particular, any
singleton class forms a chain. Equivalently, its elements can be ordered as
\(
  \Ccal_{k_0}\leadsto\Ccal_{k_1}\leadsto\cdots\leadsto\Ccal_{k_q},
\)
for some $q\in\NN_0$.

Since the class-level accessibility defines a partial order on $\cC(A)$, we may label
the classes so that
\begin{equation}\label{eq:po}
  \Ccal_k\leadsto\Ccal_\ell \text{ for } k\neq\ell
  \quad\Longrightarrow\quad
  k>\ell.
\end{equation}
After arranging the state coordinates class by class, with an arbitrary ordering
within each class, \(A\) has the block lower-triangular Frobenius form
\begin{equation}\label{eq:aff}
  A=
  \begin{pmatrix}
    A_{11} & 0      & \cdots & 0 \\
    A_{21} & A_{22} & \cdots & 0 \\
    \vdots & \vdots & \ddots & \vdots \\
    A_{m1} & A_{m2} & \cdots & A_{mm}
  \end{pmatrix},
\end{equation}
where \(A_{k\ell}=A(\Ccal_k,\Ccal_\ell)\) is the submatrix with rows indexed by
\(\Ccal_k\) and columns indexed by \(\Ccal_\ell\). 
The ordering of the blocks in the Frobenius form is generally nonunique because
some classes may be incomparable under accessibility.

Each diagonal block \(A_{kk}\) in \eqref{eq:aff} is Metzler and irreducible because
its associated graph is the strongly connected subgraph induced by the class
\(\Ccal_k\). For \(k\neq\ell\), the off-diagonal block \(A_{k\ell}\) is nonzero if
and only if \(\Ccal_k\to\Ccal_\ell\). Due to \eqref{eq:po}, every nonzero
off-diagonal block lies strictly below the diagonal. Moreover, \(\Ccal_k\) is closed
if and only if \(A_{k\ell}=0\) for every \(\ell<k\).


\subsection{Positivity and class-path expansions}

A fundamental property of Metzler matrices is that they generate positive
semigroups: if \(A\) is Metzler, then \(\e^{tA}\geq0\) for every \(t\geq0\);
see, for example, \citet[Theorem~7.1]{batkai2017positive}. The following lemma
strengthens this result by characterizing the entrywise positivity pattern of
\(\e^{tA}\) in terms of accessibility in \(\Gcal(A)\).
For any $g\in\RR^\Xsf$, let \( \supp(g) \coloneqq \{x\in\Xsf:g(x)\neq0\} \subseteq
\Xsf \) denote its support.

\begin{lemma} \label{lem:mr}
  Let \(A:\RR^\Xsf\to\RR^\Xsf\) be a linear operator whose matrix representation
  is Metzler, and let \(\Gcal(A)\) be the associated directed graph.
  For \(i,j=1,\ldots,n\), the following statements are equivalent:
  \begin{enumerate}
    \item $x_i\leadsto x_j$ in  $\Gcal(A)$.
    \item $(\e^{t A})_{ij}>0$ for every $t>0$.
    \item $(\e^{t A})_{ij}>0$ for some $t>0$.
  \end{enumerate}
  Moreover, for every $t>0$, \(g\in\RR_+^\Xsf\), and \(x\in\Xsf\), 
  $(\e^{tA} g)(x) > 0$ if and only if there exists $y\in\supp(g)$ such that
  $x\leadsto y$ in $\Gcal(A)$.
\end{lemma}

\begin{proof}[Proof of Lemma~\ref{lem:mr}]
  Choose a scalar $a\geq0$ sufficiently large that $B=A+aI\geq0$. Since
  $A$ and $I$ commute, $\e^{tA}=\e^{-at}\e^{tB}$.
  Hence, for every $t>0$, $(\e^{tA})_{ij}>0$ if and only if
  $(\e^{tB})_{ij}>0$.
  Moreover,
  \(
    \e^{tB} = \sum_{k=0}^\infty t^k B^k / k!.
  \)
  Because \(B\geq0\), every term in this series is entrywise nonnegative. Hence
  for any $t>0$, \((\e^{tB})_{ij}>0\) if and only if
  \((B^k)_{ij}>0\) for some $k\in\NN_0$.

  We first prove ${\rm (i)}\Rightarrow{\rm (ii)}$. Suppose that
  \(x_i\leadsto x_j\), and let
  \(
    x_{r_0}\to x_{r_1}\to\cdots\to x_{r_k}
  \)
  be a directed walk from $x_i=x_{r_0}$ to $x_j=x_{r_k}$. If $k=0$, then
  $(B^0)_{ij}=1$. If $k\geq1$, then, since \(B\) and \(A\) have the same
  off-diagonal entries,
  \(
    B_{r_{\ell-1} r_\ell} = A_{r_{\ell-1} r_\ell} >0,
  \)
  for $\ell=1,\ldots,k$. The expansion
  \begin{equation}\label{eq:Bk}
    (B^k)_{ij} 
    =
    \sum_{s_1,\,\ldots,\,s_{k-1} = 1}^n 
    B_{i s_1}B_{s_1s_2}\cdots B_{s_{k-1}j},
  \end{equation}
  therefore gives
  \(
    (B^k)_{ij} \geq B_{r_0 r_1} B_{r_1 r_2}\cdots B_{r_{k-1} r_k} >0.
  \)
  In either case, $(B^k)_{ij}>0$ for some $k\in\NN_0$, and therefore
  \((\e^{tA})_{ij}>0\) for every $t>0$.

  The implication ${\rm (ii)}\Rightarrow{\rm (iii)}$ is immediate.
  It remains to prove ${\rm (iii)}\Rightarrow{\rm (i)}$. Suppose that
  \((\e^{tA})_{ij}>0\) for some $t>0$. Then \((B^k)_{ij}>0\) for some
  \(k\in\NN_0\). If \(k=0\), necessarily \(i=j\), so the length-zero path gives
  \(x_i\leadsto x_j\). If \(k\geq1\), the nonnegativity of the summands in
  \eqref{eq:Bk} implies that there exist indices \(s_1,\ldots,s_{k-1}\) such that
  \(
    B_{s_0s_1}B_{s_1s_2}\cdots B_{s_{k-1}s_k}>0,
  \)
  where $s_0 \equiv i$ and $s_k \equiv j$.
  Whenever \(s_{\ell-1}\neq s_\ell\), \( A_{s_{\ell-1}s_\ell} =
  B_{s_{\ell-1}s_\ell}>0\), so \(x_{s_{\ell-1}}\to x_{s_\ell}\) is an edge of
  \(\Gcal(A)\). Removing all consecutive repetitions \(s_{\ell-1}=s_\ell\) therefore
  produces a directed walk from \(x_i\) to \(x_j\) in \(\Gcal(A)\), possibly of
  length zero. Hence \(x_i\leadsto x_j\).

  For the last statement, note that
  \(
    (\e^{tA} g)(x_i)
    =
    \sum_{j=1}^n (\e^{tA})_{ij} g(x_j)
    =
    \sum_{x_j\in\supp(g)}(\e^{tA})_{ij} g(x_j)
  \)
  for some \(i=1,\ldots,n\), $t>0$ and $g\geq0$.
  Because \(\e^{tA}\geq0\) and \(g(x_j)>0\) for every \(x_j\in\supp(g)\), the sum is
  positive if and only if \( (\e^{tA})_{ij}>0 \) for some \( x_j\in\supp(g) \). The
  statement then follows from the equivalence between (i) and (ii).
\end{proof}

When $A$ is written in the Frobenius form \eqref{eq:aff}, Lemma~\ref{lem:mr}
also shows that $\e^{tA}$ is block lower triangular for every $t\geq0$. 
The claim is immediate at $t=0$; if $t>0$ and $k<\ell$, then
$\Ccal_k\not\leadsto\Ccal_\ell$ by \eqref{eq:po}, and hence
$\bigl(\e^{tA}\bigr)_{\s\Ccal_k,\Ccal_\ell}=0$. Its block zero pattern need not,
however, coincide with that of $A$: an indirect path from $\Ccal_k$ to
$\Ccal_\ell$ makes the corresponding block of $\e^{tA}$ entrywise strictly
positive for every $t>0$, even when $A_{k\ell}=0$.
Lastly, if $A$ is Metzler and irreducible, all states are mutually accessible, so
Lemma~\ref{lem:mr} implies that $\e^{tA}\gg0$ for every $t>0$.

We next derive a path expansion for the resolvent of a Metzler matrix $A$ in the
Frobenius form \eqref{eq:aff}. Because $A$ is block triangular, $z\notin\sigma(A)$
also lies outside the spectrum of every diagonal block $A_{kk}$. Hence
$R(z;A)$ and $R(z;A_{kk})$, $k=1,\ldots,m$, are all well defined.
For $\Ccal,\Ccal'\in\cC(A)$, let $\dD(\Ccal,\Ccal')$ be the set of all directed paths
in $\Gcal_{\mathrm r}(A)$ from $\Ccal$ to $\Ccal'$.
For a directed path $\gamma=(\Ccal_{k_0},\ldots,\Ccal_{k_p})$, define its
\emph{resolvent contribution} by
\begin{equation} \label{eq:path-res}
  \mathcal R_\gamma(z;A)
  \coloneqq
  R(z;A_{k_0k_0}) A_{k_0k_1} R(z;A_{k_1k_1}) \cdots A_{k_{p-1}k_p} R(z;A_{k_pk_p}).
\end{equation}
For a path of length zero, we have $\mathcal R_{(\Ccal_k)}(z;A) = R(z;A_{kk})$.

\begin{lemma}[Resolvent path expansion] \label{lem:res-exp}
  Let $A:\RR^\Xsf\to\RR^\Xsf$ be a linear operator whose matrix representation is
  Metzler in its Frobenius form \eqref{eq:aff}. Let $\mathcal R_\gamma(z;A)$ be the
  resolvent contribution from a path $\gamma$ in $\Gcal_{\mathrm r}(A)$, defined
  in \eqref{eq:path-res}.
  For any $\Ccal,\Ccal'\in\cC(A)$, the submatrix of the resolvent $R(z;A)$ with rows
  indexed by $\Ccal$ and columns indexed by $\Ccal'$ is given by
  \begin{equation} \label{eq:bre}
    \big[R(z;A)\big]_{\s\Ccal,\Ccal'}
    =
    \sum_{\gamma\in\dD(\Ccal,\,\Ccal')}
    \mathcal R_\gamma(z;A),
    \qquad z\in\CC\setminus\sigma(A),
  \end{equation}
  with the convention that the sum is zero when $\dD(\Ccal,\Ccal')=\varnothing$.
  Moreover, if $z>s(A)$ is real, every resolvent contribution
  $\mathcal R_\gamma(z;A)$ is entrywise strictly positive, and
  \begin{equation} \label{eq:bre2}
    \big[R(z;A)\big]_{\s\Ccal,\Ccal'} 
    \begin{cases}
      \gg 0, & \text{if }\Ccal\leadsto\Ccal',\\
      = 0,   & \text{if }\Ccal\not\leadsto\Ccal'.
    \end{cases}
  \end{equation} 
\end{lemma}

\begin{proof}[Proof of Lemma~\ref{lem:res-exp}]
  Write $A=D+K$, where $D=\diag(A_{11},\ldots,A_{mm})$ and $K=A-D$. Then
  \begin{equation}\label{eq:RzD}
    R(z;D) = \diag\bigl(R(z;A_{11}),\ldots,R(z;A_{mm})\bigr),
    \qquad z\in\CC\setminus\sigma(A).
  \end{equation}
  Since $A$ is block triangular, $\sigma(A)=\bigcup_{k=1}^m\sigma(A_{kk})$, so
  $R(z;D)$ is well defined whenever $z\notin\sigma(A)$.

  Since $R(z;D)$ is block diagonal and $K$ is strictly lower block
  triangular---that is, block lower triangular with zero diagonal blocks---the
  product $KR(z;D)$ is also strictly lower block triangular. For a strictly lower
  block triangular matrix, each successive power moves its potentially nonzero blocks
  at least one block farther below the diagonal.
  Since there are only $m$ diagonal blocks, $\bigl(KR(z;D)\bigr)^m=0$.
  Thus $KR(z;D)$ is nilpotent, with nilpotency index at most $m$.

  Taking inverses in the identity \(zI-A=\bigl(I-KR(z;D)\bigr)(zI-D)\) and using
  the nilpotency of $KR(z;D)$ yields the finite resolvent expansion
  \begin{equation} \label{eq:neumann}
    R(z;A)
    =
    R(z;D)\bigl(I-KR(z;D)\bigr)^{-1}
    =
    R(z;D)\sum_{p=0}^{m-1}\bigl(KR(z;D)\bigr)^p,
    \qquad z\in\CC\setminus\sigma(A).
  \end{equation}
  Fix $z\in\CC\setminus\sigma(A)$, and abbreviate $R_{\s D}\coloneqq R(z;D)$.
  It follows from \eqref{eq:neumann} and \eqref{eq:RzD} that the $(k,\ell)$ block of
  $R(z;A)$ is
  \begin{equation}\label{eq:RzA}
    \big[R(z;A)\big]_{k\ell} = \sum_{p=0}^{m-1} \mathcal R_{k\ell}(p),
  \end{equation}
  where
  \begin{equation}\label{eq:Rkl}
    \mathcal R_{k\ell}(p)
    \coloneqq 
    \big[R_{\s D}(KR_{\s D})^p\big]_{k\ell} 
    =
    R(z;A_{kk})\big[(KR_{\s D})^p\big]_{k\ell}.
  \end{equation}

  We claim that, for $p=0,\ldots, m-1$,
  \begin{equation}\label{eq:Rklp}
    \mathcal R_{k\ell}(p)
    =
    \sum_{\substack{
      \gamma\in\dD(\Ccal_k,\Ccal_\ell)\\
      |\gamma|=p
    }}
    \mathcal R_\gamma(z;A).
  \end{equation}
  Together with \eqref{eq:RzA}, identity \eqref{eq:Rklp} yields \eqref{eq:bre}.

  We prove \eqref{eq:Rklp} by induction on $p$, beginning with the case $p=0$. 
  Under the length-zero-path convention, the right-hand side of \eqref{eq:Rklp} gives
  $R(z;A_{kk})$ if $k=\ell$ and zero otherwise. This agrees with $\mathcal
  R_{k\ell}(0)$ as defined in \eqref{eq:Rkl}. 

  Now suppose that \eqref{eq:Rklp} holds for some $0\leq p\leq m-2$. It follows from
  \eqref{eq:Rkl} that
  \begin{align*}
    \mathcal R_{k\ell}(p+1)
    =
    \big[R_{\s D}(KR_{\s D})^p(KR_{\s D})\big]_{k\ell} 
    =
    \sum_{s=1}^m \big[R_{\s D}(KR_{\s D})^p\big]_{ks} \big[KR_{\s D}]_{s\ell}
    =
    \sum_{s=1}^m \mathcal R_{ks}(p) \big[KR_{\s D}]_{s\ell}.
  \end{align*}
  Since $R_{\s D}$ is block diagonal and $K$ is strictly lower block triangular,
  \[
    \bigl[KR_{\s D}\bigr]_{s\ell}
    =
    K_{s\ell}R(z;A_{\ell\ell})
    =
    \begin{cases}
      A_{s\ell}R(z;A_{\ell\ell}), & s>\ell,\\
      0,                           & s\leq\ell.
    \end{cases}
  \]
  Therefore, using the induction hypothesis,
  \begin{align*}
    \mathcal R_{k\ell}(p+1)
    =
    \sum_{s>\ell}
    \mathcal R_{ks}(p)A_{s\ell}R(z;A_{\ell\ell}) 
    =
    \sum_{s>\ell}
    \sum_{\substack{
      \gamma\in\dD(\Ccal_k,\Ccal_s)\\
      |\gamma|=p
    }}
    \mathcal R_\gamma(z;A) A_{s\ell}R(z;A_{\ell\ell}).
  \end{align*}
  Appending the edge $\Ccal_s\to\Ccal_\ell$ to a length-$p$ path from $\Ccal_k$ to
  $\Ccal_s$ produces a length-$(p+1)$ path from $\Ccal_k$ to $\Ccal_\ell$.
  Conversely, every length-$(p+1)$ path from $\Ccal_k$ to $\Ccal_\ell$ has a unique
  final edge $\Ccal_s\to\Ccal_\ell$ and is obtained by appending this edge to its
  length-$p$ initial segment. Consequently, the product in the summation is precisely
  $\mathcal R_{\gamma'}(z;A)$ with some $\gamma'\in\dD(\Ccal_k,\Ccal_\ell)$ and
  $|\gamma'|=p+1$. Hence
  \[
    \mathcal R_{k\ell}(p+1)
    =
    \sum_{\substack{
      \gamma\in\dD(\Ccal_k,\Ccal_\ell)\\
      |\gamma|=p+1
    }}
    \mathcal R_\gamma(z;A).
  \]
  This completes the induction and proves \eqref{eq:Rklp}.

  Finally, let $z>s(A)$ be real. Since $z>s(A_{kk})$ for every $k=1,\ldots,m$,
  each diagonal-block resolvent admits the Laplace representation 
  \citep[e.g.,][Proposition~9.33(a)]{batkai2017positive}
  \( R(z;A_{kk}) = \int_0^\infty \e^{-zt}\e^{tA_{kk}}\,\diff t.\)
  Because $A_{kk}$ is irreducible and Metzler, Lemma~\ref{lem:mr} implies that 
  $\e^{tA_{kk}}\gg0$ for every $t>0$. Hence $R(z;A_{kk})\gg0$. Moreover, 
  every coupling block along a directed path is nonzero and nonnegative. Alternating
  these coupling blocks with the strictly positive diagonal-block resolvents in
  \eqref{eq:path-res} shows that every resolvent contribution
  $\mathcal R_\gamma(z;A)$ is entrywise strictly positive. The expansion
  \eqref{eq:bre}, together with the empty-sum convention, now yields
  \eqref{eq:bre2}.
\end{proof}

Equation~\eqref{eq:bre} implies that the $(k,\ell)$ block of $R(z;A)$ vanishes
whenever $\dD(\Ccal_k,\Ccal_\ell)=\varnothing$. In particular, because the
Frobenius form \eqref{eq:aff} is block lower triangular, there is no directed path
from $\Ccal_k$ to $\Ccal_\ell$ when $k<\ell$. Hence $R(z;A)$ is block lower
triangular for every $z\in\CC\setminus\sigma(A)$.

By the Perron--Frobenius properties recorded in
Appendix~\ref{app:metzler-graphs}, each irreducible block $L_k$ has right and left
Perron vectors $\phi_k,\psi_k\gg0$, which we normalize by
$\psi_k^\t\phi_k=1$. Since $\lambda_k=s(L_k)$ is algebraically simple,
\begin{equation} \label{eq:block-laurent}
  R(z;L_k)
  =
  \frac{\phi_k\psi_k^\t}{z-\lambda_k} +H_k(z),
\end{equation}
where $H_k$ is analytic near $\lambda_k$. Moreover, Lemma~\ref{lem:mr} and the
semigroup--resolvent formula imply that $R(\mu;L_k)\gg0$ whenever
$\mu>\lambda_k$; this is also the one-class case of
Lemma~\ref{lem:res-exp}.

\section{The Irreducible Hansen--Scheinkman Benchmark}
\label{app:hs}

Suppose that the pricing generator $\Lp$ is irreducible. Perron--Frobenius theory
then implies that $\lambda^*=s(\Lp)$ is algebraically simple and admits strictly
positive right and left eigenvectors. Let $\phi,\psi\gg0$ be normalized by
$\1^\t\phi=1$ and $\psi^\t\phi=1$, so that
\[
  \Lp\phi=\lambda^*\phi,
  \qquad
  \Lp^\t\psi=\lambda^*\psi.
\]
These normalizations uniquely determine both vectors. The Riesz projection at
$\lambda^*$ is $P_{\lambda^*}=\phi\psi^\t$, and
\[
  \lim_{t\to\infty}\e^{-\lambda^*t}\Pi_t
  =
  P_{\lambda^*}.
\]
Indeed, Theorem~7.6(a) and Corollary~7.5 of
\citet[][p.~84]{batkai2017positive}, applied to the irreducible Metzler matrix
$\Lp-\lambda^*I$, identify this limit with its spectral projection at zero.
Consequently, for every $g\in\RR^\Xsf$,
\[
  \lim_{t\to\infty}\e^{-\lambda^*t}(\Pi_tg)(x)
  =
  \phi(x)\,\psi^\t g,
  \qquad x\in\Xsf.
\]
The right eigenvector $\phi$ is the state-dependent long-run pricing factor, while
$\psi^\t g$ is the payoff's loading on the dominant component. Since both
eigenvectors are strictly positive, every nonzero $g\geq0$ has a positive dominant
component from every initial state and hence the common long-run yield
$-\lambda^*$.

The right eigenvector also gives the finite-state Hansen--Scheinkman factorization
\citep{hansen2009long}. Since $\Pi_t\phi=\e^{\lambda^*t}\phi$, the process
\[
  M_t^\phi
  \coloneqq
  \e^{-\lambda^*t}S_t^{\Gs}
  \frac{\phi(X_t)}{\phi(X_0)},
  \qquad t\geq0,
\]
is a strictly positive unit-mean martingale under every $\PP_x$, and the
growth-adjusted SDF factors as
\begin{equation} \label{eq:HS-SG}
  S_t^{\Gs}
  =
  \e^{\lambda^*t}
  \frac{\phi(X_0)}{\phi(X_t)}M_t^\phi.
\end{equation}
The factors in \eqref{eq:HS-SG} separate the long-run exponential rate, a bounded
state adjustment, and a martingale change of measure. Defining
$\diff\PP_x^\phi/\diff\PP_x|_{\fF_t}=M_t^\phi$ reweights state histories by their
importance for long-run pricing.

Equivalently, the Doob $h$-transform of $\Pi$ with $h=\phi$
\citep[see, e.g.,][]{doob1957conditional} is
\begin{equation} \label{eq:hs-doob-semigroup}
  P_t^\phi
  \coloneqq
  \e^{-\lambda^*t}\diag(\phi)^{-1}\Pi_t\diag(\phi),
  \qquad t\geq0.
\end{equation}
It is a Markov semigroup and is the transition semigroup of $X$ under
$(\PP_x^\phi)_{x\in\Xsf}$. Its generator and unique stationary distribution are
\[
  L^\phi
  =
  \diag(\phi)^{-1}(\Lp-\lambda^*I)\diag(\phi),
  \qquad
  \pi^\phi
  =
  \phi\odot\psi.
\]
Thus the transformed generator incorporates physical transition possibilities,
discounting, cash-flow growth, and risk adjustment through the long-run pricing
eigenfunction. The construction relies on $\phi\gg0$: under reducibility the right
eigenvector can vanish, so neither the density nor the whole-space Doob transform is
generally defined. Section~\ref{sec:unique-dominance} instead uses a transform
on the exact support of the dominant eigenvector; Appendix~\ref{app:doob}
develops this transform and the corresponding adjoint construction.
In the notation of Appendix~\ref{app:doob}, irreducibility gives
$\Lp^\rightarrow=L^\phi$ and
\[
  \Lp^\leftarrow
  =
  \diag(\pi^\phi)^{-1}(\Lp^\rightarrow)^\t\diag(\pi^\phi).
\]
Thus $\Lp^\leftarrow$ is the stationary time reversal of the transformed pricing
dynamics, not of the original physical dynamics generated by $\Lx$.
\citet[Section~2.1, p.~1506]{christensen2017nonparametric} calls the adjoint pricing
operator \emph{time reversed}. Here \emph{graph reversal} refers only to edge
reversal; under reducibility, the two transforms generally act on different
supports and are not global time reversals.

\section{Proofs}
\label{app:proofs}

\subsection{Global valuation stability} \label{app:gvs}

\begin{proof}[Proof of Proposition~\ref{prop:ge}]
  Because $\Pi_t$ is positive and the growth-adjusted SDF $S^{\Gs}$ is strictly
  positive,
  \(
    z_t^{\mathrm{env}} = \|\Pi_t\1\|_\infty = \|\Pi_t\|_\infty > 0
  \)
  for all $t\geq0$.
  Finite-dimensional semigroup theory \citep[see, e.g.][p. 48, Proposition
  4.7]{batkai2017positive} gives
  \(
    \lim_{t\to\infty}t^{-1}\log\|\Pi_t\|_\infty = s(\Lp)=\lambda^*.
  \)
\end{proof}

\begin{corollary}[Ex ante pricing rate] \label{cor:ex-ante-rate}
  Suppose that the initial state has a full-support probability distribution $\mu$
  on $\Xsf$. Define the ex ante growth-adjusted zero-coupon price and yield by
  $z_t^\mu \coloneqq \sum_{x\in\Xsf}\mu(x)z_t(x)$ for $t\geq0$ and
  $y_t^\mu \coloneqq -t^{-1}\log z_t^\mu$ for $t>0$.
  Under Assumptions~\ref{ass:markov}--\ref{ass:continuity},
  \[
    \lim_{t\to\infty}
    \frac{1}{t}\log z_t^\mu
    =
    \lambda^*,
    \qquad
    \lim_{t\to\infty}
    y_t^\mu
    =
    -\lambda^*.
  \]
\end{corollary}

\begin{proof}[Proof of Corollary~\ref{cor:ex-ante-rate}]
  Since $\mu$ has full support, we have $\mu_\circ\coloneqq\min_x\mu(x)>0$.
  It follows that
  \(
    z_t^\mu
    =
    \sum_{x\in\Xsf}\mu(x)z_t(x) 
    =
    \sum_{x\in\Xsf}\mu(x)(\Pi_t\1)(x) 
    \leq  
    \max_{x\in\Xsf}(\Pi_t\1)(x)
    =
    \|\Pi_t\1\|_\infty
    =
    z_t^{\mathrm{env}}.
  \)
  On the other hand, let $x_t^*\in\Xsf$ denote a state such that
  $z^{\mathrm{env}}_t=z_t(x_t^*)$. Then,
  \(
    z_t^\mu
    =
    \sum_{x\in\Xsf}\mu(x)z_t(x) 
    \geq
    \mu(x_t^*)z_t(x_t^*) 
    \geq
    \mu_\circ z^{\mathrm{env}}_t.
  \)
  Consequently, we have $0 < \mu_\circ z^{\mathrm{env}}_t \leq z_t^\mu \leq
  z_t^{\mathrm{env}}$, and hence
  \( 
    t^{-1} \log\mu_\circ + t^{-1} \log z^{\mathrm{env}}_t
    \leq t^{-1} \log z_t^\mu
    \leq t^{-1} \log z^{\mathrm{env}}_t.
  \)
  The corollary then follows from the squeeze theorem and Proposition~\ref{prop:ge}.
\end{proof}

\begin{proof}[Proof of Theorem~\ref{thm:gs}]
  Since $\Pi_t$ is positive, its operator norm induced by the supremum vector norm
  satisfies
  \begin{equation} \label{eq:pn}
    \|\Pi_t\|_\infty
    =
    \|\Pi_t\1\|_\infty
    =
    \max_{x\in\Xsf}z_t(x) = z^{\mathrm{env}}_t,
    \qquad t\geq0.
  \end{equation}
  Moreover, Proposition~4.7 of \citet{batkai2017positive} gives
  \begin{equation} \label{eq:sinf}
    \lambda^*
    =
    s(\Lp)
    =
    \inf_{t>0}
    \frac{1}{t}\log\|\Pi_t\|_\infty.
  \end{equation}
  It follows from \eqref{eq:pn} and \eqref{eq:sinf} that $\lambda^*<0$ if and only if
  there exists $t>0$ such that $\|\Pi_t\|_\infty<1$, which is equivalent to
  $z^{\mathrm{env}}_t<1$.
  This proves the equivalence between (G2) and (G4). The equivalence between (G2) and
  (G3) follows immediately from Proposition~\ref{prop:ge}.
  
  We next show that (G4) implies both (G5) and (G1). Suppose that (G4) holds and
  hence there exists some $T>0$ such that
  \(
    q \coloneqq z_T^{\mathrm{env}} = \|\Pi_T\|_\infty < 1.
  \)
  By continuity of the pricing semigroup, it follows that
  \(
    C \coloneqq \max_{0\leq r\leq T}\|\Pi_r\|_\infty <\infty.
  \)
  For any $t\geq0$, write $t=kT+r$, where $k\in\NN_0$ and $0\leq r<T$. The semigroup
  property gives
  \(
    \Pi_t = \Pi_T^k\Pi_r,
  \)
  and therefore
  \(
    \|\Pi_t\|_\infty \leq Cq^k.
  \)
  Since $q<1$, we have $\|\Pi_t\|_\infty\to0$ as $t\to\infty$ (equivalently, as
  $k\to\infty$). By \eqref{eq:pn}, $z_t^{\mathrm{env}}\to0$, proving (G5).
  Next, partitioning $[0,\infty)$ into intervals of length $T$, together with 
  $\|\Pi_t\|_\infty \leq Cq^k$, gives
  \begin{equation} \label{eq:ipi}
    \int_0^\infty\|\Pi_t\|_\infty\,\diff t
    =
    \sum_{k=0}^\infty
    \int_{kT}^{(k+1)T}\|\Pi_t\|_\infty\,\diff t
    \leq
    CT\sum_{k=0}^\infty q^k
    =
    \frac{CT}{1-q}
    <\infty.
  \end{equation}
  Hence, for every $g\in\RR_+^\Xsf$ and $x\in\Xsf$,
  \(
    0
    \leq
    \int_0^\infty(\Pi_tg)(x)\,\diff t
    \leq
    \|g\|_\infty
    \int_0^\infty\|\Pi_t\|_\infty\,\diff t
    <\infty.
  \)
  Thus, global valuation stability holds, proving (G1).
  
  Conversely, (G5) implies (G4): if $z^{\mathrm{env}}_t\to0$ as $t\to\infty$,
  then $z_T^{\mathrm{env}}<1$ for all sufficiently large $T$.
  It remains to show that (G1) implies (G4). Under global valuation stability,
  $\int_0^\infty z_t(x)\,\diff t<\infty$ for every $x\in\Xsf$.
  Using \eqref{eq:pn} and the finiteness of $\Xsf$,
  \[
    \int_0^\infty\|\Pi_t\|_\infty\,\diff t
    =
    \int_0^\infty
    \max_{x\in\Xsf}z_t(x)\,\diff t
    \leq
    \sum_{x\in\Xsf}
    \int_0^\infty z_t(x)\,\diff t
    <\infty.
  \]
  Consequently, $\|\Pi_t\|_\infty<1$ for some $t>0$; otherwise the integral would be
  infinite. Equation~\eqref{eq:pn} then gives $z_t^{\mathrm{env}}<1$, proving (G4).
  This establishes the equivalence of (G1)--(G5).
  
  Under these conditions, \eqref{eq:ipi} shows that
  \(
    \int_0^\infty\|\Pi_t\|_\infty\diff t < \infty.
  \)
  Hence, the operator-valued function $t\mapsto\Pi_t$ is Bochner integrable on
  $[0,\infty)$. Define $\tilde V:\RR^\Xsf\to\RR^\Xsf$ by
  $\tilde Vg \coloneqq \int_0^\infty\Pi_tg\diff t$, for $g\in\RR^\Xsf$.
  Norm integrability makes $\tilde V$ a finite linear operator, and positivity
  of $\Pi_t$ makes $\tilde V$ positive. For $g\in\RR_+^\Xsf$,
  equation~\eqref{eq:ihf} gives $\tilde Vg=Vg$. Thus $\tilde V$ is a
  positive linear extension of $V$, and the extension is unique because
  $\RR^\Xsf=\RR_+^\Xsf-\RR_+^\Xsf$. Relabeling $\tilde V$ as $V$ gives
  \(
    V=\int_0^\infty\Pi_t\,\diff t.
  \)
  Since $\Pi_t=\exp(t\Lp)$, we have
  \(
    \frac{\diff}{\diff t}\Pi_t
    =
    \Lp\Pi_t
    =
    \Pi_t\Lp,
  \)
  for $t>0$; see, for example \citet[][Theorem 4.2]{batkai2017positive}.
  Therefore, for every $\tau>0$,
  \(
    \Lp\int_0^\tau\Pi_t\,\diff t
    =
    \left(\int_0^\tau\Pi_t\,\diff t\right)\Lp
    =
    \Pi_\tau-I.
  \)
  Letting $\tau\to\infty$ and using $\Pi_\tau\to0$ gives
  \(
    \Lp V
    =
    V\Lp
    =
    -I.
  \)
  Consequently, $\Lp$ is invertible and $V=-\Lp^{-1}$.
\end{proof}

\subsection{Global spectral structure}

\begin{proof}[Proof of Theorem~\ref{thm:Pi-asym}]
  Let $k_*\in\{1,\ldots,m\}$ satisfy $\Ccal_{k_*}=\Ccal^*$ in the ordering compatible
  with the Frobenius form of $\Lp$ in \eqref{eq:frobenius}. We have 
  \begin{equation}\label{eq:l*}
    \lambda^* = s(L_{k_*}) > s(L_k)
    \quad\mbox{for all }k\in\{1,\ldots,m\}\setminus \{k_*\}.
  \end{equation}
  Since $L_{k_*}$ is irreducible Metzler, Perron--Frobenius theory implies that
  $\lambda^*$ is an algebraically simple eigenvalue of $L_{k_*}$ and the
  one-dimensional right and left eigenspaces admit strictly positive eigenvectors.
  Moreover, as the characteristic polynomial of $\Lp$ is the product of those of
  its diagonal blocks and $\lambda^*\notin\sigma(L_k)$ for $k\neq k_*$, by
  \eqref{eq:l*}, $\lambda^*$ is algebraically simple for $\Lp$.

  By \eqref{eq:l*}, the resolvent identity $R(\lambda^*;L_k) = \int_0^\infty
  \e^{-\lambda^*t}\e^{tL_k}\diff t$ holds for $k\neq k_*$. 
  Since $L_k$ is irreducible Metzler, $\e^{t L_k}\gg0$ for all $t>0$ by
  Lemma~\ref{lem:mr}. Consequently,
  \begin{equation}\label{eq:res-sp}
    R(\lambda^*;L_k) \gg0, \qquad k\in\{1,\ldots,m\}\setminus \{k_*\}.
  \end{equation}

  Let $\phi=(\phi_{\scriptscriptstyle\Ccal_1}^\t,\ldots,\phi_{\scriptscriptstyle
  \Ccal_m}^\t)^\t$, partitioned conformably with the Frobenius form of $\Lp$ in
  \eqref{eq:frobenius}, be a right eigenvector of $\Lp$ associated with
  $\lambda^*$ such that $(\lambda^* I - \Lp)\phi = 0$. Then, it follows that
  \begin{equation}\label{eq:reig}
    (\lambda^*I-L_k)\phi_{\scriptscriptstyle \Ccal_k}
    =
    \sum_{\ell=1}^{k-1} L_{k\ell} \, \phi_{\scriptscriptstyle \Ccal_\ell},
    \qquad 
    k=1,\ldots,m,
  \end{equation}
  where the sum is zero when $k=1$. First, we have
  \begin{equation}\label{eq:phi0}
    \phi_{\scriptscriptstyle\Ccal_k}=0,
    \qquad k\in\{1,\ldots,k_*-1\}.
  \end{equation}
  If $k_*=1$, \eqref{eq:phi0} is vacuous. Otherwise, for $k=1$,
  $\phi_{\scriptscriptstyle\Ccal_1}=0$ by the first equation in
  \eqref{eq:reig} and the invertibility of $\lambda^*I-L_1$. For $2\leq k<k_*$, note that
  if $\phi_{\scriptscriptstyle\Ccal_\ell}=0$ for all $\ell=1,\ldots,k-1$, then the
  $k$-th equation in \eqref{eq:reig} and the invertibility of $\lambda^*I-L_k$ imply
  that \( \phi_{\scriptscriptstyle\Ccal_k} = R(\lambda^*;L_k)
  \sum_{\ell=1}^{k-1} L_{k\ell} \, \phi_{\scriptscriptstyle \Ccal_\ell} =0.\)
  Hence \eqref{eq:phi0} holds by induction.

  For $k=k_*$, \eqref{eq:reig} and \eqref{eq:phi0} imply that 
  $(\lambda^*I-L_{k_*})\phi_{\scriptscriptstyle \Ccal_{k_*}}=0$. Since $L_{k_*}$ is
  irreducible, we can choose $\phi_{\scriptscriptstyle \Ccal_{k_*}}\gg0$ to be a
  strictly positive right eigenvector of $L_{k_*}$ associated with $\lambda^*$.
  Since $\Ccal_{k_*}\leadsto\Ccal_{k_*}$, the right-eigenvector support formula in
  \eqref{eq:supports} holds for $\Ccal=\Ccal_{k_*}$.

  Next, we show, again by induction, that
  \begin{equation}\label{eq:phi0-2}
    \phi_{\scriptscriptstyle\Ccal_k}
    \begin{cases}
      \gg0, & \text{if } \Ccal_k\leadsto\Ccal_{k_*},\\
      =0, & \text{otherwise},
    \end{cases}
    \qquad k\in\{k_*+1,\ldots,m\}.
  \end{equation}
  If $k_*=m$, \eqref{eq:phi0-2} is vacuous. Otherwise, for $k=k_*+1$, the
  $(k_*+1)$-th equation in \eqref{eq:reig}, the invertibility of
  $\lambda^*I-L_{k_*+1}$, and \eqref{eq:phi0} jointly imply that
  \(
    \phi_{\scriptscriptstyle \Ccal_{k_*+1}}
    =
    R(\lambda^*;L_{k_*+1}) L_{k_*+1,\, k_*} \,
    \phi_{\scriptscriptstyle \Ccal_{k_*}}.
  \)
  By \eqref{eq:res-sp} and $\phi_{\scriptscriptstyle \Ccal_{k_*}}\gg0$, it follows
  that $\phi_{\scriptscriptstyle \Ccal_{k_*+1}}\gg0$ if $L_{k_*+1,\,k_*}\neq0$,
  which holds exactly when $\Ccal_{k_*+1}\to\Ccal_{k_*}$, or equivalently
  $\Ccal_{k_*+1}\leadsto\Ccal_{k_*}$; otherwise, $L_{k_*+1,\,k_*}=0$ implies
  $\phi_{\scriptscriptstyle \Ccal_{k_*+1}}=0$. Therefore, \eqref{eq:phi0-2} holds for
  $k=k_*+1$.

  Now, suppose that \eqref{eq:phi0-2} holds for $k\in\{k_*+1,\ldots,k'\}$ with
  $k_*+1\leq k'<m$. The $(k'+1)$-th equation in \eqref{eq:reig} and \eqref{eq:res-sp}
  give
  \(
    \phi_{\scriptscriptstyle \Ccal_{k'+1}}
    =
    R(\lambda^*;L_{k'+1}) \sum_{\ell=k_*}^{k'} 
    L_{k'+1, \ell} \, \phi_{\scriptscriptstyle \Ccal_\ell}
  \)
  with $R(\lambda^*;L_{k'+1})\gg0$. Then, $\phi_{\scriptscriptstyle \Ccal_{k'+1}}\gg0$
  if there exists some $\ell\in\{k_*,\ldots,k'\}$ such that
  $L_{k'+1, \ell}\neq0$ and $\phi_{\scriptscriptstyle \Ccal_\ell}\gg0$.
  The latter conditions mean precisely $\Ccal_{k'+1}\to\Ccal_\ell$ and
  $\Ccal_{\ell}\leadsto\Ccal_{k_*}$ for some $\ell\in\{k_*,\ldots,k'\}$, which is
  equivalent to $\Ccal_{k'+1}\leadsto\Ccal_{k_*}$. If $\Ccal_{k_*}$ is not accessible
  from $\Ccal_{k'+1}$, then
  $L_{k'+1,\ell}\phi_{\scriptscriptstyle\Ccal_\ell}=0$ for every
  $\ell\in\{k_*,\ldots,k'\}$, so
  $\phi_{\scriptscriptstyle \Ccal_{k'+1}}=0$. Consequently, \eqref{eq:phi0-2}
  holds for $k=k'+1$, and the full result holds by induction.

  Since the ordering of distinct classes rules out $\Ccal_k\leadsto\Ccal_{k_*}$ for
  $k < k_*$, the right-eigenvector support formula in \eqref{eq:supports} holds due
  to \eqref{eq:phi0} and \eqref{eq:phi0-2}.

  Similarly, let $\psi=(\psi_{\scriptscriptstyle\Ccal_1}^\t,\ldots,
  \psi_{\scriptscriptstyle\Ccal_m}^\t)^\t$ be a left eigenvector of $\Lp$ associated
  with $\lambda^*$ such that $(\lambda^*I-\Lp^\t)\psi=0$. Its block equations are
  \begin{equation}\label{eq:leig}
    (\lambda^*I-L_k^\t)\psi_{\scriptscriptstyle\Ccal_k}
    =
    \sum_{\ell=k+1}^m
    L_{\ell k}^\t\psi_{\scriptscriptstyle\Ccal_\ell},
    \qquad k=1,\ldots,m,
  \end{equation}
  where the sum is zero when $k=m$. Starting from $k=m$ and proceeding backward,
  \eqref{eq:leig} and the invertibility in \eqref{eq:res-sp} give
  $\psi_{\scriptscriptstyle\Ccal_k}=0$ for every $k>k_*$. Choose
  $\psi_{\scriptscriptstyle\Ccal_{k_*}}\gg0$ to be a left eigenvector of $L_{k_*}$
  associated with $\lambda^*$. For $k<k_*$, \eqref{eq:leig} becomes
  \[
    \psi_{\scriptscriptstyle\Ccal_k}
    =
    R(\lambda^*;L_k^\t)
    \sum_{\ell=k+1}^{k_*}
    L_{\ell k}^\t\psi_{\scriptscriptstyle\Ccal_\ell}.
  \]
  Since $R(\lambda^*;L_k^\t)\gg0$, backward induction shows that
  $\psi_{\scriptscriptstyle\Ccal_k}\gg0$ exactly when there exists
  $\ell\in\{k+1,\ldots,k_*\}$ such that
  $\Ccal_{k_*}\leadsto\Ccal_\ell$ and $\Ccal_\ell\to\Ccal_k$, or equivalently,
  when $\Ccal_{k_*}\leadsto\Ccal_k$; otherwise,
  $\psi_{\scriptscriptstyle\Ccal_k}=0$. The class ordering also rules out
  $\Ccal_{k_*}\leadsto\Ccal_k$ for $k>k_*$. This proves the left-eigenvector support
  formula and hence \eqref{eq:supports}.

  The vectors constructed above are nonnegative and strictly positive on
  $\Ccal^*$, so $\1^\t\phi>0$ and
  $\psi^\t\phi\geq
  \psi_{\scriptscriptstyle\Ccal^*}^\t
  \phi_{\scriptscriptstyle\Ccal^*}>0$. We may therefore normalize them by
  $\1^\t\phi=1$ and $\psi^\t\phi=1$. Since $\lambda^*$ is algebraically simple,
  its right and left eigenspaces are one-dimensional, and hence these normalizations
  make $\phi$ and $\psi$ unique. Furthermore, the Riesz spectral projection is rank
  one, and is given by $P_{\lambda^*}=\phi\psi^\t$. Moreover,
  \[
    \e^{-\lambda^*t}\Pi_t
    =
    \e^{-\lambda^*t}\e^{t\Lp} [P_{\lambda^*} + (I-P_{\lambda^*})]
    =
    P_{\lambda^*}
    +
    \e^{-\lambda^*t}\e^{t\Lp}(I-P_{\lambda^*}),
    \qquad t\geq0.
  \]
  By \eqref{eq:gsv}, the restriction of $\Lp$ to the range of
  $I-P_{\lambda^*}$ has spectral bound strictly below $\lambda^*$. The second term
  therefore converges to zero as $t\to\infty$, proving
  \eqref{eq:r1lim}.
\end{proof}

\begin{proof}[Proof of Theorem~\ref{thm:Pi-asym-tied}]
  Lemma~\ref{lem:res-exp} writes every block of
  $R(z;\Lp)$ as a finite sum of class-path products. At $z=\lambda^*$, the
  pole order of such a product is the number of classes from $\cC^*$ on that path:
  each dominant diagonal-block resolvent has a simple Perron pole, whereas every
  lower-rate block resolvent is analytic and strictly positive at $\lambda^*$ by
  Lemma~\ref{lem:res-exp}, applied to the corresponding diagonal block.
  The leading coefficients are nonnegative and nonzero, so positivity rules out
  cancellation. The largest resolvent pole order is therefore $\nu^*$.

  For a finite matrix, the pole order at an eigenvalue equals the size of its largest
  Jordan block. Set $N^*\coloneqq(\Lp-\lambda^*I)P^*$. This matrix agrees with
  $\Lp-\lambda^*I$ on the range of $P^*$ and is zero on $\ker P^*$, so it is
  nilpotent of index $\nu^*$. Thus
  \[
    \e^{t(\Lp-\lambda^*I)}P^*
    =
    \sum_{j=0}^{\nu^*-1}\frac{t^j}{j!}(N^*)^jP^*.
  \]
  By \eqref{eq:gsv}, the restriction of $\Lp$ to $\ker P^*$ has spectral bound
  strictly below $\lambda^*$. Its normalized semigroup is therefore bounded by
  $C\e^{-\eta t}$ for suitable $C,\eta>0$.
  Therefore, there exist constants $C,\eta>0$ such that
  \begin{equation} \label{eq:tieds}
    \e^{-\lambda^*t}\Pi_t
    =
    \sum_{j=0}^{\nu^*-1}\frac{t^j}{j!}(N^*)^jP^*+R_t,
    \qquad
    \|R_t\|\leq C\e^{-\eta t}.
  \end{equation}
  Multiplication by $t^{1-\nu^*}$ gives \eqref{eq:price-glim}. The limit is
  nonnegative because $\Pi_t\geq0$, and it is nonzero because
  $(N^*)^{\nu^*-1}P^*\neq0$. Finally, $N^*B^*=B^*N^*=0$, so $\Lp
  B^*=B^*\Lp=\lambda^*B^*$. Semisimplicity is equivalent to $\nu^*=1$, yielding the
  final claim.
\end{proof}

\subsection{Doob transforms and the long-run pricing law}
\label{app:doob}

This subsection gives the two-transform construction underlying the Markov
interpretation in Section~\ref{sec:unique-dominance}. Maintain
Assumptions~\ref{ass:markov}--\ref{ass:continuity} and suppose that
$\cC^*=\{\Ccal^*\}$, with eigenvectors normalized as in
Theorem~\ref{thm:Pi-asym}.

Restricting the dominant eigenvectors to their supports permits Doob $h$-transforms
of the pricing semigroup and its adjoint \citep[see, e.g.,][]{doob1957conditional}.
Let $\phi_{\scriptscriptstyle+}$ and $\psi_{\scriptscriptstyle+}$ be the restrictions
of $\phi$ and $\psi$ to their respective supports, and let
$\Lp^{\phi_+}$ and $\Lp^{\psi_+}$ be the principal submatrices of $\Lp$ indexed by
$\supp(\phi)$ and $\supp(\psi)$, respectively. Define
\begin{align}
  \Lp^{\rightarrow}
  &\coloneqq
  \diag(\phi_{\scriptscriptstyle+})^{-1}
  \bigl(\Lp^{\phi_+}-\lambda^*I\bigr)
  \diag(\phi_{\scriptscriptstyle+}),
  \label{eq:ft}
  \\
  \Lp^{\leftarrow}
  &\coloneqq
  \diag(\psi_{\scriptscriptstyle+})^{-1}
  \bigl[(\Lp^{\psi_+})^\t-\lambda^*I\bigr]
  \diag(\psi_{\scriptscriptstyle+}).
  \label{eq:dt}
\end{align}

\begin{figure}[t!]
\centering
\begin{subfigure}[t]{0.35\textwidth}
\centering
\begin{minipage}[c][22mm][c]{\linewidth}
\centering
\begin{tikzpicture}[node distance=5mm,scale=0.75,transform shape]
  \node[classnode,fill=black!10] (D) {Disaster\\$\mathcal D$};
  \node[classnode,right=of D] (R) {Recovery\\$\mathcal R$};
  \node[classnode,very thick,double,right=of R] (N) {Normal\\$\mathcal N$};
  \draw[flow] (D) -- (R);
  \draw[flow] (R) -- (N);
\end{tikzpicture}
\end{minipage}
\caption{Original reduced graph.}
\end{subfigure}\hfill
\begin{subfigure}[t]{0.26\textwidth}
\centering
\begin{minipage}[c][22mm][c]{\linewidth}
\centering
\begin{tikzpicture}[scale=0.75,transform shape]
  \node[classnode,fill=black!10,very thick,double] (D)
    {Disaster\\$\mathcal D$};
\end{tikzpicture}
\end{minipage}
\caption{Right Doob chain.}
\end{subfigure}\hfill
\begin{subfigure}[t]{0.35\textwidth}
\centering
\begin{minipage}[c][22mm][c]{\linewidth}
\centering
\begin{tikzpicture}[node distance=5mm,scale=0.75,transform shape]
  \node[classnode,fill=black!10,very thick,double] (D)
    {Disaster\\$\mathcal D$};
  \node[classnode,right=of D] (R) {Recovery\\$\mathcal R$};
  \node[classnode,right=of R] (N) {Normal\\$\mathcal N$};
  \draw[flow] (N) -- (R);
  \draw[flow] (R) -- (D);
\end{tikzpicture}
\end{minipage}
\caption{Left Doob chain.}
\end{subfigure}
\caption{Original and transformed reduced graphs in the
disaster-dominant economy. Panels~(a)--(c) show
$\Gcal_{\mathrm r}(\Lp)$, $\Gcal_{\mathrm r}(\Lp^\rightarrow)$, and
$\Gcal_{\mathrm r}(\Lp^\leftarrow)$, respectively. Light shading identifies the
globally pricing-dominant class $\mathcal D$. Double borders identify the unique
closed class: $\mathcal N$ in panel~(a) and $\mathcal D$ in panels~(b) and~(c).
The graph in panel~(b) has one vertex and no edges.}
\label{fig:flow-reversal}
\end{figure}
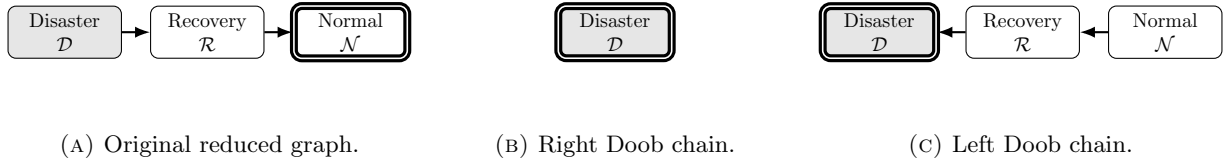

Both transforms remove the common exponential rate $\lambda^*$; the transpose in
\eqref{eq:dt} reverses the retained edges. The following proposition describes the
resulting Markov dynamics.

\begin{proposition}[Doob chains] \label{prop:valuation-laws}
  Under Assumptions~\ref{ass:markov}--\ref{ass:continuity}, suppose that
  $\cC^*=\{\Ccal^*\}$ and let $\phi$ and $\psi$ be the right and left eigenvectors of
  $\Lp$ associated with $\lambda^*$, normalized as in
  Theorem~\ref{thm:Pi-asym}. Then $\Lp^{\rightarrow}$ and $\Lp^{\leftarrow}$,
  defined in \eqref{eq:ft} and \eqref{eq:dt}, are Markov generators on $\supp(\phi)$
  and $\supp(\psi)$, respectively;
  we call the associated Markov chains the \emph{right Doob chain} and \emph{left Doob chain}.
  Moreover,
  \begin{enumerate}
    \item the reduced graph $\Gcal_{\mathrm r}(\Lp^\rightarrow)$ is the subgraph of
      $\Gcal_{\mathrm r}(\Lp)$ induced by the classes that can reach $\Ccal^*$,
      whereas $\Gcal_{\mathrm r}(\Lp^\leftarrow)$ is the edge reversal of the
      subgraph of $\Gcal_{\mathrm r}(\Lp)$ induced by the classes accessible
      from $\Ccal^*$;

    \item $\Ccal^*$ is the unique closed, and hence recurrent, class of both Doob
      chains;

    \item the unique stationary distributions of $\Lp^\rightarrow$ and
      $\Lp^\leftarrow$, extended by zero outside their respective state spaces, both
      equal \( \pi^{\mathrm{LR}}=\phi\odot\psi, \) which is supported on $\Ccal^*$.
  \end{enumerate}
\end{proposition}

\begin{proof}[Proof of Proposition~\ref{prop:valuation-laws}]
  The identities $\Lp^{\phi_+}\phi_+=\lambda^*\phi_+$ and
  $(\Lp^{\psi_+})^\t\psi_+=\lambda^*\psi_+$, together with \eqref{eq:ft} and
  \eqref{eq:dt}, imply \(\Lp^\rightarrow\1=0\) and \(\Lp^\leftarrow\1=0\). Both
  restricted matrices $\Lp^{\phi_+}$ and $\Lp^{\psi_+}$ are Metzler; positive
  diagonal similarity preserves their off-diagonal zero patterns, while transposition
  reverses those patterns. Therefore, $\Lp^\rightarrow$ and $\Lp^\leftarrow$ are
  Markov generators on $\supp(\phi)$ and $\supp(\psi)$ with the stated reduced
  graphs. Both reduced graphs are acyclic, being respectively an induced subgraph
  of \(\Gcal_{\mathrm r}(\Lp)\) and the edge reversal of one. In each graph,
  \(\Ccal^*\) is closed and accessible from every class, and hence is the unique
  closed class.

  Set $\Xsf_\phi\coloneqq\supp(\phi)$ and
  $\Xsf_\psi\coloneqq\supp(\psi)$. No pricing edge enters $\Xsf_\phi$ from
  outside: otherwise, the class containing its initial vertex would have access to
  $\Ccal^*$ and would belong to $\Xsf_\phi$. Similarly, no pricing edge leaves
  $\Xsf_\psi$, since the class containing its terminal vertex would then be
  accessible from $\Ccal^*$ and would belong to $\Xsf_\psi$. Consequently, taking
  the $\Xsf_\phi$ coordinates of $\psi^\t\Lp=\lambda^*\psi^\t$ and the
  $\Xsf_\psi$ coordinates of $\Lp\phi=\lambda^*\phi$ gives
  \(
    \psi_{\Xsf_\phi}^\t\Lp^{\phi_+} = \lambda^*\psi_{\Xsf_\phi}^\t,
  \)
  and
  \(
    \Lp^{\psi_+}\phi_{\Xsf_\psi} = \lambda^*\phi_{\Xsf_\psi}.
  \)
  Let $D_\phi\coloneqq\diag(\phi_+)$ and
  $D_\psi\coloneqq\diag(\psi_+)$. Using \eqref{eq:ft} and \eqref{eq:dt}, we obtain
  \begin{align*}
    (\phi_+\odot\psi_{\Xsf_\phi})^\t\Lp^\rightarrow
    &=
    \psi_{\Xsf_\phi}^\t
    (\Lp^{\phi_+}-\lambda^*I)D_\phi
    =0,
    \\
    (\phi_{\Xsf_\psi}\odot\psi_+)^\t\Lp^\leftarrow
    &=
    \phi_{\Xsf_\psi}^\t
    \bigl[(\Lp^{\psi_+})^\t-\lambda^*I\bigr]D_\psi
    =0.
  \end{align*}
  Thus these vectors are stationary for $\Lp^\rightarrow$ and
  $\Lp^\leftarrow$, respectively. Their zero extensions to $\Xsf$ both equal
  $\phi\odot\psi$, because $\phi$ vanishes outside $\Xsf_\phi$ and $\psi$
  vanishes outside $\Xsf_\psi$. Moreover,
  $\1^\t(\phi\odot\psi)=\psi^\t\phi=1$, and
  $\supp(\phi)\cap\supp(\psi)=\Ccal^*$: any other common class would communicate
  with $\Ccal^*$. Because $\Ccal^*$ is the unique closed class of each transformed
  graph, each generator has a unique stationary distribution. It is therefore
  $\pi^{\mathrm{LR}}=\phi\odot\psi$.
\end{proof}

For $x\in\supp(\phi)$, the right Doob chain admits a change-of-measure
interpretation. The process
\[
  M_t^\phi
  \coloneqq
  \e^{-\lambda^*t}S_t^{\Gs}\frac{\phi(X_t)}{\phi(x)},
  \qquad t\geq0,
\]
is a nonnegative unit-mean martingale and therefore defines a family of
probability measures $(\PP_x^\phi)_{x\in\supp(\phi)}$ by
$\diff\PP_x^\phi/\diff\PP_x|_{\fF_t} = M_t^\phi$.
Under $\PP_x^\phi$, the state process is Markov on $\supp(\phi)$ with generator
$\Lp^\rightarrow$. The left Doob chain comes from the adjoint construction and
does not generally correspond to another forward change of physical measure.
When the economy is irreducible, the right transform is the whole-state
Hansen--Scheinkman transform and the left is its stationary time reversal;
Appendix~\ref{app:hs} reviews this benchmark.

By Proposition~\ref{prop:valuation-laws}(ii)--(iii) and finite-state Markov-chain
convergence, the distributions of both Doob chains converge from every initial
state in their respective state spaces to $\pi^{\mathrm{LR}}$.
Thus a physically transient class can be the
unique recurrent class under both transformed dynamics.

Figure~\ref{fig:flow-reversal} illustrates these results for the three-class graph
$\mathcal D\to\mathcal R\to\mathcal N$ with $\Ccal^*=\mathcal D$.
Then $\supp(\phi)=\mathcal D$ and $\supp(\psi)=\Xsf$: the global rate is absent
from prices starting in recovery or normal states, but prices starting in disaster
can inherit it even when the payoff is delivered only after recovery. The right
Doob chain is restricted to $\mathcal D$, while the left has the reversed reduced
graph $\mathcal N\to\mathcal R\to\mathcal D$. Both have limiting distributions concentrated
on $\mathcal D$, although the physical dynamics eventually enter the recurrent
normal class $\mathcal N$.
If the uniquely dominant class moves downstream along this chain, the right
eigenvector's support expands and the left eigenvector's support contracts.

\subsection{Pricing-corridor restriction and asymptotics}
\label{app:corridor-proofs}

\begin{proof}[Proof of Proposition~\ref{prop:corridor}]
  Suppose first that $\kK(\Ccal,g)=\varnothing$. Then no state in $\Ccal$ has
  access to a state in $\supp(g)$. Lemma~\ref{lem:mr}, applied to $\Lp$, gives
  $q_t(x,g)= (\e^{t\Lp} g)(x) = 0$ for every $x\in\Ccal$ and $t>0$.
  The same conclusion holds at $t=0$, because $g$ must vanish on $\Ccal$; otherwise
  $\Ccal\in\sS(g)$ and the fact that $\Ccal\leadsto\Ccal$ would place $\Ccal$ in the
  corridor.

  Now suppose that $\kK=\kK(\Ccal,g) \neq \varnothing$. Choose
  $\Ccal'\in\kK$. By definition, there is a class
  $\Ccal''\in\sS(g)$ such that
  $\Ccal\leadsto\Ccal'\leadsto\Ccal''$. Hence
  $\Ccal\leadsto\Ccal''$, which implies both
  $\Ccal\in\kK$ and
  $\Ccal''\in\sS(g)\cap\kK$. In particular,
  $g_{\kK}\neq0$.
  The principal submatrix $L_{\kK}$ retains exactly the pricing edges whose
  endpoints lie in $\Xsf_{\kK}$. Because $\kK$ is a union of whole
  classes, these classes remain the classes of
  $L_{\kK}$. Moreover, every class on a path between two corridor classes also
  lies in the corridor: it is accessible from $\Ccal$ and has access to
  $\sS(g)$. Thus the reduced graph of $L_{\kK}$ is the subgraph of the
  reduced pricing graph induced by $\kK$, with the same accessibility relation
  among its vertices.

  To prove the semigroup restriction, let $z$ be any sufficiently large real number
  and apply
  Lemma~\ref{lem:res-exp} to $\Lp$ and $L_{\kK}$. A class in
  $\sS(g)$ that is accessible from $\Ccal$ belongs to $\kK$, and every class on
  every path from $\Ccal$ to that payoff-support class also belongs to $\kK$.
  Hence the path contributions to the two scalar resolvents coincide, giving
  \[
    e_x^\t R(z;\Lp)g
    =
    e_x^\t R(z;L_{\kK})g_{\kK}.
  \]
  For $z>\max\{s(\Lp),s(L_{\kK})\}$, the two sides of the above equation are the
  Laplace transforms of $(\e^{t\Lp}g)(x)$ and $(\e^{tL_{\kK}}g_{\kK})(x)$,
  respectively. Uniqueness of the Laplace transform, together with continuity in $t$,
  proves \eqref{eq:corridor-restriction} for every $t\geq0$.
  Lastly, $\kK(\Ccal,g)\neq\varnothing$ means that there exists \(y\in\supp(g)\) such
  that $x\leadsto y$ for some, and hence every, $x\in\mathcal C$. Lemma~\ref{lem:mr}
  then implies
  \(
    q_t(x,g) = (\e^{t\Lp}g)(x) > 0
  \) 
  for every \(t>0\) and every $x\in\Ccal$.
\end{proof}

\begin{lemma}[Positive leading corridor coefficient] \label{lem:positive-corridor-coefficient}
  Let $g\in\RR_+^\Xsf$ be nonzero, fix $x\in\Ccal$, and suppose that
  $\kK=\kK(\Ccal,g)$ is nonempty. With $\lambda$, $\nu$, $P_{\kK}$,
  and $B_{\kK}$ as in \eqref{eq:local-rate}--\eqref{eq:BK}, set
  $N_{\kK}\coloneqq(L_{\kK}-\lambda I)P_{\kK}$. Then
  \begin{equation} \label{eq:positive-corridor-coefficient}
    e_x^\t N_{\kK}^{\nu-1}P_{\kK}g_{\kK}>0,
    \qquad
    \bigl(B_{\kK}g_{\kK}\bigr)(x)>0.
  \end{equation}
\end{lemma}

\begin{proof}
  Apply Lemma~\ref{lem:res-exp} to the corridor generator
  $L_{\kK}$. By Proposition~\ref{prop:corridor}, its reduced graph is
  induced by $\kK$, and
  \begin{equation} \label{eq:corridor-scalar-resolvent}
    e_x^\t R(z;L_{\kK})g_{\kK}
    =
    \sum_{\Ccal_\ell\in\sS(g)\cap\kK}
    \sum_{\gamma:\Ccal\leadsto\Ccal_\ell}
    e_x^\t\mathcal R_\gamma(z;L_{\kK})g_{\Ccal_\ell}.
  \end{equation}
  Call the paths in this sum \emph{relevant}. Every relevant path lies in
  $\kK$, so it contains at most $\nu$ corridor-critical classes. Conversely, let
  \(
    \Ccal_{\ell_1}\leadsto\cdots\leadsto\Ccal_{\ell_\nu}
  \)
  be the ordered corridor-critical classes in a chain attaining length $\nu$. Since
  its endpoints lie in the corridor,
  $\Ccal\leadsto\Ccal_{\ell_1}$ and
  $\Ccal_{\ell_\nu}\leadsto\Ccal_h$ for some
  $\Ccal_h\in\sS(g)$. Concatenating paths that realize these relations
  produces a directed walk through all $\nu$ critical classes. The reduced graph is
  acyclic, so this walk is a relevant directed path. Thus at least one relevant
  path contains exactly $\nu$ critical classes.

  Fix such a path
  $\gamma=(\Ccal_{k_0}\to\cdots\to\Ccal_{k_p})$. For each position $h$, set
  \[
    G_{\gamma,h}
    \coloneqq
    \begin{cases}
      \phi_{k_h}\psi_{k_h}^\t,
        & \lambda_{k_h}=\lambda,\\[3pt]
      R(\lambda;L_{k_h}),
        & \lambda_{k_h}<\lambda.
    \end{cases}
  \]
  Substituting the block Laurent expansion~\eqref{eq:block-laurent} into
  \eqref{eq:path-res} shows that the coefficient of
  $(z-\lambda)^{-\nu}$ contributed by $\gamma$ to
  \eqref{eq:corridor-scalar-resolvent} is
  \begin{equation} \label{eq:path-leading-coefficient}
    e_x^\t
    G_{\gamma,0}L_{k_0k_1}G_{\gamma,1}\cdots
    L_{k_{p-1}k_p}G_{\gamma,p}g_{\Ccal_{k_p}}.
  \end{equation}
  This scalar is strictly positive. Indeed, every noncritical resolvent is strictly
  positive by Lemma~\ref{lem:res-exp}, applied to the corresponding
  diagonal block; every block Perron vector is strictly positive; every coupling
  block along the path is nonzero and nonnegative; and $g_{\Ccal_{k_p}}$ is nonzero
  and nonnegative. The same argument applies to every relevant path containing $\nu$
  critical classes. All remaining paths contribute a pole of lower order or a
  function analytic at $\lambda$.
  Hence the coefficient of $(z-\lambda)^{-\nu}$ in
  \eqref{eq:corridor-scalar-resolvent} is strictly positive.

  Theorem~\ref{thm:Pi-asym} in the unique-critical case and
  Theorem~\ref{thm:Pi-asym-tied} in the tied case identify $\nu$ as the nilpotency
  index of $N_{\kK}$. On its dominant spectral subspace,
  \[
    R(z;L_{\kK})P_{\kK}
    =
    \sum_{j=0}^{\nu-1}
    \frac{N_{\kK}^jP_{\kK}}{(z-\lambda)^{j+1}},
  \]
  while $R(z;L_{\kK})(I-P_{\kK})$ is analytic at $\lambda$. Therefore, the
  strictly positive coefficient just identified equals
  $e_x^\t N_{\kK}^{\nu-1}P_{\kK}g_{\kK}$. Dividing by $(\nu-1)!$ proves
  \eqref{eq:positive-corridor-coefficient}.
\end{proof}

\begin{proof}[Proof of Theorem~\ref{thm:pyc}]
  Write $\kK=\kK(\Ccal,g)$, $\lambda=\lambda(\Ccal,g)$, and
  $\nu=\nu(\Ccal,g)$. The corridor-critical classes are exactly the globally dominant
  classes of $L_{\kK}$, and their maximum chain length is $\nu$. Hence
  Theorem~\ref{thm:Pi-asym-tied}, applied to $L_{\kK}$, gives
  \[
    \lim_{t\to\infty}
    t^{1-\nu}\e^{-\lambda t}\e^{tL_{\kK}}
    =
    \frac{N_{\kK}^{\nu-1}P_{\kK}}{(\nu-1)!}
    =B_{\kK}.
  \]
  Combining this limit with the exact restriction
  \eqref{eq:corridor-restriction} yields
  \(
    t^{1-\nu}\e^{-\lambda t}q_t(x,g)
    \to
    (B_{\kK}g_{\kK})(x)
    \eqqcolon c_x(g).
  \)
  Lemma~\ref{lem:positive-corridor-coefficient} shows that $c_x(g)>0$. Thus
  \(
    q_t(x,g)=c_x(g)t^{\nu-1}\e^{\lambda t}(1+o(1)).
  \)
  Proposition~\ref{prop:corridor}(ii) ensures that $q_t(x,g)>0$ for $t>0$, so
  taking logarithms is valid. Dividing by $-t$ gives
  \eqref{eq:yield-clim}.
\end{proof}

The leading coefficient also admits an eigenvector representation. Write
$\kK=\kK(\Ccal,g)$ and $\lambda=\lambda(\Ccal,g)$.
Theorem~\ref{thm:Pi-asym-tied} gives $B_{\kK}\neq0$ and
$L_{\kK}B_{\kK}=B_{\kK}L_{\kK}=\lambda B_{\kK}$. Therefore, if
$d_{\kK}=\operatorname{rank}(B_{\kK})$, then
$\operatorname{ran}(B_{\kK})\subseteq\ker(L_{\kK}-\lambda I)$ and
$\operatorname{ran}(B_{\kK}^{\t})\subseteq\ker(L_{\kK}^{\t}-\lambda I)$.
Hence a rank factorization may be chosen in terms of right and left eigenvectors
$\phi_{\kK,a}$ and $\psi_{\kK,a}$ of $L_{\kK}$ at $\lambda$:
\begin{equation} \label{eq:corridor-eigenfactorization}
  B_{\kK}
  =
  \sum_{a=1}^{d_{\kK}}
  \phi_{\kK,a}\psi_{\kK,a}^{\t}.
\end{equation}
Although the individual factors need not be unique, their sum $B_{\kK}$ is the
canonical leading spectral coefficient. In particular,
\[
  c_x(g)
  =\sum_{a=1}^{d_{\kK}}
    \phi_{\kK,a}(x)\,\psi_{\kK,a}^{\t}g_{\kK}.
\]

\begin{proof}[Proof of Corollary~\ref{cor:corridor-stability}]
  If the corridor is empty, Proposition~\ref{prop:corridor} makes the
  strip price, and hence its integral, identically zero. Otherwise,
  Theorem~\ref{thm:pyc} gives
  \(
    q_t(x,g) \sim c_x(g)t^{\nu(\Ccal,g)-1}\e^{\lambda(\Ccal,g)t}
  \)
  for $c_x(g)>0$. The integral over any bounded maturity interval is finite by
  continuity of the pricing semigroup, while the displayed tail is integrable
  exactly when $\lambda(\Ccal,g)<0$. This proves the equivalence.
\end{proof}

\subsection{One-way disaster-recovery economy}
\label{app:illustration-proofs}

\begin{proof}[Proof of Proposition~\ref{prop:owr}]
  We first establish \eqref{eq:etLp}. Conformably with the block representation
  of $\Lp$ in \eqref{eq:Lp2}, write
  \[
    \Pi_t = \e^{t\Lp}
    =
    \begin{pmatrix} A_t & B_t\\ H_t & C_t \end{pmatrix},
    \qquad t\geq0.
  \]
  The pricing semigroup satisfies the evolution equation
  \[
    \frac{\diff}{\diff t}\Pi_t=\Lp\Pi_t, \qquad\Pi_0=I.
  \]
  Together with the block representation of $\Lp$ in \eqref{eq:Lp2}, this gives
  \[
    \frac{\diff}{\diff t} A_t = \LpN A_t,
    \quad
    \frac{\diff}{\diff t} B_t = \LpN B_t,
    \quad
    \frac{\diff}{\diff t} H_t = \LpD H_t + \LpDN A_t,
    \quad
    \frac{\diff}{\diff t} C_t = \LpD C_t + \LpDN B_t,
  \]
  with initial conditions $A_0 = I$, $B_0=0$, $H_0=0$ and $C_0=I$.
  It follows immediately that $A_t = \e^{t\LpN}$, $B_t=0$ and $C_t=\e^{t\LpD}$.
  Applying the variation-of-constants formula (Duhamel's formula) to the equation for
  \(H_t\) yields
  \(
    H_t
    =
    \e^{t\LpD} H_0 + \int_0^t \e^{(t-v)\LpD}\LpDN A_v\diff v.
  \)
  Substituting $H_0=0$ and $A_v=\e^{v\LpN}$, then changing variables to
  $s=t-v$, gives
  \[
    H_t=\int_0^t\e^{s\LpD}\LpDN\e^{(t-s)\LpN}\diff s,
  \]
  with the order of the matrix factors unchanged. This proves
  \eqref{eq:etLp}.

  For part~(i), fix $x\in\mathcal N$. If $\gN=0$, the block representation
  gives $q_t(x,g)=0$ for every $t\geq0$. If $\gN\neq0$, the pricing corridor is
  $\{\mathcal N\}$ and has rate $\lambdaN$. Theorem~\ref{thm:pyc} therefore gives
  $\lim_{t\to\infty}y_t(x,g)=-\lambdaN$.

  For part~(ii), fix $x\in\mathcal D$. If $\gN=0$, then $\gD\neq0$ because
  $g$ is nonzero, and the pricing corridor is $\{\mathcal D\}$. If $\gN\neq0$,
  the nonzero recovery block and irreducibility of the two diagonal blocks give
  the corridor $\{\mathcal D,\mathcal N\}$. The corresponding corridor rates are
  $\lambdaD$ and $\max\{\lambdaD,\lambdaN\}$, respectively.
  The yield expansion in Theorem~\ref{thm:pyc} gives the negatives of these rates
  as the limiting yields, including when $\lambdaD=\lambdaN$, since the
  $\log t/t$ correction vanishes as $t\to\infty$.
\end{proof}

\begin{proof}[Proof of Corollary~\ref{cor:owr}]
  Under \eqref{eq:dd}, $\mathcal D$ is the unique globally
  dominant class. Only $\mathcal D$ can reach $\mathcal D$, whereas both
  $\mathcal D$ and $\mathcal N$ are accessible from $\mathcal D$.
  Theorem~\ref{thm:Pi-asym} therefore implies
  that $\lambdaD$ is algebraically simple for $\Lp$, that the dominant right
  eigenvector is supported on $\mathcal D$, that the dominant left eigenvector has
  full support, and that the normalized semigroup converges to the associated
  rank-one Riesz projection.

  With the state ordering in \eqref{eq:Lp2}, the right eigenvector is
  $\PhiD=(0,\phiD^{\t})^{\t}$. Writing the left eigenvector as
  $(\xi^{\t},\psiD^{\t})^{\t}$, the normal-class block of
  $\Lp^{\t}(\xi^{\t},\psiD^{\t})^{\t}
  =\lambdaD(\xi^{\t},\psiD^{\t})^{\t}$ gives
  \(
    (\lambdaD I-\LpN^{\t})\xi=\LpDN^{\t}\psiD.
  \)
  Since $\lambdaD>s(\LpN)$, the matrix on the left is invertible and
  \(
    \xi=R(\lambdaD;\LpN^{\t})\LpDN^{\t}\psiD.
  \)
  Thus the full-state eigenvectors are
  \begin{equation}\label{eq:illustration-eigenvectors}
    \PhiD\coloneqq\begin{pmatrix}0\\ \phiD\end{pmatrix},
    \qquad
    \PsiD\coloneqq
    \begin{pmatrix}
      R(\lambdaD;\LpN^\t)\LpDN^\t\psiD\\[0.2em]\psiD
    \end{pmatrix}\gg0.
  \end{equation}
  The block normalizations $\1^\t\phiD=1$ and $\psiD^\t\phiD=1$ give
  $\PsiD^\t\PhiD=1$. Hence the Riesz projection is
  $P_{\scriptscriptstyle\mathcal D}\coloneqq\PhiD\PsiD^\t$, and
  \eqref{eq:r1lim} gives
  \begin{equation}\label{eq:illustration-rank-one}
    \lim_{t\to\infty}\e^{-\lambdaD t}\Pi_t
    =P_{\scriptscriptstyle\mathcal D}.
  \end{equation}
  Finally,
  \[
    \PsiD^{\t}g
    =
    \psiD^{\t}\gD
    +
    \psiD^{\t}\LpDN R(\lambdaD;\LpN)\gN,
  \]
  where we used $R(\lambdaD;\LpN^{\t})^{\t}=R(\lambdaD;\LpN)$.
  Since $\lambdaD>s(\LpN)$, the Laplace representation
  \[
    R(\lambdaD;\LpN)
    =\int_0^\infty\e^{-\lambdaD v}\e^{v\LpN}\diff v
  \]
  is convergent. Substituting this identity and applying
  $P_{\scriptscriptstyle\mathcal D}$ to $g$ proves
  \eqref{eq:dlim}. Its coefficient is strictly positive
  for $x\in\mathcal D$, since $\phiD(x)>0$, $\PsiD\gg0$, and $g\geq0$ is
  nonzero.

  For $x\in\mathcal N$, the block representation \eqref{eq:etLp} gives
  $q_t(x,g)=(\e^{t\LpN}\gN)(x)$. Since $\LpN$ is irreducible Metzler,
  \(
    \lim_{t\to\infty}\e^{-\lambdaN t}\e^{t\LpN}
    = \phiN\psiN^\t.
  \)
  The limit is positive when $\gN\neq0$, because $\phiN,\psiN\gg0$ and $\gN\geq0$.
  This proves part~(i) of the corollary. 
\end{proof}

\subsection{Rare-edge perturbations and crossover asymptotics}
\label{app:rare-proofs}

For each $\epsilon>0$ sufficiently small, let $\lambda^*(\epsilon)\coloneqq
s(\Lp(\epsilon))$ and let $\phi(\epsilon), \psi(\epsilon)\gg0$ be corresponding right
and left eigenvectors of $\Lp(\epsilon)$ associated with $\lambda^*(\epsilon)$.
For each limiting class $\Ccal\in\cC$, write $\phi_{\s\Ccal}(\epsilon)$ and
$\psi_{\s\Ccal}(\epsilon)$ for the restrictions of these vectors to $\Ccal$.
Normalize the vectors by
\(
  \1^\t\phi_{\scriptscriptstyle\Ccal^*}(\epsilon)=1,
\)
and
\(
  \phi(\epsilon)^\t\psi(\epsilon) = 1.
\)

In what follows, Lemma~\ref{lem:rp*} establishes the convergence of
$\lambda^*(\epsilon)$, $\phi_{\s\Ccal^*}(\epsilon)$ and $\psi_{\s\Ccal^*}(\epsilon)$ as
$\epsilon\downarrow0$, and Lemma~\ref{lem:rp} establishes the rarity orders of the 
$\phi_{\s\Ccal}(\epsilon)$ and $\psi_{\s\Ccal}(\epsilon)$ for $\Ccal\neq\Ccal^*$.
Then, we use these two lemmas to prove Theorem~\ref{thm:rare-distance} in
Section~\ref{sec:rare-entry}.

\begin{lemma}[Eigenvalue and dominant-block eigenvector limits]
  \label{lem:rp*}
  Assume the unique global pricing dominance condition \eqref{eq:gvd}.
  Let $\phi_*,\psi_*\gg0$ be the right and left Perron eigenvectors of
  $L_{k^*}$ associated with $\lambda^*$, normalized by
  \(
    \1^\t\phi_*=1
  \)
  and
  \(
    \phi_*^\t\psi_*=1.
  \)
  Then
  \[
    \lambda^*(\epsilon)\to\lambda^*,
    \qquad
    \phi_{\s\Ccal^*}(\epsilon)\to\phi_*,
    \qquad
    \psi_{\s\Ccal^*}(\epsilon)\to\psi_*,
    \qquad
    \text{as }\epsilon\downarrow0.
  \]
\end{lemma}

\begin{proof}[Proof of Lemma~\ref{lem:rp*}]
  Since $\Lp(\epsilon)\to\Lp$ and the spectral bound is continuous on
  finite-dimensional matrix spaces,
  $\lambda^*(\epsilon) = s(\Lp(\epsilon))\to s(\Lp) = \lambda^*$
  as $\epsilon\downarrow0$.
  The irreducibility of $L_{k^*}$ implies that $\lambda^*=s(L_{k^*})$
  is an algebraically simple eigenvalue of $L_{k^*}$. By
  \eqref{eq:gvd}, $\lambda^*$ does not belong to the spectrum of any
  other diagonal block in the Frobenius form of $\Lp$. Therefore, $\lambda^*$ is an
  algebraically simple eigenvalue of $\Lp$.

  Let $P_*$ be the Riesz spectral projection of $\Lp$ associated with
  $\lambda^*$. Finite-dimensional perturbation theory implies that, for
  all sufficiently small $\epsilon>0$, the corresponding rank-one
  projection $P_*(\epsilon)$ is well defined and $P_*(\epsilon)\to P_*$;
  see \citet[Chapter~II, \S~5.1, especially Theorem~5.1 and equation~(5.2),
  pp.~107--108]{kato1995perturbation}. Hence the
  one-dimensional eigenspaces $\ran P_*(\epsilon)$ converge to $\ran P_*$.
  Let $\phi$ be the right eigenvector of $\Lp$ spanning $\ran P_*$ and
  normalized by $\1^\t\phi_{\s\Ccal^*}=1$.
  The block eigenvector equation on $\Ccal^*$ gives
  \(
    L_{k^*}\phi_{\s\Ccal^*}
    =
    \lambda^*\phi_{\s\Ccal^*}.
  \)
  By the uniqueness of the normalized Perron eigenvector of $L_{k^*}$,
  it follows that $\phi_{\s\Ccal^*}=\phi_*$.
  Since the normalizing functional $v\mapsto\1^\t v_{\s\Ccal^*}$ is nonzero on
  $\ran P_*$, convergence of the one-dimensional right eigenspaces, together with the
  common normalization $\1^\t\phi_{\s\Ccal^*}(\epsilon)=1$, yields
  $\phi(\epsilon)\to\phi$. Restricting this convergence to $\Ccal^*$ gives
  \(
    \phi_{\s\Ccal^*}(\epsilon)\to\phi_{\s\Ccal^*}=\phi_*
  \)
  as $\epsilon\downarrow0$.

  It remains to identify the limit of the left eigenvectors under the normalization
  $\phi(\epsilon)^\t\psi(\epsilon)=1$. Since $\lambda^*(\epsilon)$ is algebraically simple,
  \(
    P_*(\epsilon)
    =
    \phi(\epsilon)\psi(\epsilon)^\t.
  \)
  Fix any $x\in\Ccal^*$. The convergence $\phi_{\s\Ccal^*}(\epsilon)\to\phi_*$ implies
  \(
    \phi_{\s\Ccal^*}(\epsilon)(x)\to\phi_*(x).
  \)
  Since $\phi_*(x)>0$, $\phi_{\s\Ccal^*}(\epsilon)(x)$ is bounded away from zero for all sufficiently
  small $\epsilon>0$. Hence, as $\epsilon\downarrow0$,
  \[
    \psi(\epsilon)^\t
    =
    \frac{e_{x}^\t P_*(\epsilon)}{\phi_{\s\Ccal^*}(\epsilon)(x)}
    \to
    \frac{e_{x}^\t P_*}{\phi_*(x)}
    \eqqcolon
    \psi^\t.
  \]
  Since $\lambda^*$ is an algebraically simple eigenvalue of $\Lp$, the
  generalized left eigenspace coincides with the left eigenspace: $\ran(P_*^\t) =
  \ker(\Lp^\t-\lambda^* I)$. Therefore, $\psi = P_*^\t e_x/\phi_*(x) \in
  \ker(\Lp^\t-\lambda^* I)$, which implies that $\psi$ is a left eigenvector of $\Lp$
  associated with $\lambda^*$. Moreover, $\phi\in\ran(P_*)$ implies $P_*\phi=\phi$.
  Since $x\in\Ccal^*$ and $\phi_{\s\Ccal^*}=\phi_*$, it follows that
  $\psi^\t\phi = e_x^\t P_* \phi/\phi_*(x) = e_x^\t\phi/\phi_*(x) = 1$.
  The block left-eigenvector equations imply that $\psi_{\s\Ccal^*}$ is a strictly
  positive left Perron eigenvector of $L_{k^*}$. Moreover, the support
  characterization in Theorem~\ref{thm:Pi-asym} shows that $\phi$ and $\psi$
  have overlapping class support only on $\Ccal^*$. Consequently,
  \(
    1 = \phi^\t\psi = \phi_*^\t\psi_{\s\Ccal^*}.
  \)
  The normalization and uniqueness of the left Perron vector of $L_{k^*}$ thus
  give $\psi_{\s\Ccal^*}=\psi_*$. Restricting $\psi(\epsilon)\to\psi$ to $\Ccal^*$
  completes the proof.
\end{proof}

\begin{lemma}[Rarity orders on nondominant eigenvector blocks] \label{lem:rp}
  Under condition \eqref{eq:gvd} and Assumption~\ref{ass:edge}, for every
  $\Ccal\in\cC\setminus\{\Ccal^*\}$, there exist
  $\bar\phi_{\s\Ccal},\bar\psi_{\s\Ccal}\gg0$ such that
  \begin{equation} \label{eq:rare-path-lemma}
    \phi_{\s\Ccal}(\epsilon)
    =
    \epsilon^{d_{\s\Ccal\Ccal^*}}\big(\bar\phi_{\s\Ccal} + o(1)),
    \qquad
    \psi_{\s\Ccal}(\epsilon)
    =
    \epsilon^{d_{\s\Ccal^*\Ccal}}\big(\bar\psi_{\s\Ccal} + o(1)),
    \qquad
    \text{as }\epsilon\downarrow0,
  \end{equation}
  where $d_{\s\Ccal\Ccal^*},d_{\s\Ccal^*\Ccal}\geq0$ are the exponents defined in
  \eqref{eq:rdm}.
\end{lemma}

\begin{proof}[Proof of Lemma~\ref{lem:rp}]
  Let $\cC^\circ\coloneqq\cC\setminus\{\Ccal^*\}$ and $\Xsf^\circ =
  \Xsf\setminus\Ccal^*$. After a fixed simultaneous permutation of block rows and
  columns, place the block indexed by $\Ccal^*$ first and write
  \[
    \Lp(\epsilon)
    =
    \begin{pmatrix}
      L_{k^*k^*}(\epsilon) & c(\epsilon)\\
      b(\epsilon)          & M(\epsilon)
    \end{pmatrix},
    \qquad
    \phi(\epsilon)
    =
    \begin{pmatrix}
      \phi_{\s\Ccal^*}(\epsilon)\\ \phi_{\s\Xsf^\circ}(\epsilon)
    \end{pmatrix}.
  \]
  Here, $M(\epsilon)$ is the principal submatrix indexed by $\Xsf^\circ$
  and the block column $b(\epsilon)$ has entries $L_{k k^*}(\epsilon)$ for
  $\Ccal_k\in\cC^\circ$. The lower block of the eigenvector equation is
  \begin{equation}\label{eq:rce}
    \big(\lambda^*(\epsilon)I-M(\epsilon)\big)\,\phi_{\s\Xsf^\circ}(\epsilon) 
    =
    b(\epsilon)\, \phi_{\s\Ccal^*}(\epsilon).
  \end{equation}

  Let $M_0(\epsilon)$ retain the diagonal blocks of $M(\epsilon)$ and those
  off-diagonal blocks whose edge-rarity exponent is zero, and set
  \[
    E(\epsilon)\coloneqq M(\epsilon)-M_0(\epsilon).
  \]
  Thus $E(\epsilon)$ has zero diagonal blocks and retains exactly the
  positive finite-weight interclass blocks of $M(\epsilon)$. It represents the rare
  edges among the classes in $\cC^\circ$, whereas $M_0(\epsilon)$
  contains their within-class dynamics and non-rare interclass edges.
  In particular, $M_0(0)=M(0)$.

  By \eqref{eq:gvd}, $s(M(0)) =\max_{\Ccal_k\in\cC^\circ}s(L_k) <\lambda^*$.
  Lemma~\ref{lem:rp*} gives $\lambda^*(\epsilon)\to\lambda^*$, while
  $M_0(\epsilon)\to M(0)$ and continuity of the spectral bound give
  $s(M_0(\epsilon))\to s(M(0))$. Therefore,
  $\lambda^*(\epsilon)>s(M_0(\epsilon))$ for all sufficiently small $\epsilon>0$.
  Hence the resolvent of $M_0(\epsilon)$ at $\lambda^*(\epsilon)$ admits the Laplace
  representation:
  \[
    R_0(\epsilon)
    \coloneqq
    R\bigl(\lambda^*(\epsilon);M_0(\epsilon)\bigr)
    = \int_0^\infty \e^{-\lambda^*(\epsilon)\,t} \e^{t M_0(\epsilon)}\diff t
    \geq0,
    \qquad \epsilon>0\text{ sufficiently small.}
  \]
  Moreover,
  \(
    R_0(\epsilon) \to R(\lambda^*;M(0))\eqqcolon R_0(0)
  \)
  as $\epsilon\downarrow0$, so the family $\{R_0(\epsilon)\}$ is uniformly bounded.

  For all sufficiently small $\epsilon\geq0$, the positive within-class edges and
  the zero-weight interclass blocks of the limiting matrix persist. Hence the
  diagonal blocks of $M_0(\epsilon)$ remain irreducible, and its class-level graph
  is exactly the zero-weight subgraph induced by $\cC^\circ$. After a simultaneous
  permutation of these classes into Frobenius order, Lemma~\ref{lem:res-exp}
  implies that the $(\Ccal,\Ccal')$ block of $R_0(\epsilon)$ is strictly positive
  if $\Ccal\leadsto_0\Ccal'$ and is zero otherwise.

  Let $w_{\mathrm m} = \min\{w_{h\ell}:h\neq\ell,\ 0<w_{h\ell}<\infty\}>0$. Such a positive
  minimum exists because $\Gcal_w$ is strongly connected whereas its
  zero-weight subgraph is reducible. Since
  $\|E(\epsilon)\|=O(\epsilon^{w_{\mathrm m}})$, the resolvent identity gives the
  uniformly convergent Neumann expansion
  \begin{equation} \label{eq:rare-neumann}
    R\bigl(\lambda^*(\epsilon);M(\epsilon)\bigr)
    =
    R_0(\epsilon) \sum_{n=0}^{\infty} \big(E(\epsilon)R_0(\epsilon)\big)^n.
  \end{equation}
  Indeed, uniform boundedness of $R_0(\epsilon)$ gives
  $\|E(\epsilon)R_0(\epsilon)\|=O(\epsilon^{w_{\mathrm m}})<1$ for all sufficiently
  small $\epsilon$.

  Combining \eqref{eq:rce} and \eqref{eq:rare-neumann} gives
  \begin{equation} \label{eq:rare-right-series}
    \phi_{\s\Xsf^\circ}(\epsilon)
    =
    \sum_{n=0}^\infty
    R_0(\epsilon) \big(E(\epsilon)R_0(\epsilon)\big)^n\,
    b(\epsilon)\phi_{\s\Ccal^*}(\epsilon).
  \end{equation}
  Because the series converges in norm, each block may be expanded term by term
  over the finitely many intermediate class indices.
  For $k\neq\ell$, write
  $w_{\s\Ccal_k\Ccal_\ell}\coloneqq w_{k\ell}$. Fix
  $\Ccal\in\cC^\circ$. For $n\geq0$, let
  $\mathscr A_{\s\Ccal,n}$ consist of the tuples
  \(
    \iota=(\Ccal_0^+,\Ccal_1^-,\Ccal_1^+,\ldots,\Ccal_n^-,\Ccal_n^+)
  \)
  in $(\cC^\circ)^{2n+1}$ for which
  \begin{align*}
    \Ccal\leadsto_0\Ccal_0^+,
    \quad
    \Ccal_j^-\neq\Ccal_{j-1}^+,
    \quad \ 0<w_{\s\Ccal_{j-1}^+\Ccal_j^-}<\infty,
    \quad
    \Ccal_j^-\leadsto_0\Ccal_j^+,
    \quad j=1,\ldots,n, \quad
    w_{\s\Ccal_n^+\Ccal^*}<\infty.
  \end{align*}
  where $\leadsto_0$ denotes accessibility in the zero-weight subgraph induced by
  $\cC^\circ$, including a path of length zero. Here $\Ccal_j^-$ is the class
  reached by the $j$-th rare edge, and $\Ccal_j^+$ is the class reached after the
  following zero-weight path segment. For $n\geq1$, $\iota$ represents a walk of
  the form
  \[
    \Ccal
    \leadsto_0\Ccal_0^+
    \xrightarrow{\,w>0\,}\Ccal_1^-
    \leadsto_0\Ccal_1^+
    \xrightarrow{\,w>0\,}\cdots\xrightarrow{\,w>0\,}
    \Ccal_n^-
    \leadsto_0\Ccal_n^+
    \xrightarrow{\,w\geq0\,}\Ccal^*,
  \]
  where the labeled arrows are single edges. For $n=0$, the walk is
  \(
    \Ccal\leadsto_0\Ccal_0^+
    \xrightarrow{\,w\geq0\,}\Ccal^*.
  \)
  The block products in \eqref{eq:rare-right-series} are precisely
  \begin{equation} \label{eq:rare-block-product}
    T_\iota(\epsilon)
    \coloneqq
    [R_0(\epsilon)]_{\s\Ccal\Ccal_0^+}
    \prod_{j=1}^n
    \left(
      [E(\epsilon)]_{\s\Ccal_{j-1}^+\Ccal_j^-}
      [R_0(\epsilon)]_{\s\Ccal_j^-\Ccal_j^+}
    \right)
    b_{\s\Ccal_n^+}(\epsilon)
    \phi_{\s\Ccal^*}(\epsilon),
  \end{equation}
  with the product ordered from $j=1$ to $j=n$ and interpreted as the identity
  when $n=0$. Thus
  \begin{equation} \label{eq:rare-block-sum}
    \phi_{\s\Ccal}(\epsilon)
    =
    \sum_{n=0}^\infty
    \sum_{\iota\in\mathscr A_{\s\Ccal,n}}
    T_\iota(\epsilon).
  \end{equation}
  Associate with $\iota\in\mathscr A_{\s\Ccal,n}$ the total edge weight
  \begin{equation} \label{eq:rare-block-exponent}
    w(\iota)
    \coloneqq
    \sum_{j=1}^n w_{\s\Ccal_{j-1}^+\Ccal_j^-}
    +w_{\s\Ccal_n^+\Ccal^*}.
  \end{equation}
  Assumption~\ref{ass:edge}, Lemma~\ref{lem:rp*}, and
  $R_0(\epsilon)\to R_0(0)$ imply, for each fixed $\iota$,
  \begin{equation} \label{eq:rare-block-limit}
    T_\iota(\epsilon)
    =
    \epsilon^{w(\iota)}
    \bigl(H_\iota+o(1)\bigr),
    \qquad H_\iota\gg0.
  \end{equation}
  Reading the limiting product from right to left proves $H_\iota\gg0$: a nonzero
  nonnegative edge matrix sends a strictly positive vector to a nonzero nonnegative
  vector, and the next strictly positive resolvent block restores strict positivity.
  Each $\iota$ also encodes at least one directed walk from $\Ccal$ to $\Ccal^*$ with
  weight $w(\iota)$; deleting cycles produces a path with no greater weight.
  Conversely, separate the final edge into $\Ccal^*$ and decompose the preceding
  part of any directed path into maximal zero-weight segments and intervening
  positive-weight edges. This produces an element of some
  $\mathscr A_{\s\Ccal,n}$ with the same weight. Therefore,
  \begin{equation} \label{eq:rare-minimum-block-exponent}
    \min_{\iota\in\bigcup_{n\geq0}\mathscr A_{\s\Ccal,n}}w(\iota)
    =
    \min_{\gamma\in\dD_w(\Ccal,\Ccal^*)}W(\gamma)
    =
    d_{\s\Ccal\Ccal^*}.
  \end{equation}
  Choose $N_{\s\Ccal}\in\NN_0$ such that
  $(N_{\s\Ccal}+1)w_{\mathrm m}>d_{\s\Ccal\Ccal^*}$. Uniform boundedness of
  $R_0(\epsilon)$, $b(\epsilon)$, and $\phi_{\s\Ccal^*}(\epsilon)$, together with
  $\|E(\epsilon)\|=O(\epsilon^{w_{\mathrm m}})$, gives
  \begin{equation} \label{eq:rare-neumann-tail}
    \sum_{n=N_{\s\Ccal}+1}^\infty
    \left\|\left[
      R_0(\epsilon)\big(E(\epsilon)R_0(\epsilon)\big)^n
      b(\epsilon)\phi_{\s\Ccal^*}(\epsilon)
    \right]_{\s\Ccal}\right\|
    \leq
    C_0\sum_{n=N_{\s\Ccal}+1}^\infty
    \bigl(C_1\epsilon^{w_{\mathrm m}}\bigr)^n
    =O\bigl(\epsilon^{(N_{\s\Ccal}+1)w_{\mathrm m}}\bigr)
    =o\bigl(\epsilon^{d_{\s\Ccal\Ccal^*}}\bigr).
  \end{equation}
  Since $w(\iota)\geq n w_{\mathrm m}$ for
  $\iota\in\mathscr A_{\s\Ccal,n}$, every minimizer in
  \eqref{eq:rare-minimum-block-exponent} occurs with $n\leq N_{\s\Ccal}$. Define
  \begin{equation*}
    \mathscr A_{\s\Ccal}^{\min}
    \coloneqq
    \left\{
      \iota\in\bigcup_{n=0}^{N_{\s\Ccal}}\mathscr A_{\s\Ccal,n}:
      w(\iota)=d_{\s\Ccal\Ccal^*}
    \right\},
    \qquad
    \bar\phi_{\s\Ccal}
    \coloneqq
    \sum_{\iota\in\mathscr A_{\s\Ccal}^{\min}}H_\iota.
  \end{equation*}
  This set is nonempty because the minimum over directed paths is attained, and
  \eqref{eq:rare-block-limit} gives $\bar\phi_{\s\Ccal}\gg0$. The part of
  \eqref{eq:rare-block-sum} with $n\leq N_{\s\Ccal}$ is finite: after division by
  $\epsilon^{d_{\s\Ccal\Ccal^*}}$, its minimum-weight terms converge to
  $\bar\phi_{\s\Ccal}$ and all other terms vanish. Combining this fact with
  \eqref{eq:rare-neumann-tail} yields
  \(
    \epsilon^{-d_{\s\Ccal\Ccal^*}}
    \phi_{\s\Ccal}(\epsilon) \to \bar\phi_{\s\Ccal}\gg0,
  \)
  proving the first expansion in \eqref{eq:rare-path-lemma}.

  For the left eigenvector, apply the same path-expansion argument to
  $\Lp(\epsilon)^\t$. Transposition preserves the limiting classes and their
  spectral bounds while reversing every edge without changing its weight. Thus
  the relevant minimum weight is $d_{\s\Ccal^*\Ccal}$ in the original graph.
  Lemma~\ref{lem:rp*} also gives
  $\psi_{\s\Ccal^*}(\epsilon)\to\psi_*\gg0$, the only dominant-block convergence
  used above. Consequently, there exists $\bar\psi_{\s\Ccal}\gg0$ such that
  \(
    \epsilon^{-d_{\s\Ccal^*\Ccal}}
    \psi_{\s\Ccal}(\epsilon) \to \bar\psi_{\s\Ccal},
  \)
  proving the second expansion in \eqref{eq:rare-path-lemma}.
\end{proof}

\begin{proof}[Proof of Theorem~\ref{thm:rare-distance}]
  Lemma~\ref{lem:rp*} gives $\lambda^*(\epsilon)\to\lambda^*$ as
  $\epsilon\downarrow0$. For the dominant class, set
  \(
    \bar\phi_{\s\Ccal^*}\coloneqq\phi_*
  \)
  and
  \(
    \bar\psi_{\s\Ccal^*}\coloneqq\psi_*,
  \)
  where $\phi_*$ and $\psi_*$ are defined in Lemma~\ref{lem:rp*}.
  Then, Lemma~\ref{lem:rp*}, together with $d_{\s\Ccal^*\Ccal^*}=0$, and
  Lemma~\ref{lem:rp} imply that for each $\Ccal\in\cC$, there are vectors
  $\bar\phi_{\s\Ccal},\bar\psi_{\s\Ccal}\gg0$ such that
  \begin{equation}\label{eq:pps}
    \phi_{\s\Ccal}(\epsilon)
    =
    \epsilon^{d_{\s\Ccal\Ccal^*}} \bigl(\bar\phi_{\s\Ccal}+o(1)\bigr),
    \qquad
    \psi_{\s\Ccal}(\epsilon)
    =
    \epsilon^{d_{\s\Ccal^*\Ccal}} \bigl(\bar\psi_{\s\Ccal}+o(1)\bigr),
    \qquad
    \text{as }\epsilon\downarrow0.
  \end{equation}
  The normalization imposed at the beginning of this subsection gives
  \(
    P_*(\epsilon)=\phi(\epsilon)\psi(\epsilon)^\t
  \)
  and
  \(
    P_{*,\s\Ccal\Ccal'}(\epsilon)
    =\phi_{\s\Ccal}(\epsilon)\psi_{\s\Ccal'}(\epsilon)^\t.
  \)
  Therefore, defining
  \(
    \bar P_{\s\Ccal\Ccal'}
    \coloneqq
    \bar\phi_{\s\Ccal}\bar\psi_{\s\Ccal'}^\t\gg0,
  \)
  equation~\eqref{eq:pps} yields
  \[
    \epsilon^{-\left(
      d_{\s\Ccal\Ccal^*}+d_{\s\Ccal^*\Ccal'}
    \right)}
    P_{*,\s\Ccal\Ccal'}(\epsilon)
    \to
    \bar P_{\s\Ccal\Ccal'},
    \qquad\text{as } \epsilon\downarrow0,
  \]
  which proves \eqref{eq:rp}.
  By the definition $\pi_\epsilon^{\mathrm{LR}}(x)=P_*(\epsilon)(x,x)$,
  \(
    \sum_{x\in\Ccal}\pi_\epsilon^{\mathrm{LR}}(x)
    =
    \operatorname{tr}P_{*,\s\Ccal\Ccal}(\epsilon).
  \)
  Taking $\Ccal'=\Ccal$ in \eqref{eq:rp} and taking traces then gives \eqref{eq:rv}.

  Finally, if $\Ccal\neq\Ccal^*$, then $d_{\s\Ccal\Ccal^*}+d_{\s\Ccal^*\Ccal}>0$.
  Otherwise, nonnegativity of the edge weights would give zero-weight paths in both
  directions between two distinct limiting classes, which yields a contradiction.
  Equation~\eqref{eq:rv} therefore implies
  \(
    \sum_{x\in\Ccal}\pi_\epsilon^{\mathrm{LR}}(x)\to0
  \)
  for every $\Ccal\neq\Ccal^*$. Since $P_*(\epsilon)$ is a rank-one projection,
  \(
    \sum_{x\in\Xsf}\pi_\epsilon^{\mathrm{LR}}(x)
    =\operatorname{tr}P_*(\epsilon)=1.
  \)
  Hence
  \(
    \sum_{x\in\Ccal^*}\pi_\epsilon^{\mathrm{LR}}(x)
    =1-\sum_{x\in\Xsf\setminus\Ccal^*}\pi_\epsilon^{\mathrm{LR}}(x) \to 1,
  \)
  as $\epsilon\downarrow0$.
\end{proof}

We summarize the projection notation introduced in Sections~\ref{sec:rare-entry} and
\ref{sec:crossover} that will be used below. For every sufficiently small
$\epsilon>0$, $P_*(\epsilon)$ is the rank-one Riesz projection of $\Lp(\epsilon)$
associated with its Perron eigenvalue $\lambda^*(\epsilon)$, whereas
$P_\circ(\epsilon)$ is the Riesz projection of $\Lp(\epsilon)$ associated with the
spectral set $\Sigma_\circ(\epsilon)$ defined in \eqref{eq:Sc}. At $\epsilon=0$,
these projections satisfy $P_*(0)=P_*$ and $P_\circ(0)=P_\circ$, where $P_*$ and
$P_\circ$ are Riesz projections of $\Lp$ associated with its eigenvalue $\lambda^*$
and spectral set $\Sigma_\circ$, respectively.
Moreover, $s_\circ = s\big(\Lp\,|_{\ker P_*}\big)$.

\begin{lemma} \label{lem:sd}
  Under Assumption~\ref{ass:edge} and the unique-dominant-class condition
  \eqref{eq:gvd}, there exist constants $\epsilon_0,C,\delta>0$ such that, for
  every $\epsilon\in(0,\epsilon_0]$ and $t\geq0$, setting
  \(
    \mathcal E_\epsilon(t)
    \coloneqq
    \e^{t\Lp(\epsilon)}\big(I-P_*(\epsilon)-P_\circ(\epsilon)\big),
  \)
  we have
  \begin{equation} \label{eq:cs}
    \e^{t\Lp(\epsilon)}
    =
    \e^{\lambda^*(\epsilon)t}P_*(\epsilon) 
    +
    \e^{t\Lp(\epsilon)}P_\circ(\epsilon) 
    +
    \mathcal E_\epsilon(t),
    \qquad
    \|\mathcal E_\epsilon(t)\| \leq C\e^{(s_\circ-\delta)t}.
  \end{equation}
  Moreover, as $\epsilon\downarrow0$,
  \begin{equation} \label{eq:cp}
    \|P_*(\epsilon)-P_*\| +\|P_\circ(\epsilon)-P_\circ\|
    =
    O(\|\Lp(\epsilon)-\Lp\|),
    \qquad
    \lambda^*(\epsilon)
    =
    \lambda^*+O(\|\Lp(\epsilon)-\Lp\|).
  \end{equation}
\end{lemma}

\begin{proof}[Proof of Lemma~\ref{lem:sd}]
  Choose disjoint positively oriented contours $\Gamma_*$ and $\Gamma_\circ$
  such that $\Gamma_*$ encloses $\{\lambda^*\}$ and no other point of
  $\sigma(\Lp)$, while $\Gamma_\circ$ satisfies $\Sigma_\circ \subset
  \operatorname{int}(\Gamma_\circ)$ and $\overline{\operatorname{int}(\Gamma_\circ)}
  \subset U_\circ$.
  Let 
  $\Sigma_\dagger \coloneqq \sigma(\Lp)\setminus(\{\lambda^*\}\,\cup\,\Sigma_\circ)$.
  If $\Sigma_\dagger\neq\varnothing$, finiteness of the spectrum permits a
  $\delta>0$ and a positively oriented contour $\Gamma_\dagger$, disjoint from
  $\Gamma_*$ and $\Gamma_\circ$, such that
  \begin{equation} \label{eq:rcb}
    \Sigma_\dagger\subset\operatorname{int}(\Gamma_\dagger),
    \qquad
    \max_{z\in\Gamma_\dagger}\Re z\leq s_\circ-\delta.
  \end{equation}
  If $\Sigma_\dagger=\varnothing$, fix any $\delta>0$ and set $\Gamma_\dagger =\varnothing$.

  Since $\lambda^*(\epsilon)\to\lambda^*$, as shown in
  Theorem~\ref{thm:rare-distance}, and the spectrum is continuous under
  finite-dimensional matrix perturbations, there exists $\epsilon_0>0$ such that,
  for every $\epsilon\in[0,\epsilon_0]$,
  \[
    \sigma(\Lp(\epsilon))
    \cap\operatorname{int}(\Gamma_*)
    =
    \{\lambda^*(\epsilon)\},
    \qquad
    \sigma(\Lp(\epsilon))
    \cap\operatorname{int}(\Gamma_\circ)
    =
    \Sigma_\circ(\epsilon).
  \]
  When $\Sigma_\dagger\neq\varnothing$, we may also choose $\epsilon_0$ so that
  \(
    \sigma(\Lp(\epsilon))
    \cap\operatorname{int}(\Gamma_\dagger)
    =
    \sigma(\Lp(\epsilon))
    \setminus
    \bigl(\{\lambda^*(\epsilon)\}\cup\Sigma_\circ(\epsilon)\bigr).
  \)
  Consequently, for every $\epsilon\in[0,\epsilon_0]$, the Riesz projections
  $P_*(\epsilon)$ and $P_\circ(\epsilon)$ have the contour integral representations
  \begin{equation}\label{eq:crep}
    P_*(\epsilon)
    =
    \frac{1}{2\pi i}
    \int_{\Gamma_*}R(z;\Lp(\epsilon))\,\diff z,
    \qquad
    P_\circ(\epsilon)
    =
    \frac{1}{2\pi i}
    \int_{\Gamma_\circ}R(z;\Lp(\epsilon))\,\diff z.
  \end{equation}
  When $\Sigma_\dagger\neq\varnothing$, define additionally
  \(
    P_\dagger(\epsilon) \coloneqq
    (2\pi i)^{-1} \int_{\Gamma_\dagger} R(z;\Lp(\epsilon))\,\diff z
  \)
  for $\epsilon\in[0,\epsilon_0]$.
  In particular, for every $\epsilon\in(0,\epsilon_0]$, the spectral sets enclosed by $\Gamma_*$,
  $\Gamma_\circ$, and, when $\Sigma_\dagger\neq\varnothing$, by $\Gamma_\dagger$ form
  a partition of $\sigma(\Lp(\epsilon))$.
  Therefore, $I-P_*(\epsilon)-P_\circ(\epsilon)=P_\dagger(\epsilon)$ if
  $\Sigma_\dagger\neq\varnothing$ and equals zero if $\Sigma_\dagger=\varnothing$.
  Moreover, $\Lp(\epsilon)P_*(\epsilon) =\lambda^*(\epsilon)P_*(\epsilon)$.
  Multiplying the resulting decomposition of the identity by $\e^{t\Lp(\epsilon)}$
  proves the identity in \eqref{eq:cs}.

  Next, we prove the bound in \eqref{eq:cs}.
  If $\Sigma_\dagger=\varnothing$, then $I-P_*(\epsilon)-P_\circ(\epsilon)=0$, and
  hence $\mathcal E_\epsilon(t)=0$, so the bound is immediate. 
  Suppose now that
  $\Sigma_\dagger\neq\varnothing$ and hence 
  $I-P_*(\epsilon)-P_\circ(\epsilon) = P_\dagger(\epsilon)$. Let
  \[
    A_\epsilon
    \coloneqq
    \Lp(\epsilon)\big|_{\ran P_\dagger(\epsilon)},
    \qquad
    B_\epsilon
    \coloneqq
    \Lp(\epsilon)\big|_{\ker P_\dagger(\epsilon)}.
  \]
  Since $P_\dagger(\epsilon)$ is the Riesz projection associated with the spectrum
  enclosed by $\Gamma_\dagger$, under the invariant decomposition
  \(
    \CC^\Xsf
    =
    \ran P_\dagger(\epsilon)
    \oplus
    \ker P_\dagger(\epsilon),
  \)
  the resolvent decomposes into
  \begin{equation}\label{eq:Rz1}
    R(z;\Lp(\epsilon))
    =
    R(z;A_\epsilon)P_\dagger(\epsilon)
    +
    R(z;B_\epsilon)(I-P_\dagger(\epsilon)),
    \qquad z\in\CC\setminus\sigma(\Lp(\epsilon)).
  \end{equation}
  Since the contour \(\Gamma_\dagger\) encloses \(\sigma(A_\epsilon)\) but no point of
  \(\sigma(B_\epsilon)\), \(R(z;B_\epsilon)\) is holomorphic on
  an open neighborhood of \(\overline{\operatorname{int}(\Gamma_\dagger)}\).
  It follows from Cauchy’s theorem that
  \begin{equation}\label{eq:Rz2}
    \int_{\Gamma_\dagger} \e^{tz}
    R(z;B_\epsilon) \bigl(I-P_\dagger(\epsilon)\bigr)\diff z = 0.
  \end{equation}
  By \eqref{eq:Rz1}, \eqref{eq:Rz2}, and the Dunford--Taylor formula applied to
  $A_\epsilon$ with $f(z)=\e^{tz}$
  \citep[see, e.g.,][Chapter~I, equation~(5.47), p.~44]{kato1995perturbation}, we have
  \begin{equation}\label{eq:mE}
    \frac{1}{2\pi i} \int_{\Gamma_\dagger} \e^{tz} R(z;\Lp(\epsilon)) \diff z
    =
    \frac{1}{2\pi i} \int_{\Gamma_\dagger} \e^{tz} R(z;A_\epsilon) \diff z\,
    P_\dagger(\epsilon)
    =
    \e^{tA_\epsilon} P_\dagger(\epsilon)
    =
    \e^{t\Lp(\epsilon)} P_\dagger(\epsilon)
    =
    \mathcal E_\epsilon(t).
  \end{equation}

  To analyze the resolvent $R(z;\Lp(\epsilon))$ on the left side of \eqref{eq:mE}, we 
  set $\Delta_\epsilon\coloneqq\Lp(\epsilon)-\Lp$, and note 
  \begin{equation}\label{eq:zIL}
    zI-\Lp(\epsilon) = \bigl(I-\Delta_\epsilon R(z;\Lp)\bigr)(zI-\Lp).
  \end{equation}
  Let $\Gamma = \Gamma_*\cup\Gamma_\circ\cup\Gamma_\dagger$. Because \(\Gamma\) is
  compact and does not intersect \(\sigma(\Lp)\), the resolvent $R(z;\Lp)$ is
  continuous on \(\Gamma\). Hence 
  \begin{equation}\label{eq:mg}
    M_{\s\Gamma}\coloneqq\sup_{z\in\Gamma}\|R(z;\Lp)\|<\infty.
  \end{equation}
  Since \(\|\Delta_\epsilon\|\to0\) as $\epsilon\downarrow0$, reducing
  $\epsilon_0$ if necessary gives
  \[
    \sup_{0<\epsilon\leq\epsilon_0}
    \sup_{z\in\Gamma}
    \|\Delta_\epsilon R(z;\Lp)\|
    \leq
    M_{\s\Gamma}
    \sup_{0<\epsilon\leq\epsilon_0}
    \|\Delta_\epsilon\|
    <\frac12.
  \]
  Consequently, for every $\epsilon\in (0,\epsilon_0]$ and $z\in\Gamma$, $I-\Delta_\epsilon
  R(z;\Lp)$ is invertible, with
  \(
    \bigl(I-\Delta_\epsilon R(z;\Lp)\bigr)^{-1}
    =
    \sum_{n=0}^{\infty} \bigl(\Delta_\epsilon R(z;\Lp)\bigr)^n.
  \)
  Therefore, uniformly over $\epsilon\in (0,\epsilon_0]$ and $z\in\Gamma$,
  \[
    \big\|\big(I-\Delta_\epsilon R(z;\Lp)\big)^{-1}\big\|
    \leq \sum_{n=0}^\infty \|\Delta_\epsilon R(z;\Lp)\|^n 
    =\frac{1}{1-\|\Delta_\epsilon R(z;\Lp)\|}
    \leq 2.
  \]
  It then follows from \eqref{eq:zIL} that
  \(
    R(z;\Lp(\epsilon)) = R(z;\Lp)\bigl(I-\Delta_\epsilon R(z;\Lp)\bigr)^{-1}
  \)
  for every $\epsilon\in(0,\epsilon_0]$ and $z\in\Gamma$, and
  \begin{equation}\label{eq:urb1}
    \sup_{0<\epsilon\leq\epsilon_0}
    \sup_{z\in\Gamma}\|R(z;\Lp(\epsilon))\|
    \leq
    \sup_{z\in\Gamma}\|R(z;\Lp)\|\,
    \sup_{0<\epsilon\leq\epsilon_0}
    \sup_{z\in\Gamma}\big\|\big(I-\Delta_\epsilon R(z;\Lp)\big)^{-1}\big\|
    \leq 2 M_{\s\Gamma}.
  \end{equation}
  By \eqref{eq:rcb}, \eqref{eq:mE} and \eqref{eq:urb1}, there exists a constant $C>0$
  such that, for every $t\geq0$, 
  \begin{equation*}
    \sup_{0<\epsilon\leq\epsilon_0}
    \|\mathcal E_\epsilon(t)\|
    \leq
    \frac{|\Gamma_\dagger|}{2\pi}
    \sup_{0<\epsilon\leq\epsilon_0}
    \sup_{z\in\Gamma_\dagger}\|R(z;\Lp(\epsilon))\|
    \,
    \e^{t\max_{z\in\Gamma_\dagger}\Re z}
    \leq C\e^{(s_\circ-\delta)t},
  \end{equation*}
  where $|\Gamma_\dagger|$ denotes the total arclength of $\Gamma_\dagger$.
  This proves the bound in \eqref{eq:cs}.

  Finally, we prove the two estimates in \eqref{eq:cp}. For $j\in\{*,\circ\}$, it
  follows from \eqref{eq:crep} that $P_j(\epsilon)-P_j(0) = (2\pi
  i)^{-1}\int_{\Gamma_j} [R(z;\Lp(\epsilon)) - R(z;\Lp)]\diff z$.
  By the resolvent identity
  \(
    R(z;\Lp(\epsilon))-R(z;\Lp) = R(z;\Lp(\epsilon))\Delta_\epsilon R(z;\Lp),
  \)
  \eqref{eq:mg}, and \eqref{eq:urb1}, we have 
  \(
    \sup_{z\in\Gamma} \|R(z;\Lp(\epsilon))-R(z;\Lp)\|
    \leq
    2M_{\s\Gamma}^2\|\Delta_\epsilon\|
  \)
  for every $\epsilon\in(0,\epsilon_0]$. Therefore,
  \[
    \|P_j(\epsilon)-P_j(0)\|
    \leq
    \frac{|\Gamma_j|}{2\pi}
    \sup_{z\in\Gamma} \|R(z;\Lp(\epsilon))-R(z;\Lp)\|
    \leq
    |\Gamma_j| M_{\s\Gamma}^2\|\Delta_\epsilon\|/\pi 
    =
    O(\|\Delta_\epsilon\|),
  \]
  which proves the first estimate in \eqref{eq:cp}.
  For the second estimate in \eqref{eq:cp}, note that both $P_*(\epsilon)$ and $P_*$
  have rank one. Moreover, it follows from 
  $\Lp(\epsilon)P_*(\epsilon) = \lambda^*(\epsilon)P_*(\epsilon)$ and $\Lp P_* = \lambda^*P_*$
  that
  \[
    \lambda^*(\epsilon)
    =\tr\bigl(\Lp(\epsilon)P_*(\epsilon)\bigr),
    \qquad
    \lambda^*=\tr(\Lp P_*).
  \]
  Consequently,
  \(
    \lambda^*(\epsilon) - \lambda^*
    =
    \tr\big(\Lp(\epsilon)P_*(\epsilon) - \Lp P_* \big)
    =
    \tr\big(\Delta_\epsilon P_*(\epsilon)\big) 
    +
    \tr\big( \Lp (P_*(\epsilon)-P_*) \big),
  \)
  and hence 
  \[
    |\lambda^*(\epsilon) - \lambda^*|
    \leq
    \left|\tr\bigl(\Delta_\epsilon P_*(\epsilon)\bigr)\right|
    +
    \left|\tr\big( \Lp (P_*(\epsilon)-P_*) \big)\right|
    = O(\|\Delta_\epsilon\|).
  \]
  This proves the second estimate in \eqref{eq:cp} and completes the proof.
\end{proof}

\begin{proof}[Proof of Theorem~\ref{thm:cd}]
  Apply the operator identity \eqref{eq:cs} in Lemma~\ref{lem:sd} to $g$, evaluate
  at $x$, and set $R_{\epsilon,t}\coloneqq (\mathcal E_\epsilon(t)g)(x)$.
  This gives \eqref{eq:usg}; the bound on $\mathcal E_\epsilon(t)$ gives its uniform
  remainder bound after enlarging the constant. The eigenvalue estimate in \eqref{eq:ggo} follows from
  \eqref{eq:cp} and the definition of $r_L(\epsilon)$, while \eqref{eq:rpl} gives
  $a_*(x,g;\epsilon)=a\epsilon^d(1+o(1))$ for some $a>0$. This proves
  \eqref{eq:ggo}.

  Fix $T<\infty$. Duhamel's identity and
  $\|\Lp(\epsilon)-\Lp\|=r_L(\epsilon)$ give
  $\e^{t\Lp(\epsilon)}-\e^{t\Lp}=O(r_L(\epsilon))$ uniformly over
  $0\leq t\leq T$. Together with the projection estimate in \eqref{eq:cp}, this
  proves the stated compact-maturity convergence of $b_t(x,g;\epsilon)$.

  To identify its limit, fix $\lambda\in\Sigma_\circ$ and set
  $N_\lambda\coloneqq(\Lp-\lambda I)P_\lambda$. On the invariant subspace
  $\ran P_\lambda$, the operator $\Lp$ equals $\lambda I+N_\lambda$, where
  $N_\lambda^{\nu_\lambda}=0$. Hence
  \(
    \e^{t\Lp}P_\lambda
    =
    \e^{\lambda t}
    \sum_{j=0}^{\nu_\lambda-1}
    \frac{t^j}{j!}(\Lp-\lambda I)^jP_\lambda.
  \)
  Summing over $\lambda\in\Sigma_\circ$, using
  $P_\circ=\sum_{\lambda\in\Sigma_\circ}P_\lambda$, and applying
  $e_x^\t(\cdot)g$ proves \eqref{eq:lsg}.

  Suppose now that $\lambda_{\kK}<s_\circ$. Exact corridor restriction gives
  $q_t(x,g;0)=e_{x,\kK}^\t e^{tL_{\kK}}g_{\kK}$ for every $t\geq0$, where
  $e_{x,\kK}$ and $g_{\kK}$ denote the restrictions to $\Xsf_{\kK}$. Taking
  Laplace transforms for $\Re z$ sufficiently large yields
  \begin{equation} \label{eq:cdcr}
    e_x^\t R(z;\Lp)g
    =
    e_{x,\kK}^\t R(z;L_{\kK})g_{\kK}.
  \end{equation}
  Both sides are rational functions of $z$ and agree on a right half-plane, so
  \eqref{eq:cdcr} extends as an identity wherever both sides are
  defined. For $\lambda\in\Sigma_\circ$, we have
  $\Re\lambda=s_\circ>\lambda_{\kK}=s(L_{\kK})$; hence
  $\lambda\notin\sigma(L_{\kK})$, and the right-hand side is holomorphic near
  $\lambda$. The Laurent expansion of the left side of \eqref{eq:cdcr} is
  \[
    e_x^\t R(z;\Lp)g
    =
    \sum_{j=0}^{\nu_\lambda-1}
    \frac{e_x^\t(\Lp-\lambda I)^jP_\lambda g}{(z-\lambda)^{j+1}}
    +h_\lambda(z),
  \]
  where $h_\lambda$ is holomorphic near $\lambda$. Holomorphy of the right-hand side
  of \eqref{eq:cdcr} therefore forces every coefficient in the
  principal part above to vanish. Equation~\eqref{eq:lsg} then gives
  $b_t(x,g;0)\equiv0$.
\end{proof}

\begin{proof}[Proof of Corollary~\ref{cor:noncommuting}]
  Fix a sufficiently small $\epsilon>0$. By Assumption~\ref{ass:edge},
  $\Lp(\epsilon)$ is irreducible. Theorem~\ref{thm:Pi-asym}, specifically the
  rank-one Perron limit \eqref{eq:r1lim}, applied to $\Lp(\epsilon)$,
  therefore gives
  \[
    \e^{-\lambda^*(\epsilon)t}q_t(x,g;\epsilon)
    \to
    \big(P_*(\epsilon)g\big)(x)
    =a_*(x,g;\epsilon)>0.
  \]
  Here strict positivity follows because the right and left Perron vectors of the
  irreducible Metzler matrix $\Lp(\epsilon)$ are strictly positive and $g\geq0$ is
  nonzero.
  Thus $q_t(x,g;\epsilon)=\e^{\lambda^*(\epsilon)t}
  (a_*(x,g;\epsilon)+o(1))$, and taking logarithms gives
  \[
    \lim_{t\to\infty}
    \left[-\frac1t\log q_t(x,g;\epsilon)\right]
    =-\lambda^*(\epsilon).
  \]
  Taking $\epsilon\downarrow0$ on both sides proves the first limit in
  \eqref{eq:noncommuting}.

  For the reverse order, fix $t>0$ and set
  $\Delta_\epsilon\coloneqq\Lp(\epsilon)-\Lp$. Duhamel's identity, also known as
  the variation-of-constants identity, gives
  \[
    \e^{t\Lp(\epsilon)}-\e^{t\Lp}
    =
    \int_0^t
    \e^{(t-s)\Lp(\epsilon)}\Delta_\epsilon \e^{s\Lp}\diff s.
  \]
  Since $\Lp(\epsilon)\to\Lp$, the two semigroup factors in the integrand are
  uniformly bounded for sufficiently small $\epsilon$ and $0\leq s\leq t$.
  Hence, for some $C_t<\infty$ independent of $\epsilon$,
  \(
    \|\e^{t\Lp(\epsilon)}-\e^{t\Lp}\| \leq C_t\|\Delta_\epsilon\|\to0.
  \)
  Applying this operator convergence to $g$ and evaluating at $x$ gives
  $q_t(x,g;\epsilon)\to q_t(x,g;0)$. The corridor is nonempty by hypothesis, so
  Proposition~\ref{prop:corridor}(ii) gives
  $q_t(x,g;0)>0$. Continuity of the logarithm therefore yields
  \[
    \lim_{\epsilon\downarrow0}
    \left[-\frac1t\log q_t(x,g;\epsilon)\right]
    =
    -\frac1t\log q_t(x,g;0).
  \]
  Finally, the yield expansion \eqref{eq:yield-clim} in
  Theorem~\ref{thm:pyc} implies that the right-hand side converges to
  $-\lambda(\Ccal,g)=-\lambda_{\kK}$ as $t\to\infty$. This proves the second limit in
  \eqref{eq:noncommuting}.
\end{proof}

\begin{proof}[Proof of Theorem~\ref{thm:crossover}]
  Under Assumption~\ref{ass:single-corridor-eigenvalue},
  $\Sigma_\circ=\{\lambda_\circ\}$ and
  $\lambda_\circ=\lambda_{\kK}$ is algebraically simple. Hence \eqref{eq:lsg}
  gives $b_t(x,g;0)=\e^{\lambda_\circ t}b_\circ(x,g;0)$. Moreover, $d>0$ makes
  the global coefficient $\big(P_*g\big)(x)$ vanish, while every remaining
  spectral term has real part strictly below $\lambda_\circ$. The full spectral
  expansion therefore gives
  \[
    q_t(x,g;0)
    =\e^{\lambda_\circ t}b_\circ(x,g;0)
      +o(\e^{\lambda_\circ t}).
  \]
  On the other hand, Theorem~\ref{thm:pyc} gives
  $q_t(x,g;0)=c_x(g)t^{\nu(\Ccal,g)-1}\e^{\lambda_\circ t}(1+o(1))$, with
  $c_x(g)>0$. Comparing the two expansions yields $\nu(\Ccal,g)=1$ and
  $b_\circ(x,g;0)=c_x(g)>0$. Combining
  \eqref{eq:usg} and \eqref{eq:corridor-exponential-term} proves
  \eqref{eq:uts}.
  
  For all small $\epsilon$, both coefficients are positive and
  $\lambda^*(\epsilon)>\lambda_\circ(\epsilon)$, so equality of the modes has the unique
  solution \eqref{eq:crossover-main}.
  By \eqref{eq:ggo} and \eqref{eq:corridor-eigenvalue-rates},
  \[
    \log\frac{b_\circ(x,g;\epsilon)}{a_*(x,g;\epsilon)}
    =d\log(1/\epsilon)+\log\frac{b_\circ(x,g;0)}{a}+o(1).
  \]
  Since $\Delta(\epsilon)\to\Delta>0$, substitution gives
  \eqref{eq:crossover-loading-correction} and
  $t_\epsilon^c\sim(d/\Delta)\log(1/\epsilon)$. In particular,
  $t_\epsilon^c\to\infty$, so the crossover maturity is positive for all
  sufficiently small $\epsilon$.

  For the crossover-window limit, the defining equality at $t_\epsilon^c$ implies
  that, for every $s\in\RR$,
  \[
    \frac{
      a_*(x,g;\epsilon)\,\e^{\lambda^*(\epsilon)(t_\epsilon^c+s)}
    }{
      b_\circ(x,g;\epsilon)\,\e^{\lambda_\circ(\epsilon)(t_\epsilon^c+s)}
    }
    =
    \e^{[\lambda^*(\epsilon)-\lambda_\circ(\epsilon)]s}.
  \]
  Since
  $\lambda^*(\epsilon)-\lambda_\circ(\epsilon)\to\Delta$, the right-hand side
  converges to $e^{\Delta s}$ uniformly for $s$ in compact sets. Moreover,
  $b_\circ(x,g;\epsilon)\to b_\circ(x,g;0)>0$ and, for all sufficiently small $\epsilon$,
  \(
    |\lambda_\circ(\epsilon)-\lambda_\circ|<\delta/2.
  \)
  Hence, for each fixed $S<\infty$ and all $|s|\leq S$,
  \[
    \frac{|R_{\epsilon,t_\epsilon^c+s}|}
         {b_\circ(x,g;\epsilon)\,\e^{\lambda_\circ(\epsilon)(t_\epsilon^c+s)}}
    \leq
    C_S \e^{-(\delta/2)(t_\epsilon^c-S)}
    \to 0.
  \]
  Dividing \eqref{eq:uts} by the corridor term therefore shows that the full price
  divided by that term converges uniformly to $1+\e^{\Delta s}$. Dividing the
  global-to-corridor ratio by this limit proves \eqref{eq:crossover}.

  For the refinement following the theorem, the eigenvalue bounds in
  \eqref{eq:ggo} and \eqref{eq:corridor-eigenvalue-rates} give
  $\Delta(\epsilon)=\Delta+O(r_L(\epsilon))$. Under
  \eqref{eq:crossover-rate-condition},
  \[
    \left(\frac{1}{\Delta(\epsilon)}-\frac{1}{\Delta}\right)
    \log(1/\epsilon)
    =O\!\left(r_L(\epsilon)\log(1/\epsilon)\right)=o(1).
  \]
  Replacing the denominator in \eqref{eq:crossover-loading-correction} by
  $\Delta$ therefore proves \eqref{eq:crossover-constant-correction}.
\end{proof}

We conclude this subsection with the fixed-recovery specialization used to prove
Corollaries~\ref{cor:rare-disaster-concentration} and
\ref{cor:rare-disaster-crossover}. Retain the notation from
Section~\ref{sec:rare-fixed-recovery}. For every sufficiently small $\epsilon>0$,
let $\phi_\epsilon,\psi_\epsilon\gg0$ be right and left Perron vectors of
$\Lp(\epsilon)$, normalized by $\1^\t\phi_{\epsilon,\mathcal D}=1$ and
$\psi_\epsilon^\t\phi_\epsilon=1$. Under this normalization,
$\pi_\epsilon^{\mathrm{LR}}=\phi_\epsilon\odot\psi_\epsilon$.

\begin{lemma}[Fixed-recovery first-order expansion]
\label{lem:fixed-recovery-expansion}
  Under Assumption~\ref{ass:illustration-near} and $\lambdaD>\lambdaN$, let
  $\bar\pi_{\mathcal N}$ be the stationary distribution of $Q$. Then
  \begin{equation} \label{eq:illustration-physical-mass}
    \pi_\epsilon^P(\mathcal D)
    =
    \epsilon c_P+o(\epsilon),
    \qquad
    c_P
    =
    -\bar\pi_{\mathcal N}^{\t}F U^{-1}\1_{\mathcal D}
    >0.
  \end{equation}
  Moreover,
  \begin{equation}
  \label{eq:illustration-loading-limit}
    \lambda^*(\epsilon)\to\lambdaD,
    \qquad
    \phi_{\epsilon,\mathcal D}\to\phiD,
    \qquad
    \epsilon^{-1}\phi_{\epsilon,\mathcal N}
    \to R(\lambdaD;\LpN)J\phiD\gg0,
  \end{equation}
  while $\psi_\epsilon\to\PsiD$, with $\PsiD$ defined in
  \eqref{eq:illustration-eigenvectors}.
\end{lemma}

\begin{proof}
  Write the physical stationary distribution conformably as
  \[
    (\pi_\epsilon^P)^\t
    =
    \bigl((\pi_{\epsilon,\mathcal N}^P)^\t,
          (\pi_{\epsilon,\mathcal D}^P)^\t\bigr).
  \]
  The disaster block of $(\pi_\epsilon^P)^\t\Lx(\epsilon)=0$ is
  \[
    (\pi_{\epsilon,\mathcal N}^P)^\t F_\epsilon
    +(\pi_{\epsilon,\mathcal D}^P)^\t U
    =0.
  \]
  Since $s(U)<0$, $U$ is invertible and
  \[
    (\pi_{\epsilon,\mathcal D}^P)^\t
    =
    -(\pi_{\epsilon,\mathcal N}^P)^\t F_\epsilon U^{-1}.
  \]
  Hence $\pi_\epsilon^P(\mathcal D)=O(\epsilon)$. The normal block of the stationary
  equation is
  \(
    (\pi_{\epsilon,\mathcal N}^P)^\t(Q-D_\epsilon)
    +(\pi_{\epsilon,\mathcal D}^P)^\t H=0.
  \)
  Since $D_\epsilon=O(\epsilon)$, every limit point of
  $\pi_{\epsilon,\mathcal N}^P$ has unit mass and is stationary for $Q$.
  Irreducibility therefore gives
  $\pi_{\epsilon,\mathcal N}^P\to\bar\pi_{\mathcal N}$. Dividing by $\epsilon$
  and using \eqref{eq:rare-entry-main} yields
  \[
    \epsilon^{-1}(\pi_{\epsilon,\mathcal D}^P)^\t
    \longrightarrow
    -\bar\pi_{\mathcal N}^\t F U^{-1}.
  \]
  Because $-U^{-1}=\int_0^\infty e^{tU}\,\diff t\gg0$, the limiting row vector is strictly positive and \eqref{eq:illustration-physical-mass} follows.
  
  Continuity of the spectral bound gives $\lambda^*(\epsilon)\to\lambdaD$, since
  $\lambdaD>\lambdaN$ is the unique spectral bound of $\Lp(0)$. The right Perron
  equation for \eqref{eq:rare-fixed-recovery-generators} is
  \begin{align}
    (\lambda^*(\epsilon) I_{\mathcal N}-\LpN+D_\epsilon)
    \phi_{\epsilon,\mathcal N}
    &=
    E_\epsilon\phi_{\epsilon,\mathcal D},
    \label{eq:proof-rare-first}\\
    \LpDN\phi_{\epsilon,\mathcal N}
    +\LpD\phi_{\epsilon,\mathcal D}
    &=
    \lambda^*(\epsilon)\phi_{\epsilon,\mathcal D}.
    \label{eq:proof-rare-second}
  \end{align}
  For small $\epsilon$, the inverse in the first equation exists and converges to
  $R(\lambdaD;\LpN)$. Since $E_\epsilon=O(\epsilon)$, we obtain
  $\phi_{\epsilon,\mathcal N}=O(\epsilon)$. By the second equation, every limit
  point of $\phi_{\epsilon,\mathcal D}$ is a right eigenvector of $\LpD$ at
  $\lambdaD$. The normalization $\1^\t\phi_{\epsilon,\mathcal D}=1$ and uniqueness
  of the normalized Perron vector therefore give
  $\phi_{\epsilon,\mathcal D}\to\phiD$. Dividing
  \eqref{eq:proof-rare-first} by $\epsilon$ and using
  $\epsilon^{-1}E_\epsilon\to J$ gives
  \[
    \epsilon^{-1}\phi_{\epsilon,\mathcal N}
    \longrightarrow
    R(\lambdaD;\LpN)J\phiD
    \gg0.
  \]
  
  Since $\lambdaD$ is a simple eigenvalue of $\Lp(0)$, continuity of its rank-one
  Riesz projection implies convergence of the corresponding left eigenspaces. We
  have $\phi_\epsilon\to\PhiD$, and the normalization
  $\psi_\epsilon^\t\phi_\epsilon=1$ therefore selects the limit $\PsiD$. Hence
  $\psi_\epsilon\to\PsiD$, which proves
  \eqref{eq:illustration-loading-limit}.
\end{proof}

\begin{proof}[Proof of Corollary~\ref{cor:rare-disaster-concentration}]
Equation~\eqref{eq:illustration-physical-mass} gives the physical-mass expansion
with $c_P>0$. Define, within this proof,
\(
  \bar\phi_{\mathcal N}
  =R(\lambdaD;\LpN)J\phiD
\)
and
\(
  c_{\mathrm{LR}}=(\PsiD)_{\mathcal N}^\t\bar\phi_{\mathcal N}.
\)
Lemma~\ref{lem:fixed-recovery-expansion} gives
$\bar\phi_{\mathcal N}\gg0$ and $\psi_{\epsilon,\mathcal N}\longrightarrow
(\PsiD)_{\mathcal N}\gg0$, so $c_{\mathrm{LR}}>0$ and
\[
  \epsilon^{-1}\pi_\epsilon^{\mathrm{LR}}(\mathcal N)
  =
  \epsilon^{-1}\psi_{\epsilon,\mathcal N}^\t
  \phi_{\epsilon,\mathcal N}
  \longrightarrow
  (\PsiD)_{\mathcal N}^\t\bar\phi_{\mathcal N}
  =c_{\mathrm{LR}}.
\]
Since $\pi_\epsilon^{\mathrm{LR}}$ has total mass one, its disaster mass is
$1-c_{\mathrm{LR}}\epsilon+o(\epsilon)$, as claimed.
\end{proof}

\begin{proof}[Proof of Corollary~\ref{cor:rare-disaster-crossover}]
Assumption~\ref{ass:illustration-near} satisfies Assumption~\ref{ass:edge} with
$\eta(\epsilon)=\epsilon^{\min\{\eta_F,1\}}$: the normal diagonal changes by
$O(\epsilon)$, the rare entry block opens at order $\epsilon$ with relative
remainder $O(\epsilon^{\eta_F})$, and recovery is a zero-weight edge. Moreover,
$r_L(\epsilon)=O(\epsilon)$ by
\eqref{eq:rare-fixed-recovery-generators}--\eqref{eq:rare-entry-main}.

For $x\in\mathcal N$ and $g=\1$, the rare entry edge gives
$d_{\sss \mathcal N\mathcal D}=1$, and the payoff-side term is zero because
$\mathcal D\in\sS(\1)$. Hence $d(\mathcal N,\1)=1$. The limiting corridor is
$\kK(\mathcal N,\1)=\{\mathcal N\}$ and has rate $\lambda_{\kK}=\lambdaN$.
Since $\LpN$ is irreducible, $\lambdaN$ is algebraically simple and strictly
dominates the real parts of its other eigenvalues. The spectrum of the limiting
block-triangular generator is
$\sigma(\LpN)\cup\sigma(\LpD)$, so
\eqref{eq:rds} makes $\lambdaN$ the unique eigenvalue,
other than $\lambdaD$, with maximal real part. Thus
Assumption~\ref{ass:single-corridor-eigenvalue} holds with
$\lambda^*=\lambdaD$ and $\lambda_{\kK}=\lambda_\circ=\lambdaN$.

Theorem~\ref{thm:crossover} now applies. Equations~\eqref{eq:ggo} and
\eqref{eq:corridor-eigenvalue-rates}, together with $r_L(\epsilon)=O(\epsilon)$,
give the strictly positive coefficients $a_x$ and $b_x$ in the statement and
\[
  a_*(x,\1;\epsilon)=a_x\epsilon(1+o(1)),
  \qquad
  b_\circ(x,\1;\epsilon)=b_x+O(\epsilon),
  \qquad
  \lambda^*(\epsilon)-\lambda_\circ(\epsilon)
  =\lambdaD-\lambdaN+O(\epsilon).
\]
Since $\epsilon\log(1/\epsilon)\to0$, condition
\eqref{eq:crossover-rate-condition} holds. Substituting $d=1$, $a=a_x$,
$b_\circ(x,\1;0)=b_x$, and $\Delta=\lambdaD-\lambdaN$ into
\eqref{eq:crossover-constant-correction} proves
\eqref{eq:rare-disaster-crossover-main}.
\end{proof}

\bibliographystyle{apalike}
\bibliography{main}

\end{document}